\documentclass[12pt,oneside]{amsart}

      \usepackage{amssymb}
      \usepackage{color}
       \usepackage{a4}
\usepackage{graphicx}
\usepackage{amsthm}     
 \newif\ifhyper\IfFileExists{hyperref.sty}{\hypertrue}{\hyperfalse}
\usepackage{amsfonts,amssymb,amsmath,amsthm,bbm,longtable,verbatim}
\ifhyper\usepackage{hyperref}\fi
\usepackage{amsfonts}
\usepackage{amsmath}
\usepackage{euscript}
\usepackage{amsthm}
\usepackage{amssymb}
\usepackage{tabularx}
\usepackage{multirow}
\definecolor{dgreen}{RGB}{40,150,28}

      \DeclareMathOperator{\sign}{sgn}
\newtheorem{theorem}{Theorem}[section]
\newtheorem{lemma}[theorem]{Lemma}
\newtheorem{proposition}[theorem]{Proposition}
\newtheorem{corollary}[theorem]{Corollary}

\newtheorem{remark}[theorem]{Remark}

\theoremstyle{definition}
\newtheorem{definition}[theorem]{Definition}

     \newcommand{\SLE}{{\rm SLE}}

\newcommand{\F}{{_2F_1}}

      \makeatletter
      \def\@setcopyright{}
      \def\serieslogo@{}
      \makeatother
   
\begin{document}

%


   \author{Bertrand Duplantier$^{(1)}$}
      \address{$^{(1)}$Institut de Physique Th\'eorique, Universit\'e Paris-Saclay, CEA, CNRS, F-91191 Gif-sur- Yvette Cedex, France}
\email{bertrand.duplantier@cea.fr}

\author{Xuan Hieu Ho$^{(2)}$}
   \address{$^{(2)}$ MAPMO, Universit\'e d'Orl\'eans, B\^atiment de math\'ematiques, rue de Chartres B.P.6759-F-45067 Orl\'eans Cedex 2, France}
\email{hoxhieu@gmail.com}

 \author{Thanh Binh Le$^{(2)}$}
   \email{mr.lethanhbinh@gmail.com}

   \author{Michel Zinsmeister$^{(2)}$}
  
   \email{zins@univ-orleans.fr}




  \title[Logarithmic coefficients and multifractality of whole-plane SLE]{Logarithmic coefficients and generalized multifractality of whole-plane SLE}


   \begin{abstract}
    It has been shown that for $f$  an instance of the whole-plane $\SLE$ unbounded conformal map from the unit disk $\mathbb D$ to the slit plane, the derivative moments $\mathbb{E}(\vert f'(z) \vert^p)$ can be written in a closed form for certain values of $p$ depending continuously on the SLE parameter $\kappa\in (0,\infty)$. We generalize this property to the mixed moments, $\mathbb{E}\big(\frac{\vert f'(z) \vert^p}{\vert f(z) \vert^q}\big)$, along integrability curves in the moment plane  $(p,q) \in \mathbb R^2$ depending continuously on $\kappa$, by extending the so-called Beliaev--Smirnov equation to this case. The generalization of this integrability property to the $m$-fold transform of $f$ is also given. We define a novel generalized integral means spectrum, $\beta(p,q;\kappa)$, corresponding to the singular behavior of the above mixed moments. \textcolor{black}{By inversion, it allows a unified description of the unbounded interior and bounded exterior versions of whole-plane SLE, and of their $m$-fold generalizations.} The average generalized spectrum of whole-plane $\SLE$ is found to take four possible forms, separated by five phase transition lines in the moment plane $\mathbb R^2$. The average generalized spectrum of the $m$-fold whole-plane SLE is directly obtained from the $m=1$ case by a linear map acting in the moment plane. We also conjecture the precise form of the universal generalized integral means spectrum.
   \end{abstract}   
   \thanks{The first author wishes to thank the Isaac Newton Institute (INI) for
Mathematical Sciences at Cambridge University, where part of this work was completed, for its hospitality and support during the 2015 program ``Random Geometry'', supported by EPSRC Grant Number EP/K032208/1. B.D. also gratefully acknowledges the support of a Simons Foundation fellowship at INI during the Random Geometry program.  B.D. acknowledges financial support from the French Agence Nationale
de la Recherche via the grant ANR-14-CE25-0014 ``GRAAL''; he is also partially funded by the CNRS Projet international de coop\'eration scientifique (PICS) ``Conformal Liouville Quantum Gravity'' n$^{\mathrm o}$PICS06769. The research by the second author is supported by a joint scholarship from MENESR and R\'egion Centre; the third author is supported by a scholarship of the Government of Vietnam. B.D. and M.Z. are partially funded by the CNRS-{\sc insmi} Groupement de Recherche (GDR 3475) ``Analyse Multifractale''. BD and MZ also gratefully acknowledge the continuous hospitality of the Institut Henri Poincar\'e in Paris.}

   \date{\today}


   \maketitle



   \section{Introduction}
   \subsection{Logarithmic coefficients}
   Consider $f$, a holomorphic function in the unit disk $\mathbb D$,
   \begin{equation}\label{fhol}
   f(z)=\sum_{n\ge 0}a_n z^n.
   \end{equation}  Bieberbach observed in 1916 \cite{Bi} that if $f$ is further assumed to be injective, then
   $$\vert a_2 \vert \le 2\vert a_1 \vert,$$
   and he conjectured  that $\vert a_n \vert \le n\vert a_1 \vert$ for all $n > 2$. This famous conjecture has been proved in 1984 by de Branges \cite{dB}. A crucial ingredient of his proof is the theory of growth processes that was developed by Loewner in 1923 \cite{Lo}, precisely  in order to solve the $n=3$ case of the Bieberbach conjecture.\\
 \begin{figure}[tb]
\begin{center}
\includegraphics[angle=90,width=.803290\linewidth]{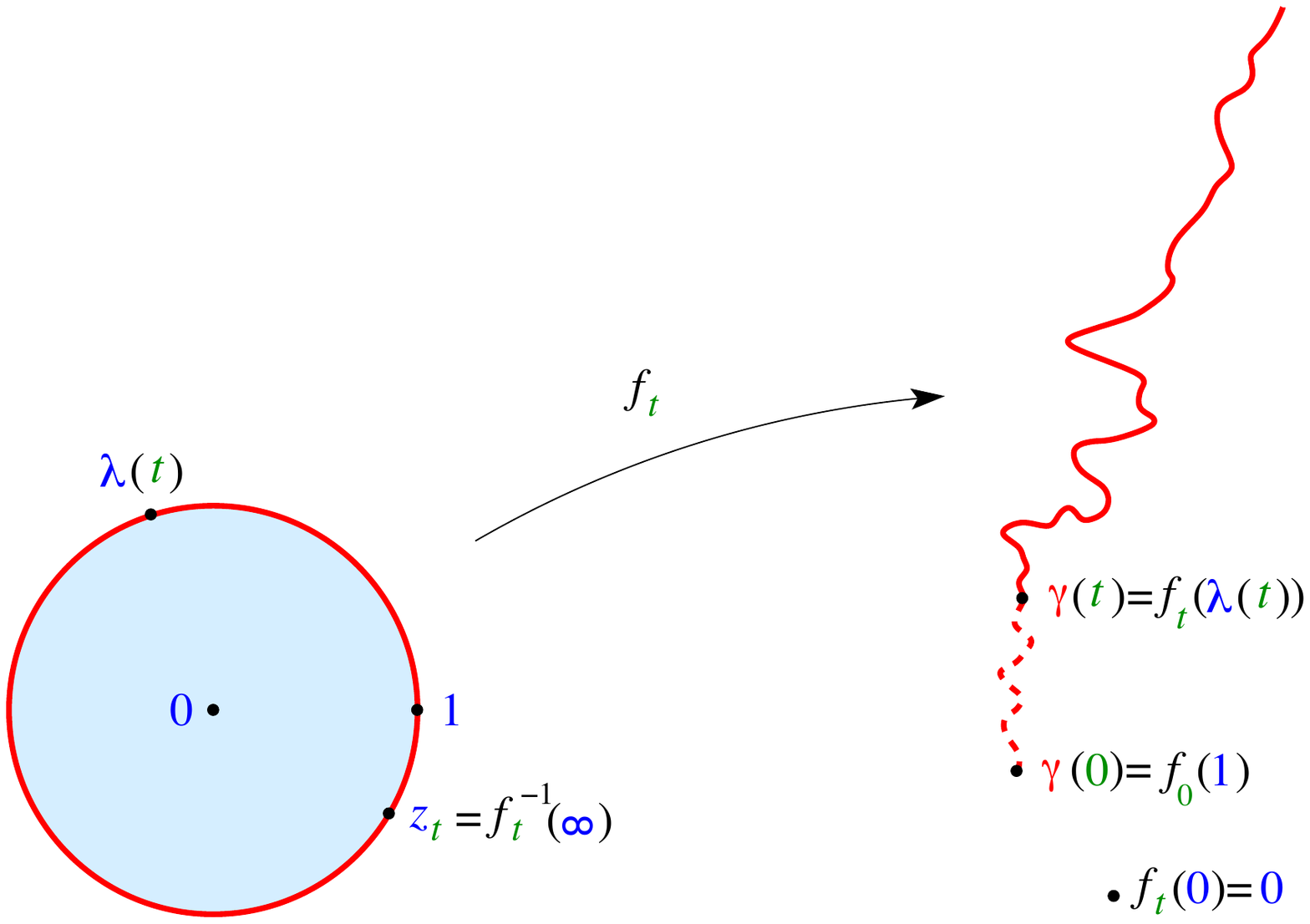}
\caption{{\it Loewner map $z\mapsto f_t(z)$ from $\mathbb D$ to the slit domain $\Omega_t=\mathbb C\backslash \gamma([t,\infty))$ (here slit by a single curve $\gamma([t,\infty))$ for $\SLE_{\kappa\leq 4}$). One has $f_t(0)=0, \forall t\geq 0$. At $t=0$, the driving function $\lambda (0)=1$, so that the image of $z=1$ is at the tip $\gamma(0)=f_0(1)$ of the curve. \textcolor{black}{[From Ref. \cite{DNNZ}].}}}
\label{whpl}
\end{center}
\end{figure}   
      Let $\gamma:[0,\infty) \rightarrow \mathbb C$ be a  simple curve such that $\vert\gamma(t)\vert\rightarrow + \infty$ as $t \rightarrow +\infty$ and such that $\gamma(t) \ne 0, t\ge 0$. Let $\Omega_t = \mathbb C \setminus \gamma([t, \infty))$ and $f_t: \mathbb D = \{ \vert z \vert \le 1\} \rightarrow \Omega_t$ be the Riemann map characterized by $f_t(0)=0, f'_t(0) > 0$ (See Fig. \ref{whpl}).  It is easy to see that $t \mapsto f'_t(0)$ is an  increasing continuous function, diverging   to $+\infty$ as $t \rightarrow +\infty$. Assuming that $f'_0(0) = 1$, and changing parameterization if necessary, we may set $f'_t(0)=e^t, t \ge 0$. Loewner has shown that $f_t$ satisfies the following PDE:
   \begin{equation} \label{loewner}
\frac{\partial}{\partial t} f_t(z) = z \frac{\partial}{\partial z} f_t(z) \frac{\lambda(t)+z}{\lambda(t)-z},
   \end{equation}
   where $\lambda: [0,\infty) \rightarrow \partial \mathbb D$ is a continuous function on the unit circle. 
   With the sole information that $\vert \lambda(t) \vert =1, \forall t$, he could prove that $\vert a_3 \vert \le 3\vert a_1 \vert$. 
    
   Besides Loewner's  theory of growth processes, de Branges' proof also heavily relied on the considereration, developed by Grunsky \cite{Gr} and later Lebedev and Milin \cite{LM}, of {\it logarithmic} coefficients. More precisely, if $f: \mathbb{D} \rightarrow \mathbb{C}$ is holomorphic and injective with $f(0) = 0$, we may consider the power series, 
   \begin{equation}\label{ptlog}
   \log\frac{f(z)}{z} = 2\sum\limits_{n\ge{1}}\gamma_nz^n.
   \end{equation}
   
   The purpose of introducing this logarithm was to prove Robertson's conjecture \cite{R}, which was known to imply Bieberbach's. Let $f$ be in the $\mathcal S$ class of schlicht functions, i.e., holomorphic and injective in the unit disk, and normalized such that $f(0)=0,\, f'(0)=1$. There is a branch $f^{[2]}$ of $z\mapsto\sqrt{f(z^2)}$ which is an odd function in $\mathcal S$ and plays a crucial role in the theory of univalent functions. Let us then write
   \begin{equation} \label{def::h}f^{[2]}(z):=z\sqrt{f(z^2)/z^2}=\sum_{n=0}^{\infty}b_{2n+1}z^{2n+1},
   \end{equation}
  with $b_1=1$.  Robertson's conjecture states that:
   \begin{equation}\label{robertson} \forall n\geq 0,\,\sum_{k=0}^{n}\vert b_{2k+1}\vert^2\leq n+1.
   \end{equation}
  The Lebedev and Milin approach to this conjecture consisted in observing that
   $$ \log \frac{f^{[2]}(\sqrt{z})}{\sqrt{z}}=\frac 12 \log\frac{f(z)}{z},$$
   and consequently that
   $$\sum_{n=0}^{\infty}b_{2n+1}z^n=\exp\left({\sum_{n=1}^{\infty}\gamma_n z^n}\right).$$
  They proved what is now called the second Lebedev-Milin inequality, a combinatorial inequality connecting the coefficients of any power series to those of its exponential, namely
   \begin{equation} \label{LM3}
\forall n\geq 0,\;\sum_{k=0}^{n}\vert b_{2k+1} \vert^2\leq (n+1)\exp{\left(\frac{1}{n+1}\sum_{m=1}^n\sum_{k=1}^{m}\left(k\vert\gamma_k\vert^2-\frac{1}{k}\right)\right)}.
\end{equation}
   This naturally led Milin \cite{M} to conjecture that
   \begin{equation} \label{MiCo}
   \forall f\in \mathcal S,\,\forall n\geq 1,\;\sum_{m=1}^n\sum_{k=1}^m \left(k\vert \gamma_k\vert^2-\frac 1k\right)\leq 0;
   \end{equation}
     this conjecture, proved by de Branges in 1984, implies Robertson's, hence Bieberbach's conjecture. 
    
   Returning to Loewner's theory, his derivation of Eq. \eqref{loewner}  above is only half of the story. There is indeed a converse: given  any continuous function $\lambda: [0,+\infty[ \rightarrow \mathbb{C}$ with $\vert\lambda(t)\vert = 1$ for $t \geq 0$, the the Loewner equation \eqref{loewner}, supplemented by the boundary (``initial'') condition, $\lim_{t\to +\infty}f_t(e^{-t}z)=z$,    has a solution $(t, z) \mapsto f_t(z)$, such that  $(f_t(z))_{t\geq 0}$ is a chain of Riemann maps onto simply connected domains $(\Omega_t)$ that are  increasing with $t$. 
   
   In 1999, Schramm  \cite{Schr} introduced into the Loewner equation the {\it random} driving function,
  \begin{equation}\label{lambda}
   \lambda(t):=\sqrt{\kappa}B_t,
   \end{equation}
   where $B_t$ is  standard one dimensional Brownian motion and $\kappa$  a non-negative parameter, thereby making Eq. \eqref{loewner}   a stochastic PDE, and creating the celebrated {\it Schramm-Loewner Evolution}  $\SLE_\kappa$.
   
    The associated conformal maps $f_t$ from $\mathbb D$ to $\mathbb C\setminus \gamma([t,\infty)$,  obeying  \eqref{loewner} for \eqref{lambda}, define the {\it interior whole-plane} Schramm-Loewner evolution. Their  coefficients $a_n(t)$, which are random variables, are defined by a normalized series expansion, as described in the following proposition \cite{DNNZ}. 
\begin{proposition} \label{prop1} Let $(f_t(z))_{t\geq 0}, z\in \mathbb D,$ be the interior Schramm--Loewner whole-plane process driven by  $\lambda(t)=e^{i\sqrt{\kappa}B_t}$ in Eq. \eqref{loewner}. We write
\vskip -.2cm
\begin{equation}\label{ftexp}
f_t(z)=e^t\big(z+\sum_{n\geq 2}a_n(t) z^n\big).
\end{equation}
\vskip -.2cm and for its  
   logarithm,
 \begin{equation}\label{logftexp}
   \log\frac{e^{-t}f_t(z)}{z} = 2\sum\limits_{n\ge{1}}\gamma_n(t)z^n.
   \end{equation} 
   Then the {\it conjugate} whole-plane Schramm--Loewner evolution $e^{-i\sqrt{\kappa}B_t} f_t\big(e^{i\sqrt{\kappa}B_t}z\big)$ has the {\it same law} as $f_0(z)$, hence $e^{i(n-1)\sqrt{\kappa}B_t}a_n(t)\stackrel{\rm (law)}{=}a_n(0)$. From this and Eqs. \eqref{ftexp}, \eqref{logftexp}, follows the identity $e^{in\sqrt{\kappa}B_t}\gamma_n(t)\stackrel{\rm (law)}{=}\gamma_n(0)$. 
   In the sequel, we set: $a_n:=a_n(0)$ and $\gamma_{n}:=\gamma_n(0)$.
\end{proposition}
	   The starting point of the present article is the observation, made in Ref. \cite{HAL-DNNZ}, that the $\SLE_\kappa$ process, in its {\it interior} whole-plane version, has a rich algebraic structure, giving rise to a host of (integrability-like) closed form results. The first hint was the fact that, beyond the coefficient expectations $\mathbb E(a_n)$ for Eq. \eqref{ftexp}, the coefficient squared moments, $\mathbb E(|a_n|^2)$, have very simple  expressions for specific values of $\kappa$. This has been developed in detail in Refs. \cite{DNNZ} and \cite{IL} (see also Refs. \cite{2012arXiv1203.2756L,2013JSMTE..04..007L,1751-8121-47-16-165202}), by using a PDE obeyed by the derivative moments $\mathbb E(|f'(z)|^p)$. Following the work of Rohde--Schramm \cite{MR2153402}, it was originally derived by  Beliaev--Smirnov in Ref. \cite{BS} to study the average integral means spectrum of the  \textcolor{black}{\emph{exterior}} version of the whole-plane $\SLE_\kappa$ map. Note also that similar ideas already appeared in Ref. \cite{2010JSP...139..108K}, where A. Kemppainen  studied in detail the coefficients associated with the Schramm--Loewner evolution, using a stationarity property of SLE \cite{2006JPhA...39L.657K}. However, the  focus there was on expectations of the moments of those coefficients, rather than on the moments of their moduli. 
    
      \textcolor{black}{Here, we study in particular the logarithmic coefficients \eqref{logftexp} of whole-plane SLE and the generalizations thereof.  The main idea of this work is the introduction for the (unbounded) whole-plane SLE map $f$ of  {\it mixed moments}, $\mathbb E(|f'(z)|^p/|f(z)|^q)$, for any $(p,q)\in \mathbb R^2$, which are found to obey a two-parameter family of Beliaev--Smirnov-type equations. We also define and study in detail the \emph{generalized integral means spectrum} associated with these mixed moments, which is a priori a function of $p$ and $q$.}
      
      \textcolor{black}{We argue that the general approach proposed here is the natural one for whole-plane SLE.  It unifies the earlier SLE integral means  studies of Refs. \cite{BDZ,BS,DNNZ,IL,2012arXiv1203.2756L,2013JSMTE..04..007L} into a much broader framework, which in particular allows one to fully exploit the \emph{inversion symmetry} between the unbounded  inner and bounded outer versions, and also to cover  the  $p=q$ logarithmic case, as well as the integral means of the map itself or of its $m$-fold transforms.}  
      
      \textcolor{black}{As we shall see, the $(p,q)$-plane for whole-plane SLE is structured by \emph{integrable probability lines}, as well as by \emph{phase-transition lines}.  The integrability lines yield a continuum of closed forms for mixed moments, that generalize results in Refs. \cite{DNNZ,IL,2012arXiv1203.2756L,2013JSMTE..04..007L}. The phase transition lines mark the breakdown of the standard analysis of B--S-type equations initiated in Ref. \cite{BS}, as already seen in the case of  the standard integral means spectrum in Refs. \cite{BDZ,DNNZ,IL,2012arXiv1203.2756L,2013JSMTE..04..007L}. In particular, a certain phase transition \emph{point} here manifests the appearance, for the bounded exterior whole-plane SLE,  of a subjacent spectrum, thus  requiring a novel proof as given in Ref. \cite{BDZ}. To study the phase structure of the whole-plane generalized spectrum, we unify and extend here the non-standard methods initiated in Refs. \cite{BDZ,DNNZ}. Moreover, this field leads to further interesting open questions, which apparently exceed the power of the various methods proposed so far.} 
       
  \subsection{Main results}
   A first motivation of this article is the proof, originally obtained for small $n$ by the third author \cite{LTB},  of the following. 
  \textcolor{black}{\begin{theorem}\label{logarithmic theorem}
   Let $f(z):=f_0(z)$ be the time $0$ unbounded whole-plane $\SLE_\kappa$ map, in the same setting as in Proposition \ref{prop1}, such that
   \begin{equation*}
   \log \frac{f(z)}{z}=2\sum_{n\geq 1}\gamma_n z^n;
   \end{equation*}
   then, for $\kappa=2$, which corresponds to the scaling limit of the Loop-Erased Random Walk \cite{lawler2004,Schr},
    \begin{align*}
 &\mathbb E(\gamma_1)=-1/2,\,\,\,\mathbb E(\gamma_{n})=0,\,\,\,n\geq 2,\\
& \mathbb{E}(\vert \gamma_n \vert^2)= \frac{1}{2n^2},\,\,\,n\geq 1,\\
& \mathbb E(\gamma_n\bar{\gamma}_{n+1})=-\frac{1}{4n(n+1)},\,\,\,
 \mathbb E(\gamma_n\bar{\gamma}_{n+k})=0,\,\,\,n\geq 1,k\geq 2.
\end{align*}
   \end{theorem}}
        Let us then briefly return to the Lebedev-Milin theory. 
   By Theorem \ref{logarithmic theorem}, we have for $\SLE_{2}$,
   $$\mathbb E\left(\sum_{m=1}^{n}\sum_{k=1}^{m}\left(k\vert\gamma_k\vert^2-\frac 1k\right)\right)=-\frac{1}{2}\sum_{m=1}^{n}\sum_{k=1}^{m}\frac{1}{k}=-\frac{n+1}{2}\sum_{k=2}^{n+1}\frac{1}{k},$$
   which gives an example of the validity ``in expectation"  of the Milin conjecture. 
   Recalling Definition \eqref{def::h}, 
   we also get, in expectation, a check of Robertson's conjecture \eqref{robertson}:
 $$ \mathbb E\left(\log{\frac{1}{n+1}\sum_{k=0}^n\vert b_{2k+1}\vert^2}\right)\leq -\frac{1}{2}\sum_{k=2}^{n+1}\frac{1}{k}.$$
 
  The idea behind the proof of Theorem \ref{logarithmic theorem} is to differentiate \eqref{ptlog},
   $$z\frac{d}{dz}\log \frac{f(z)}{z} =z \frac{f'(z)}{f(z)}- 1=2\sum_{n\geq 1}n\gamma_{n}z^n,$$
   such that
   \begin{equation}\label{dev2}
 {\left|z \frac{f'(z)}{f(z)} \right|}^{2}=1 +2\sum_{n\geq 1}n\gamma_{n}(z^n+\bar{z}^{n})
 +4\sum_{n\geq 1}\sum_{m\geq 1}nm\gamma_{n}\bar{\gamma}_{m}z^n\bar{z}^m,
 \end{equation}
   and  to compute $\mathbb{E} \big( {\left|z \frac{f'(z)}{f(z)} \right|}^2 \big)$. We indeed prove: 
   \begin{theorem}\label{theorem kappa2}
   Let $f$ be the interior whole-plane $\SLE_\kappa$ map, in the same setting as in Theorem \ref{logarithmic theorem}; then for $\kappa=2$,
   $$\mathbb{E} \bigg( {\left|z \frac{f'(z)}{f(z)} \right|}^{2} \bigg) =\frac{(1-z)(1-\bar{z})}{1-z\bar{z}}.$$
   \end{theorem}
This gives
 $$
 \mathbb{E} \bigg({\left|z \frac{f'(z)}{f(z)} \right|}^{2} \bigg) 
 =1-\sum_{n\geq 0}\big(z^{n+1}\bar{z}^{n}+z^{n}\bar{z}^{n+1}\big)+2\sum_{n\geq 1}z^{n}\bar{z}^n,
 $$
and  Theorem \ref{logarithmic theorem} follows by direct identification to Eq. \eqref{dev2}. 

Theorem \ref{theorem kappa2}
is actually a consequence of Theorems \ref{maintheorem1/2} and \ref{maintheorem} of Sections \ref{sec3} and \ref{sec4} below, which give  expressions in closed form for the mixed moments, 
   \begin{equation}\label{mixedmom}
  (a)\,\,\,\,\, \mathbb{E}\bigg(\frac{\left(f'(z)\right)^{p/2}}{\left(f(z)\right)^{q/2}}\bigg);\,\,\,\,(b)\,\,\,\mathbb{E}\bigg({\frac{\vert f'(z) \vert^p}{\vert f(z)\vert^q}}\bigg),
   \end{equation}
    along an integrability curve $\mathcal R$, which is a {\it parabola} in the $(p,q)$ plane depending on the SLE parameter $\kappa$. In fact, we establish a general integrability result along the $\mathcal R$  parabola for the SLE {\it two-point function}:
$$ G(z_1,\bar z_2):=\mathbb{E}\bigg(z_1^{\frac{q}{2}} \frac{(f'(z_1))^\frac{p}{2}}{(f(z_1))^\frac{q}{2}}  \overline{\left[z_2^{\frac{q}{2}}\frac{(f'(z_2))^\frac{p}{2}}{(f(z_2))^\frac{q}{2}}\right]}\bigg). 
$$
The mixed moments \eqref{mixedmom} can also be seen respectively as the value $G(z,0)$  of this SLE two-point function at $(z_1=z, z_2\to 0)$ for $(a)$, and the value $G(z,\bar z)$ at coinciding points,  $z=z_1=z_2$, for $(b)$. These integrability theorems, which provide  full generalizations of the results of Refs. \cite{DNNZ} and \cite{IL},  give rise to a host of new algebraic identities concerning the whole-plane $\SLE_\kappa$ random map.
   
These  integrability results are generalized in Section \ref{mfold} to the so-called $m$-fold symmetric transforms of the whole-plane SLE map \cite{DNNZ}, $f^{[m]}: z\mapsto \sqrt[m]{f(z^m)}, m\in \mathbb N\setminus \{0\}$,  defined for $f \in \mathcal S$ as the holomorphic branch  whose derivative is equal to $1$ at $0$. These are functions in $\mathcal S$ whose Taylor series are of the form $f(z)=\sum_{k \ge 0}a_{mk+1}z^{mk+1}$, the $m=2$ case  corresponding to odd functions.
 To extend  the definition to negative $m$,  the $m$-fold transform of the outer map  is conjugate by inversion of the $(-m)$-fold transform of the inner map, $f^{[m]}(\zeta)=1/f^{[-m]}(1/\zeta)$, for $m\in \mathbb Z\setminus\mathbb N$ and $\zeta\in \mathbb D_-:=\mathbb C\setminus \overline{\mathbb D}$.  	In particular, the map, $\zeta\in {\mathbb D_-} \mapsto f^{[-1]}(\zeta)=1/f(1/\zeta),$  is just the exterior whole-plane map considered by Beliaev and Smirnov in Ref. \cite{BS}, from ${\mathbb D_-}$ to the slit plane, a domain with bounded boundary. 
 \textcolor{black}{Note similarly that for all $m<0$,  $f^{[m]}(\mathbb D_-)$ has bounded boundary}. Interestingly enough, a linear map in the $(p,q)$-moment plane allows one to directly relate the generalized mixed moments of the $m$-fold map to those of $f$, thus yielding the integrability Theorem \ref{mcase}. 
 
 
  From the consideration in Eq. \eqref{mixedmom}  of the mixed moments $(b)$ of moduli of $f'$ and $f$  for whole-plane SLE, we are led to introduce the following definition.
\textcolor{black}{\begin{definition}\label{def:gims} For a function $f\in \mathcal S$, the {\it (average) generalized integral means spectrum} $\beta_f(p,q)$, depending on $p$ and $q$, is defined as  
\begin{equation*}
\beta_f(p,q):=\limsup_{r\to 1^{-}} \frac{\log \int_{r\partial\mathbb D} \mathbb E\left(\frac{\vert f'(z)\vert^p}{\vert f(z)\vert^q}\right)|dz|}{\log\left(\frac{1}{1-r}\right)}.
\end{equation*}
\end{definition}
In the case where the limit exists, this can be written as,
\begin{equation*}
\int_{r\partial\mathbb D} \mathbb E\left(\frac{\vert f'(z)\vert^p}{\vert f(z)\vert^q}\right)|dz|\stackrel{(r\to 1^{-})}{\asymp} \left(1-r\right)^{-\beta_f(p,q)},
\end{equation*}
in the sense of the equivalence of the logarithms of both terms.
\begin{remark} The $q=0$ case yields the \emph{standard}  spectrum $\beta_f(p):=\beta_f(p,0)$, which is defined for any univalent function. For a map $f \in \mathcal S$, the introduction of the $q$-parameter is aimed at the analysis of the behavior of $f$ at infinity. 
 Indeed, when $f$ is bounded, the generalized integral means spectrum does not depend on $q$ and coincides with the standard one. 
\end{remark}
\begin{remark}\label{qq'duality} {\it Exterior-Interior Duality}. For $\hat f:=f^{[-1]}$ and $0<r<1$, we  have 
\begin{equation}\label{qq'}
\int_{r^{-1}\partial\mathbb D} \mathbb E\left(\frac{\vert \hat f'(\zeta)\vert^p}{\vert\hat  f(\zeta)\vert^{q}}\right)|d\zeta|=r^{2p-2}\int_{r\partial\mathbb D} \mathbb E\left(\frac{\vert f'(z)\vert^p}{\vert f(z)\vert^{q'}}\right)|dz|, 
\end{equation} 
 for $q+q'=2p$. In particular, the $(p,0)$ {\rm standard} integral means for the exterior  \cite{BDZ,BS} or interior \cite{DNNZ,2012arXiv1203.2756L,2013JSMTE..04..007L} whole-plane maps correspond to the $(p,2p)$ {\rm generalized} integral means of their respective inverted maps. 
\end{remark}}
\textcolor{black}{The introduction of this generalized spectrum in the whole  $(p,q)$-plane also allows for a unified description of the \emph{standard}  spectra $\beta_{f^{[m]}}(p)$ of the whole collection of $m$-fold 
transforms of a given map $f$, in terms of its generalized spectrum $\beta_f(p,q)$. We indeed have the identity (Section \ref{mfoldspectrum}):
\begin{equation}\label{betamp0}
\beta_{f^{[m]}}(p)=\beta_f\left(p, (1-1/m)p\right),\,\,\,m\in \mathbb Z\setminus\{0\}.
\end{equation}
In particular, for $m=-1$,  the ray $q=2p$ describes the standard integral means spectrum in the exterior (bounded) case, as studied in Refs. \cite {BDZ,BS}.}

In Section \ref{secintmean}, we  study this generalized spectrum, $\beta(p,q;\kappa)$, for the interior whole-plane $\SLE_\kappa$ and  $(p,q)\in \mathbb R^2$. We show that it takes four possible forms, $\beta_0(p)$, $\beta_{\mathrm{tip}}(p)$, $\beta_{\mathrm{lin}}(p)$ and $\beta_1(p,q)$. The first three spectra are independent of $q$, and are respectively given by the bulk, the tip and the linear SLE average spectra appearing in the work by Beliaev and Smirnov \cite{BS}, recently revisited \textcolor{black}{and partially corrected}  in Ref. \cite{BDZ} by Beliaev and two of the present authors.  The bulk case corresponds to the earlier harmonic measure multifractal spectrum derived by the first author in Refs. \cite{2000PhRvL..84.1363D,MR2112128} (see also \cite{BDone,PhysRevLett.89.264101,2008NuPhB.802..494D} and \cite{1751-8121-41-28-285006,PhysRevLett.95.170602,2007JPhA...40.2165R}), and recently established in an almost sure sense by Gwynne, Miller and Sun \cite{2014arXiv1412.8764G};  the tip harmonic measure spectrum was predicted by Hastings \cite{PhysRevLett.88.055506} and proved in an almost sure sense by Johansson Viklund and Lawler \cite{0911.3983}.  The  fourth spectrum, $\beta_1(p,q)$, is the extension to non-vanishing $q$ of a novel integral means spectrum, which was discovered and studied in Refs. \cite{DNNZ} and \cite{IL,2013JSMTE..04..007L}, and which is due to the unboundedness of whole-plane SLE. As shown in Ref. \cite{DNNZ}, this spectrum is also closely related to the (non-standard) SLE tip exponents obtained by quantum gravity techniques in Ref. \cite{MR2112128}, and  to the so-called radial SLE derivative exponents of Ref. \cite{MR2002m:60159b}. \textcolor{black}{As shown in Ref. \cite{BDZ} for the exterior bounded case, the corresponding spectrum, $\beta_1(p,2p)$, due to the SLE seed or `second tip', is subjacent in between the bulk and the tip spectra.} 

\emph{Five different phase transition lines  partition the $(p,q)$-plane into four domains, where the generalized integral means spectrum  takes four different forms, as given by  the following theorem.}
\textcolor{black}{\begin{theorem}\label{propseparbis}
 Define the functions:
 \begin{align*} \nonumber
\beta_{\mathrm{tip}}(p;\kappa):=& -p-1+\frac 14(4+\kappa-\sqrt{(4+\kappa)^2-8\kappa p}),\\  \beta_0(p;\kappa):=&-p+\frac{4+\kappa}{4\kappa}(4+\kappa-\sqrt{(4+\kappa)^2-8\kappa p}),\\ 
\beta_{\mathrm{lin}}(p;\kappa):=& \,\,p-\frac{(4+\kappa)^2}{16\kappa},\\ 
\beta_1(p,q;\kappa):=&3p-2q-\frac{1}{2}-\frac{1}{2}\sqrt{1+2\kappa(p-q)}.
\end{align*}
Consider the `green' parabola $\mathcal G$ defined in the $(p,q)$-plane by the parametric equations,  
\begin{equation*}\begin{split}
&p=p_{\mathcal G}(\gamma):=\frac{(4+\kappa)^2}{8\kappa}-\frac{\kappa}{2}\gamma^2,
\\ 
&q=q_{\mathcal G}(\gamma):=\frac{(4+\kappa)^2}{8\kappa}+\gamma-\kappa\gamma^2,\,\,\gamma\in \mathbb R,
\end{split}\end{equation*}
and the `blue' quartic $\mathcal Q$,
\begin{equation*}\begin{split} 
&p=p_{\mathcal Q}(\gamma):=\frac{\kappa}{16}+\left(1+\frac{\kappa}{4}\right)\gamma-\frac{\kappa}{2}\gamma^2-\frac{1}{8}\Delta^{\frac{1}{2}}(\gamma),\\ 
&q=q_{\mathcal Q}(\gamma):=p_{\mathcal Q}(\gamma)+\gamma-\frac{\kappa}{2}\gamma^2,\\
&\Delta(\gamma):=4\kappa^2\gamma^2-2\kappa(4+\kappa)\gamma+\frac{1}{4}(8+\kappa)^2+4\kappa,\,\,\gamma \in \mathbb R.
\end{split}\end{equation*}
 Phase transition lines for the generalized integral means spectrum $\beta(p,q)$ of whole-plane $\SLE_\kappa$ are  in the $(p,q)$-plane (Fig. \ref{separ}):
 \begin{itemize}
 \item (i)  the vertical  half-line $D_0$ above $P_0=(p_0,q_0)$, with\\ $p_0=3(4+\kappa)^2/32\kappa, q_0=(4+\kappa)(8+\kappa)/16\kappa$, where $\beta_0(p;\kappa)=\beta_{\mathrm{lin}}(p;\kappa)$;
  \item (ii)  the unit slope half-line $D_1$ originating at $P_0$,  whose equation is\\ $q-p=(16-\kappa^2)/32\kappa$ with $p\geq p_0$, and where $\beta_{\mathrm{lin}}(p;\kappa)=\beta_1(p,q;\kappa)$;
   \item (iii) the section of green parabola $\mathcal G$, with  parametric coordinates\\   $\big(p_{\mathcal G}(\gamma),q_{\mathcal G}(\gamma)\big)$  for $\gamma\in[1/4+1/\kappa,1+2/\kappa]$, between $P_0$ and $Q_0=(p'_0,q'_0)$, with $p'_0=-1-3\kappa/8, q'_0=-2-7\kappa/8$, and where $\beta_0(p;\kappa)=\beta_1(p,q;\kappa)$;
    \item (iv) the vertical half-line  $D'_0$ above point $Q_0$, where $\beta_{\mathrm{tip}}(p;\kappa)=\beta_0(p;\kappa)$;
      \item (v) the  branch of the blue quartic $\mathcal Q$ from  $Q_0$ to $\infty$, with parametric coordinates  $\big(p_{\mathcal Q}(\gamma),q_{\mathcal Q}(\gamma)\big)$  for $\gamma\in[1+{2}/{\kappa},+\infty)$, where $\beta_{\mathrm{tip}}(p;\kappa)=\beta_1(p,q;\kappa)$.
  \end{itemize}
  The analytic form of $\beta(p,q;\kappa)$ changes across lines (i), (ii), (iv) and (v) according to Figure  \ref{separ}; in case (ii), $\beta(p,q;\kappa)=\beta_{\mathrm{lin}}(p;\kappa)$ above $D_1$, whereas $\beta(p,q;\kappa)\leq \beta_1(p,q;\kappa)$ below $D_1$.
 \end{theorem} 
   \begin{figure}[h!]
\begin{center}
\includegraphics[angle=0,width=.493290\linewidth]{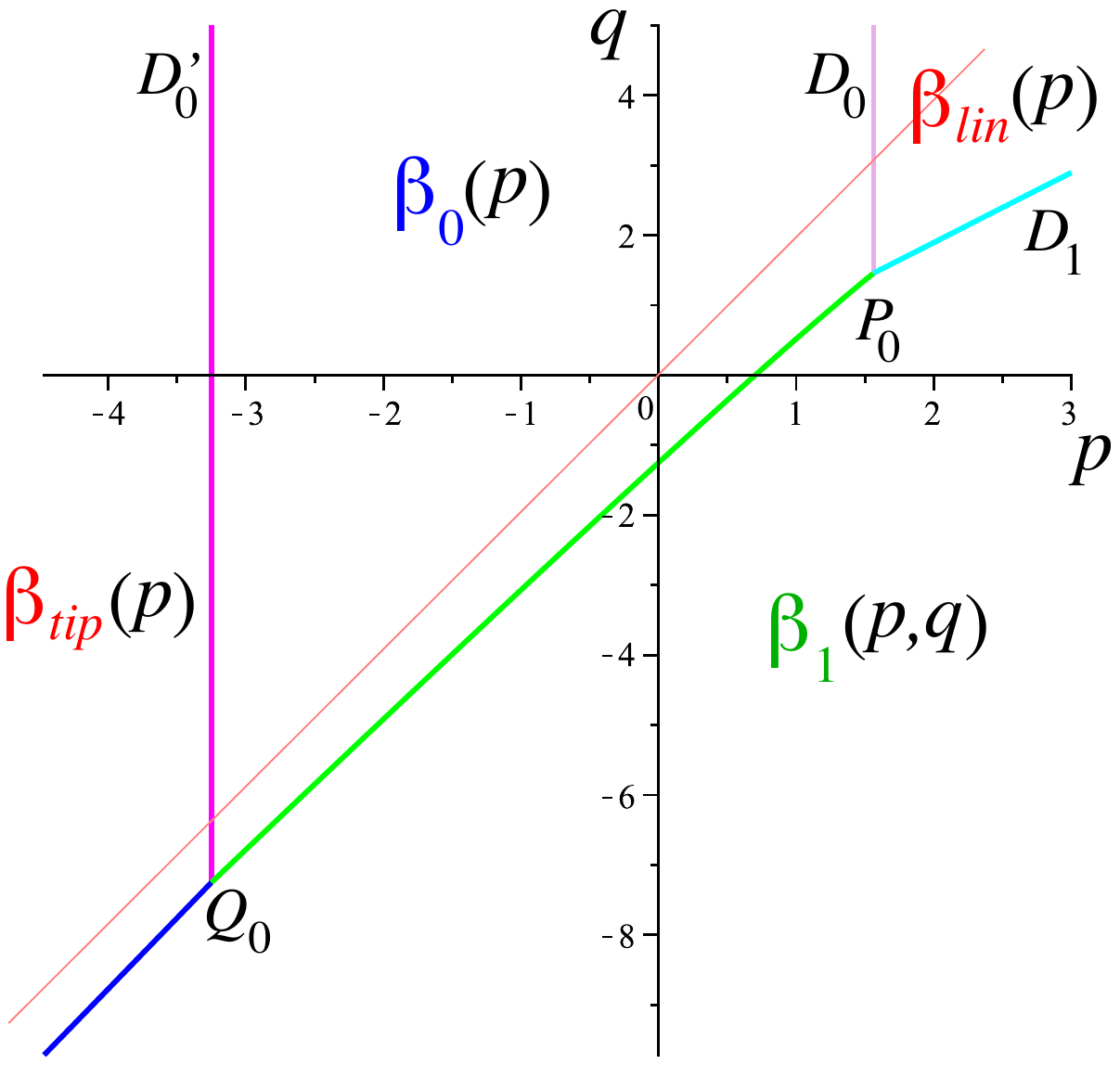}
\includegraphics[angle=0,width=.493290\linewidth]{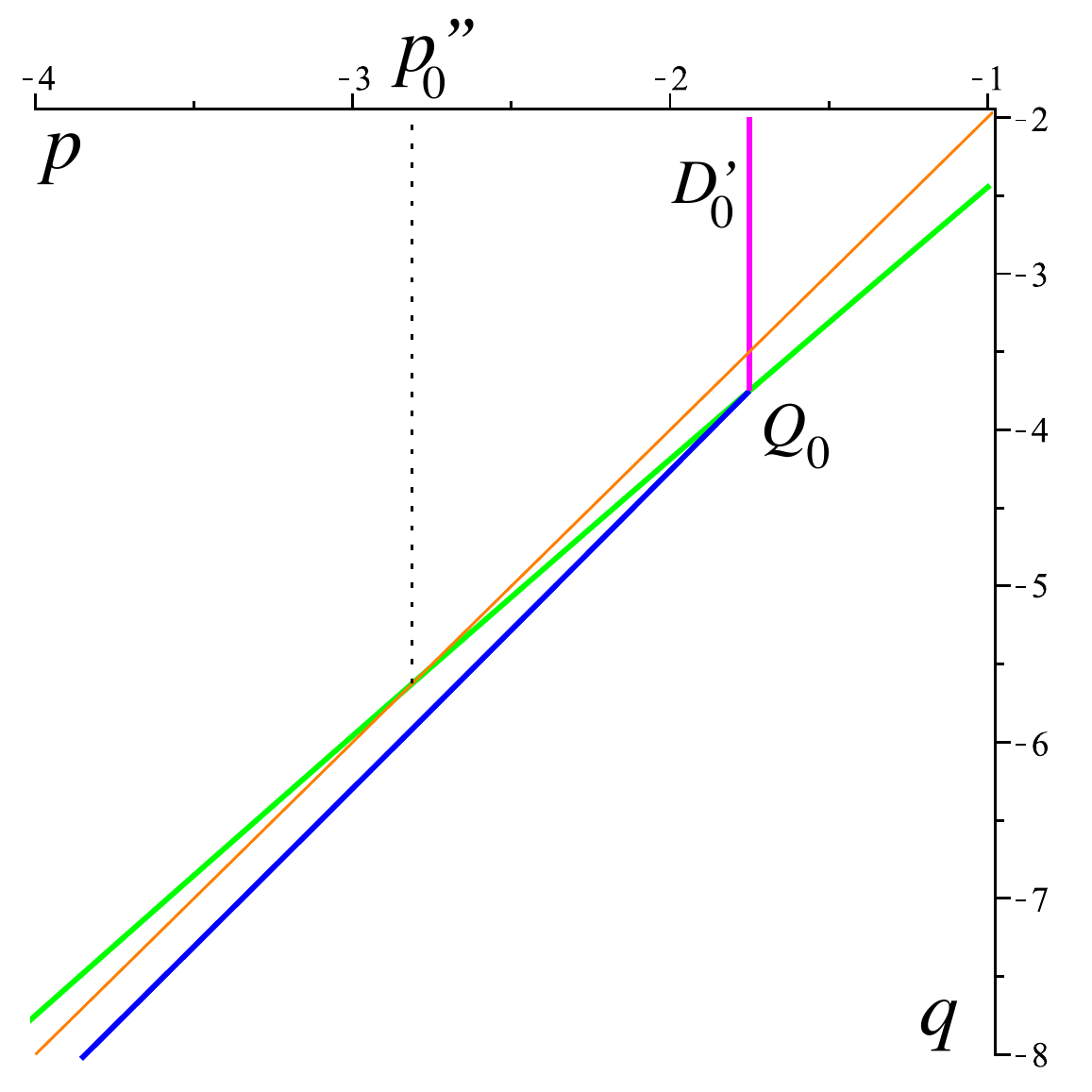}
\caption{{\it Left: Phase diagram for the SLE generalized integral means spectrum $\beta(p,q)$ illustrating Theorem \ref{propseparbis}. The $q=2p$ (coral) straight  line yields the standard spectrum of the exterior version of whole-plane SLE, as  studied in Refs. \cite{BDZ,BS}. Right:  [Zoom below $Q_0$] This line, being asymptotically parallel to the blue quartic $\mathcal Q$, does not intersect the lower domain where the form $\beta_1$ holds, but intersects the green parabola $\mathcal G$  below point $Q_0$ at abscissa $p_0''$.}} 
\label{separ}
\end{center}
\end{figure} 
\begin{remark}The proof of Theorem \ref{propseparbis}, as given in Section \ref{proofs}, establishes the analytic form of $\beta(p,q;\kappa)$ in an upper semi-infinite domain, whose interior strictly includes $D_0$, $D'_0$ and the branches of quartic and parabola of the theorem's statement, and whose frontier includes $D_1$.
\end{remark}
\begin{remark}
The fact that the line $q=2p$ intersects the extension of the green parabola $\mathcal G$ below $Q_0$ at $p_0''=-\frac{1}{128}(4+\kappa)^2(8+\kappa)$ (Fig. \ref{separ}, right) signals the presence of an underlying phase transition from the bulk spectrum to the $\beta_1$ spectrum. As recently shown in Ref. \cite{BDZ}, the original proof of Ref. \cite{BS} breaks down there and requires for $p\leq p_0''$ a novel argument. 
\end{remark}}
    The generalization of this four-domain structure to the generalized integral means spectrum of $m$-fold transforms, $f^{[m]},\,m\in \mathbb N\setminus\{0\}$, is obtained in a straightforward way from the above mentioned linear map in $(p,q)$ co-ordinates (Theorem \ref{propseparmlightbis}). The structure obtained appears so robust that the {\it universal} generalized spectrum $B(p,q)$, i.e., the maximum of $\beta_f(p,q)$ over all univalent functions $f\in \mathcal S$, shows a similar partition of the mixed moment plane.  \textcolor{black}{In Section  \ref{univ},    we give a precise result for its four forms, which incorporates known results and conjectures on the standard universal spectra for univalent, not necessarily bounded, functions \cite{9780511546617,Makanaliz,Pommerenke}.}
      
    \subsection{Synopsis} This article is organized as follows. Section \ref{sec2} deals with logarithmic derivatives. It sets up the martingale techniques needed for dealing with mixed moments. Section \ref{sec3} uses them for the study of the  complex one-point function $(a)$ in \eqref{mixedmom}, which is shown to obey a simple differential equation in complex variable $z$. This leads to Theorem \ref{maintheorem1/2}, which establishes a closed form for this function along the integrability  parabola $\mathcal R$ in the $(p,q)$-plane.  Section \ref{sec4} is concerned with the moduli one-point function $(b)$ in \eqref{mixedmom}, and more generally, with the SLE two-point function $G(z_1,\bar z_2)$. 
A PDE in $(z_1,\bar z_2)$ is derived for $G(z_1,\bar z_2)$, which yields a proof of Theorem \ref{maintheorem} establishing closed form expressions for $G$  for all $(p,q)\in \mathcal R$. Section \ref{mfold} deals with the generalization of the previous integrability results to the $m$-fold symmetric transforms $f^{[m]}, m\in \mathbb Z\setminus \{0\}$, of the whole-plane SLE map $f$.  Section \ref{secintmean} is devoted to the study of the averaged generalized spectrum $\beta(p,q;\kappa)$ of the whole-plane $\SLE_\kappa$ random map $f$, as well as to the averaged generalized integral means spectrum $\beta^{[m]}(p,q;\kappa)$ of the $m$-fold transform $f^{[m]}$ for $m\in \mathbb N\setminus \{0\}$. Of particular interest are the five phase transition lines separating the four different analytic expressions of $\beta$ (or $\beta^{[m]}$) in the moment plane, given in \textcolor{black}{Theorems \ref{propseparbis} and \ref{propseparmlightbis}}.  
In the final Section \ref{univ}, we give a full description of the expected form for the universal generalized integral means spectrum, $B(p,q)$, in terms of known or conjectured results  on the standard universal spectrum for univalent functions.
\subsection*{Acknowledgments}  It is a pleasure to thank Kari Astala for extended discussions about the universal generalized integral means spectrum.    \section{Expectations of logarithmic derivatives}\label{sec2}
   In this section, we prove the following simple result:
      \begin{proposition}\label{mart}
  Let $f(z)=f_0(z)$ be the interior whole-plane $\SLE_2$ map at time $0$, in the same setting as in Proposition \ref{prop1}; we then have
    $$\mathbb{E} \bigg( {z \frac{f'(z)}{f(z)} }\bigg) =1-z.$$
    \end{proposition}
\textcolor{black}{The method explained here will allow us to  swiftly move to more complicated cases in the next sections.}  Let us then introduce
\begin{equation} \label{G1}G(z):=\mathbb{E} \left(z \frac{f'(z)}{f(z)} \right), \end{equation}
and, \textcolor{black}{following Refs. \cite{MR2153402} and \cite{BS}},  aim at finding a partial differential equation satisfied by $G$. For the benefit of the reader not familiar with \textcolor{black}{Refs. \cite{BS,MR2153402}},  let us detail  the strategy of these papers that we will apply in various contexts here.

The starting point is to consider the radial $\SLE_\kappa$, solution to the ODE,
$$    \partial_tg_t(z)=g_t(z)\frac{\lambda(t)+g_t(z)}{\lambda(t)-g_t(z)}, \;z\in\mathbb D,$$
with the initial condition $g_0(z)=z$, and where $\lambda(t)=e^{i\sqrt{\kappa}B_t}$. The map $g_t$  conformally maps  a subdomain of the unit disk onto the latter. As we shall see shortly, the whole-plane map $f$ is rather related to the map $g_t^{-1}$, but this last function satisfies, by Loewner's theory, a PDE not well-suited to It\^{o} calculus. To overcome this difficulty, one runs backward the ODE of radial SLE, i.e., one compares  $g_t^{-1}$ to $g_{-t}$. This is the purpose of Lemma 1 in \cite{BS} (an analog of Lemma 3.1 in Ref. \cite{MR2153402}), which states that, for $t\in \mathbb R$, 
$g_{-t}(z)$ has the same law as  the process $\tilde{f}_t(z)$, defined as follows.
 \begin{definition} \label{reverse}
 The  (conjugate, inverse) radial SLE process $\tilde f_t$ is defined, for $t \in \mathbb R$, as
\begin{equation} \label{deftilde}\tilde{f}_t(z):=g_t^{-1}(z\lambda(t))/{\lambda(t)}.
\end{equation}
\end{definition}
The lemma then results from the simple observation that
$$\tilde{f}_s(z)=\hat{g}_{-s}(z),$$
where, for fixed $s\in \mathbb R$, the new process $\hat{g}_t(z):=g_{s+t}\circ g_s^{-1}(z\lambda(s))/\lambda(s)$ can be  shown to be a radial SLE. 
This lemma implies in particular that $\tilde{f}_t$ is solution to the ODE:
\begin{equation}\label{reverseeq}
   \partial_t\tilde f_t(z)=\tilde f_t(z)\frac{\tilde f_t(z)+\lambda(t)}{\tilde f_t(z)-\lambda(t)},\,\,\,\tilde f_0(z)=z.
   \end{equation}
  To apply It\^{o}'s stochastic calculus, one then uses Lemma 2 in Ref. \cite{BS}, which is a version of the SLE's Markov property,
  $$ \tilde{f}_t(z)=\lambda(s)\tilde{f}_{t-s}(\tilde{f}_s(z)/\lambda(s)).$$
  To finish, one has to relate the whole-plane SLE to the (modified) radial one. This is done through Lemma 3 in \cite{BS},  which is in our present setting (with a change of an $e^{-t}$ convergence factor there to an  $e^{t}$ factor here, when passing from the exterior to the interior of the unit disk $\mathbb D$):
  \begin{lemma}\label{lemmefond}
    The limit in law, $\lim_{t\to +\infty} e^{t}\tilde{f}_t(z)$,  exists, and has the {\it same law} as the (time zero) interior whole-plane random map $f_0(z)$:
    $$
    \lim_{t\to +\infty} e^t \tilde f_t(z)\stackrel{\rm (law)}{=} f_0(z).$$    \end{lemma}
    Let us now turn to the proof of Proposition \ref{mart}.
   \begin{proof}   
  Let us introduce the auxiliary, time-dependent, radial variant of the SLE one-point function $G(z)$ \eqref{G1} above,
    \begin{align}
    \label{tildeG} \widetilde{G}(z,t):=\mathbb{E}\left( z\frac{\tilde{f}_t'(z)}{\tilde{f}_t(z)} \right),
    \end{align}
   where $\tilde{f}_t$ is a modified radial SLE map at time $t$ as in  Definition \ref{reverse}. Owing to Lemma \eqref{lemmefond}, we have 
   \begin{equation}\label{limtildeG1}\lim_{t\to+\infty} \widetilde{G}(z,t)=G(z).
  \end{equation} 
  
    We then use a martingale technique to obtain an equation satisfied by $\widetilde{G}(z,t)$. 
     For $s\leq t$, define  $\mathcal{M}_s:=\mathbb{E}\left( \frac{\tilde{f}'_t(z)}{\tilde{f}_t(z)} \vert \mathcal{F}_s\right)$, where $\mathcal{F}_s$ is the $\sigma$-algebra generated by $\{B_u,\,u\leq s\}$. $(\mathcal{M}_s)_{s\ge 0}$ is by construction a martingale. Because of the Markov property of SLE, we have \cite{BS}
    \begin{align*}
     \mathcal{M}_s=\mathbb{E} \bigg( \frac{ \tilde{f}'_t(z) }{ \tilde{f}_t(z)} \vert \mathcal{F}_s \bigg) &= \mathbb{E} \bigg(\frac{ \tilde{f}'_{s}(z)}{\lambda(s)}  \frac{ \tilde{f}'_{t - s}(\tilde{f}_{s}(z)/\lambda(s))}{ \tilde{f}_{t - s}(\tilde{f}_{s}(z)/\lambda(s))} \vert \mathcal{F}_s\bigg)\\
     &=\frac{ \tilde{f}'_{s}(z)}{\lambda(s)} \mathbb{E}\bigg(\frac{ \tilde{f}'_{t - s}(\tilde{f}_{s}(z)/\lambda(s))}{ \tilde{f}_{t - s}(\tilde{f}_{s}(z)/\lambda(s))} \vert \mathcal{F}_s\bigg)\\
     &=\frac{ \tilde{f}'_{s}(z)}{\tilde{f}_s(z)} \widetilde{G}(z_s,\tau),
     \end{align*}
    where $z_s:=\tilde{f}_s(z)/\lambda(s)$, and $\tau:=t-s$.
    
   We  have from Eq. \eqref{reverseeq}
   \begin{align}\label{dlogf}
    \partial_s \log  \tilde{f}'_s &= \frac{\partial_z\left[\tilde{f}_s\frac{\tilde{f}_s+\lambda(s)}{\tilde{f}_s-\lambda(s)}\right]}{\tilde{f}'_s} =\frac{\tilde{f}_s+\lambda(s)}{\tilde{f}_s-\lambda(s)}-\frac{2\lambda(s)\tilde{f}_s}{(\tilde{f}_s-\lambda(s))^2} \\\nonumber
    &= 1-\frac{2}{(1-z_s)^2},\\\label{dlogfp}
    \partial_s \log \tilde{f}_s &=\frac{\partial_s\tilde{f}_s}{\tilde{f}_s}=\frac{z_s+1}{z_s-1},\\\label{dzs}
    dz_s&=z_s\bigg[\frac{z_z+1}{z_s-1}-\frac{\kappa}{2} \bigg]ds-i z_s \sqrt{\kappa}dB_s.
    \end{align}
    The coefficient of the $ds$-drift  term of the It\^o derivative of $\mathcal{M}_s$ is obtained from the above as,
 \begin{equation}\label{PtildeG}
 \frac{ \tilde{f}'_{s}(z)}{\tilde{f}_s(z)}\bigg[ -\frac{2z_s}{(1-z_s)^2} + z_s\left(\frac{z_s+1}{z_s-1}-\frac{\kappa}{2}  \right)\partial_{z}-\partial_{\tau}- \frac{\kappa }{2}z_s^2\partial^2_{z} \bigg]\widetilde{G}(z_s,\tau),
 \end{equation}
  and vanishes by the (local) martingale property. 
    Because $\tilde{f}_s$ is univalent, $\tilde{f}'_s$ does not vanish in $\mathbb D$, therefore the bracket above vanishes. 

Owing to the existence of the limit \eqref{limtildeG1}, we now take the limit as $\tau\to +\infty$ in the above, and  obtain the ODE,     
\begin{align}
    \label{equationG}
    \mathcal{P}(\partial)[G(z)]&:=-\frac{2z}{(1-z)^2}G(z) + z\left(\frac{z+1}{z-1}-\frac{\kappa}{2}  \right)G'(z)- \frac{\kappa}{2} z^2 G''(z)\\ \nonumber
   &=\left[-\frac{2z}{(1-z)^2} + z\left(\frac{z+1}{z-1}  \right)\partial_z- \frac{\kappa}{2} (z \partial_z)^2\right]G(z) =0.
    \end{align}
    
    \textcolor{black}{In the above, the exchange of the $\tau \to +\infty$ limit and of partial derivation of $\tilde G(z,\tau)$ with respect to $z$ is justified by the fact that the $\tau$-family $e^{\tau}\tilde f_\tau(z)$ in Lemma \eqref{lemmefond} and all its $z$-derivatives are normal, i.e., uniformly bounded in any compact of $\mathbb D$, so that the spatial derivatives of $\tilde G(z,\tau)$ form an equicontinuous family. A further requirement is that $\lim_{\tau \to +\infty} \frac{\partial}{\partial \tau} \tilde G=0$. Use of the Schramm-Loewner equation \eqref{reverseeq} for $\tilde f_\tau$ shows that
$$    \frac{\partial}{\partial \tau} \left(\frac{\tilde f'_\tau(z)}{\tilde f_\tau(z)}\right)=-\frac{2\tilde f'_\tau(z)}{(\tilde f_\tau(z)-\lambda(\tau))^2}.$$
Since $z\mapsto e^\tau\tilde{f}_\tau(z)$ belongs to the $\mathcal{S}$ class, classical Koebe distortion theorems then show that the right-hand side is bounded by $C(z) e^{-\tau}$, with $C$  defined in $\mathbb D$; this insures both the validity of the exchange of expectation and $\tau$-derivation, and the vanishing of the limit above.}

      Following Ref. \cite{DNNZ}, we now look for solutions to Eq. \eqref{equationG} of the form
   $\varphi_\alpha(z):=(1-z)^\alpha.$ We have
      $$\mathcal{P}(\partial)[\varphi_\alpha]=A(2,2,\alpha)\varphi_{\alpha}+B(2,\alpha)\varphi_{\alpha-1}+C(2,\alpha)\varphi_{\alpha-2},$$
where, in anticipation of the notation that will be introduced in Section \ref{sec3} below,
\begin{align*}
A(2,2,\alpha)&:=\alpha-\frac{\kappa}{2}\alpha^2,\\
B(2,\alpha)&:=2-\left(3+\frac{\kappa}{2}\right)\alpha+\kappa\alpha^2,\\
C(2,\alpha)&:=-2+\left(2+\frac{\kappa}{2}\right)\alpha-\frac{\kappa}{2}\alpha^2,
\end{align*}
with, identically, $A+B+C=0$.   The linear independence of $\varphi_{\alpha},\varphi_{\alpha-1}, \varphi_{\alpha-2}$ thus shows that $\mathcal{P}(\partial)[\varphi_\alpha]=0$ is equivalent to $A=B=C=0$, 
 which yields $\kappa=2, \alpha=1$, and $1-z$ is a solution. \textcolor{black}{That it must be $G$ follows from two facts. 
First, $G$ is holomorphic, since it is clearly continuous by Lebesgue theorem and holomorphicity then follows from Morera's theorem.
Second, the space of formal power series that are solutions to Eq. \eqref{equationG} is one-dimensional. This is analogous to Lemma 3.1 in Ref. \cite{DNNZ}, and follows here from the fact that 
$ \mathcal{P}(\partial)(z^n)=-(\frac{\kappa}{2}n^2+n)z^n+O(z^{n+1}).$}
\end{proof}. 
\section{SLE one-point Function} \label{sec3}
  \begin{figure}[htbp]
\begin{center}
\includegraphics[angle=0,width=.43290\linewidth]{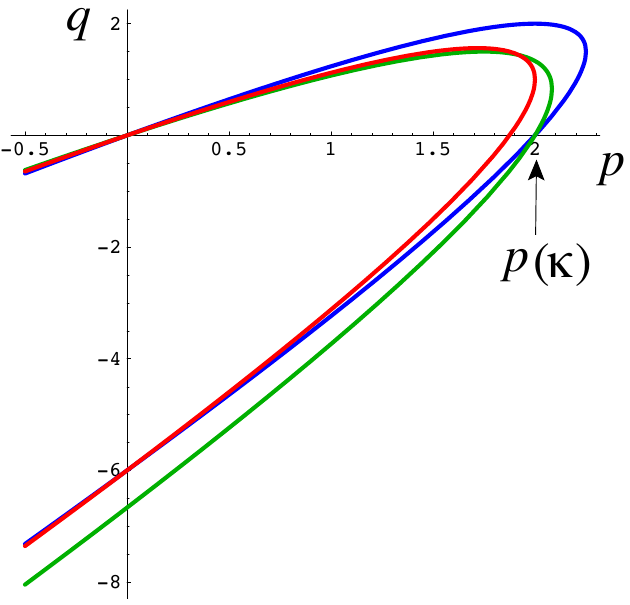}
\caption{{\it Integral curves  $\mathcal R$  of Theorem \ref{maintheorem1/2}, for $\kappa=2$ (\textcolor{blue}{blue}), $\kappa=4$ (\textcolor{red}{red}), and $\kappa=6$ (\textcolor{dgreen}{green}). In addition to the origin, the $q=0$ intersection point  with the $p$-axis is at $p(\kappa):={(6+\kappa)(2+\kappa)}/{8\kappa}$, with $p(2)=p(6)=2$ \cite{DNNZ,IL}.}}
\label{whpl2}
\end{center}
\end{figure}  Let us now turn to the natural generalization of Proposition \ref{mart}.
 \begin{theorem}\label{maintheorem1/2}
  Let $f(z)=f_0(z)$ be the interior whole-plane $SLE_\kappa$ map at time zero, in the same setting as in Proposition \ref{prop1}. Consider   
     the curve $\mathcal R$, defined parametrically by
        \begin{align}\label{para}
        p=-\frac{\kappa}{2} \gamma^2+\left(2+\frac{\kappa}{2}\right)\gamma,\,\,\,\,
     2p-q=\left(1+\frac{\kappa}{2}\right)\gamma,\,\,\, \gamma\in \mathbb R.
    \end{align}
   On $\mathcal R$, the whole-plane $\SLE_\kappa$ one-point function has the integrable form,
     $$\mathbb{E} \bigg( {\frac{(f'(z))^\frac{p}{2}}{(f(z)/z)^\frac{q}{2}} }\bigg) =(1-z)^{\gamma}.$$ 
   \end{theorem}    
    \begin{remark}  Eq. \eqref{para} describes a parabola in the $(p,q)$ plane (see  (Fig. \ref{whpl2}), which is given in Cartesian coordinates by 
    \begin{equation}\label{paraRCart}
    2\kappa \left(\frac{2p-q}{2+\kappa}\right)^2-\left(4+\kappa\right)\frac{2p-q}{2+\kappa}+p=0,  
    \end{equation}  
    with two branches,
     \begin{equation}\label{p,q}\begin{split}
      \gamma&=\gamma_0^{\pm}(p):= \frac{1}{2\kappa}\left( 4+\kappa \pm \sqrt{(4+\kappa)^2-8\kappa p} \right),\,\,\, p \leq \frac{(4+\kappa)^2}{8\kappa},
      \\ q&=2p-\left(1+\frac{\kappa}{2}\right)\gamma_0^{\pm}(p).
     \end{split}
\end{equation} 
or, equivalently,
\begin{equation}\label{p,qbis}
    2p =q +  \frac{2+\kappa}{8\kappa}\left(6+\kappa \pm \sqrt{(6+\kappa)^2-16\kappa q}\right),\,\,\,q \leq \frac{(6+\kappa)^2}{16\kappa}.     
 \end{equation}    
      \end{remark} 
 \begin{proof} Our aim is to derive an ODE satisfied by the whole-plane SLE one-point function,
\begin{equation}\label{Gpq} 
G(z):=\mathbb{E}\bigg(z^{\frac{q}{2}} \frac{(f'(z))^\frac{p}{2}}{(f(z))^\frac{q}{2}} \bigg),
\end{equation}
which, by construction, stays finite at the origin and  such that $G(0)=1$. 

Let us introduce the shorthand notation,
\begin{equation}\label{defX}
X_t(z):=\frac{(\tilde{f}'_t(z))^\frac{p}{2}}{(\tilde{f}_t(z))^\frac{q}{2}}, \end{equation}
where $\tilde f_t$ is the conjugate, reversed radial SLE process  in $\mathbb D$, as introduced in Definition \ref{reverse}, and such that 
by Lemma \ref{lemmefond}, the limit, $\lim_{t\to +\infty} e^t \tilde f_t(z)\stackrel{\rm (law)}{=} f_0(z)$, is the same in law as the whole-plane map at time zero. Applying the same method as in the previous section, we consider the time-dependent   function
 \begin{equation}\label{tildeGpq}
 \widetilde{G}(z,t):=\mathbb{E}\left(z^{\frac{q}{2}} X_t(z)\right),
 \end{equation}
 such that 
  \begin{equation}\label{inftytildeGpq}
 \lim_{t\to +\infty}\exp\left(\frac{p-q}{2}t\right)\widetilde{G}(z,t)=G(z).
 \end{equation}
 
  Consider now the martingale $(\mathcal{M}_s)_{t\geq s\ge 0}$, defined by
 $$\mathcal{M}_s=\mathbb{E}(X_t(z) \vert \mathcal{F}_s).$$
 By the SLE Markov property we get, setting $z_s:=\tilde{f}_s(z)/\lambda(s)$, 
  \begin{equation}\label{Ms}
      \mathcal{M}_s=X_s(z)\widetilde{G}(z_s,\tau),\,\,\,\tau:=t-s.
           \end{equation}
As before, the partial differential equation satisfied by $\widetilde{G}(z_s,\tau)$ is obtained by expressing the fact that the $ds$-drift term of  the It\^o differential of Eq. \eqref{Ms},
$$d\mathcal{M}_s= \widetilde G\,dX_s +X_s\, d\widetilde G,$$  vanishes.
The  differential of $X_s$  is simply computed from Eqs. \eqref{dlogf} and \eqref{dlogfp} above as:
\begin{equation}\label{diffX} \begin{split}
dX_s(z)&=X_s(z) F(z_s)ds,\\
F(z)&:=\frac{p}{2}\left[1-\frac{2}{(1-z)^2}\right]-\frac{q}{2}\left[1-\frac{2}{1-z}\right].
\end{split}\end{equation}
The  It\^o differential $d\widetilde G$ brings in the $ds$ terms proportional to $\partial_{z_s}\widetilde G,\,\partial^2_{z_s}\widetilde G$, and $\partial_\tau\widetilde G$; therefore, in the PDE satisfied by $\widetilde G$,  the latter terms are exactly the same as in the PDE \eqref{PtildeG}. We therefore directly arrive at the vanishing condition of the  overall drift term coefficient in $d\mathcal M_s$,
 \begin{equation}\label{PtildeGpq}
 X_s(z)\bigg[ F(z_s) + z_s\left(\frac{z_s+1}{z_s-1}-\frac{\kappa}{2}  \right)\partial_{z}-\partial_{\tau}- \frac{\kappa }{2}z_s^2\partial^2_{z} \bigg]\widetilde{G}(z_s,\tau)=0.
 \end{equation}
Since $X_s(z)$ does not vanish in $\mathbb D$, the bracket in \eqref{PtildeGpq} must identically vanish:
  \begin{equation}\label{PtildeGpqbis}
\bigg[ F(z_s) + z_s\frac{z_s+1}{z_s-1}\partial_{z}-\partial_{\tau}- \frac{\kappa }{2}(z_s\partial_{z})^2 \bigg]\widetilde{G}(z_s,\tau)=0,
 \end{equation}
 where we used  $z\partial_z+z^2\partial_z^2=(z\partial_z)^2$.
 
 To derive the ODE satisfied by $G(z)$ \eqref{Gpq}, we first recall its expression as the limit \eqref{inftytildeGpq}, which further implies 
$$\lim_{\tau\to +\infty}\exp{\left(\frac{p-q}{2}\tau\right)}\partial_\tau\tilde{G}(z,\tau)=-\frac{p-q}{2}G(z),$$
\textcolor{black}{with an exchange of derivative and limit similar to that in the proof of Proposition \ref{mart}.}
 Multiplying the PDE \eqref{PtildeGpq} satisfied by $\widetilde{G}$ by $\exp(\frac{p-q}{2}\tau)$ and letting $\tau\to+\infty$, we  get 
 \begin{align}\nonumber
 \mathcal{P}(\partial)[G(z)]&:=\left[-\frac{\kappa}{2}(z\partial_z)^2-\frac{1+z}{1-z}z\partial_z  +F(z)+\frac{p-q}{2} \right]G(z)\\ \label{PGpq}&=\left[-\frac{\kappa}{2}(z\partial_z)^2-\frac{1+z}{1-z}z\partial_z   -\frac{p}{(1-z)^2}+\frac{q}{1-z}+p-q \right]G(z)=0.
  \end{align} 
\textcolor{black}{Again, the space of  holomorphic solutions to (\ref{PGpq}) is one-dimensional, because $F(z)+(p-q)/2$ vanishes at $z=0$. For the boundary condition $G(0)=1$, we now look for special solutions of the form $\varphi_\alpha(z)=(1-z)^\alpha$.}  This function satisfies the simple differential operator algebra \cite{DNNZ}
\begin{equation}\label{Pdelta}\mathcal{P}(\partial)[\varphi_\alpha]=A(p,q,\alpha)\varphi_{\alpha}+B(q,\alpha)\varphi_{\alpha-1}+C(p,\alpha)\varphi_{\alpha-2},
\end{equation}
\begin{align}\label{A}
A(p,q,\alpha)&:=p-q+\alpha-\frac{\kappa}{2}\alpha^2,\\ \label{B}
B(q,\alpha)&:=q-\left(3+\frac{\kappa}{2}\right)\alpha+\kappa\alpha^2,
\\ \label{C}C(p,\alpha)&:=-p+\left(2+\frac{\kappa}{2}\right)\alpha-\frac{\kappa}{2}\alpha^2,
\end{align}
such that, identically,  $A+B+C=0$. Because $\varphi_{\alpha},\varphi_{\varphi-1},\varphi_{\alpha-2}$ are linearly independent, the condition $\mathcal{P}(\partial)[\varphi_\gamma]=0$ is equivalent to the system $A=C=0$, hence  $C(p,\gamma)=0$ and $A(p,q,\gamma)-C(p,\gamma)=2p-q-(1+\kappa/2)\gamma=0$. It yields precisely  the parabola parametrization \eqref{para} given in Theorem \ref{maintheorem1/2}, and has for solution \eqref{p,q}.
\end{proof}
    \section{SLE two-point function}\label{sec4}
\subsection{Beliaev--Smirnov type equations}\label{2pointGF}
 In this section, we will determine the mixed moments of moduli, $\mathbb{E}\left(\frac{\vert f'(z) \vert^p}{\vert f(z)\vert^q}\right)$, for  $(p,q)$ belonging to the same parabola $\mathcal R$ as in Theorem \ref{maintheorem1/2}, and where $f=f_0$ is the (time zero) interior whole-plane $\SLE_\kappa$ map. 
 
 In contradistinction to the method used in Refs. \cite{BS,DNNZ} for writing a PDE obeyed by $\mathbb E(|f'(z)|^p)$, we shall use here a slightly different approach, building on the results obtained in Section \ref{mart}. We shall  study the SLE two-point   function for $z_1,z_2\in \mathbb D$,
  \begin{equation}\label{gf2}
 G(z_1,\bar z_2):=\mathbb{E}\bigg(z_1^{\frac{q}{2}} \frac{(f'(z_1))^\frac{p}{2}}{(f(z_1))^\frac{q}{2}}  \overline{\left[z_2^{\frac{q}{2}}\frac{(f'(z_2))^\frac{p}{2}}{(f(z_2))^\frac{q}{2}}\right]}\bigg). 
 \end{equation}
 
As before, we define a time-dependent, auxiliary two-point function,
  \begin{equation}\label{tildeG2pq}
  \begin{split}
 \widetilde{G}(z_1,\bar z_2,t)&:=\mathbb{E}\left(z_1^{\frac{q}{2}}\frac{(\tilde{f}'_t(z_1))^\frac{p}{2}}{(\tilde{f}_t(z_1))^\frac{q}{2}}
\overline{\left[z_2^{\frac{q}{2}} \frac{(\tilde{f}'_t(z_2))^\frac{p}{2}}{(\tilde{f}_t(z_2))^\frac{q}{2}}\right]}\right)\\
 &=\mathbb{E}\left(z_1^{\frac{q}{2}} X_t(z_1)\overline{z_2^{\frac{q}{2}} X_t(z_2)}\right),
 \end{split}
 \end{equation} 
 where as above  $\tilde f_t$ is the reverse radial $\SLE_\kappa$ process \ref{reverse}, and where we used the shorthand notation \eqref{defX}. This time, the two-point function \eqref{gf2} is the limit
   \begin{equation}\label{inftytildeG2pq}
 \lim_{t\to +\infty}e^{(p-q)t}\widetilde{G}(z_1,\bar{z}_2,t)=G(z_1,\bar{z}_2).
 \end{equation} 

Let us define the two-point martingale $(\mathcal M_s)_{t\geq s\geq0}$, with 
$$\mathcal{M}_s:=\mathbb{E} (X_t(z_1)\overline{X_t(z_2)} \vert \mathcal{F}_s ).$$
 By the Markov property of SLE,
 \begin{equation}\label{Ms2}
 \mathbb{E} \big(X_t(z_1)\overline{X_t(z_2)} \vert \mathcal{F}_s \big) 
  = X_s(z_1)\overline{X_s(z_2)}\, \widetilde{G}(z_{1s},\bar z_{2s},\tau),\,\,\,\tau:=t-s,
 \end{equation}
 where 
 \begin{align}
 z_{1s}:=\tilde{f}_s(z_1)/\lambda(s); \,\,\, \bar z_{2s}:=\overline{\tilde{f}_s(z_2)/\lambda(s)}=\overline{\tilde{f}_s(z_2)}\lambda(s).
  \end{align} 
  Their It\^o differentials, $dz_{1s}$ and $d\bar{z}_{2s}$, are as in \eqref{dzs},
 \begin{equation}\label{dz12}\begin{split}
 dz_{1s}&=z_{1s}\bigg[ \frac{z_{1s}+1}{z_{1s}-1}-\frac{\kappa}{2} \bigg]ds-i\sqrt{\kappa}\,z_{1s}\,dB_s,\\
 d\bar z_{2s}&=\bar z_{2s}\bigg[ \frac{\bar{z}_{2s}+1}{\bar{z}_{2s}-1}-\frac{\kappa}{2} \bigg]ds+i\sqrt{\kappa}\,\bar{z}_{2s}\,dB_s. 
 \end{split}\end{equation} 
 As before, the partial differential equation satisfied by $\widetilde{G}(z_{1s},z_{2s},\tau)$ is obtained by expressing the fact that the $ds$-drift term of  the It\^o differential of Eq. \eqref{Ms2},
\begin{equation}\label{dMs2}
d\mathcal{M}_s= [dX_s(z_1)\overline{X_s(z_2)} +X_s(z_1)d\overline{X_s(z_2)}]\,\widetilde G+ X_s(z_1)\overline{X_s(z_2)}\, d\widetilde G,
\end{equation}  
vanishes.

The  differentials of $X_s$, $\overline{X_s}$  are as in  Eq. \eqref{diffX}  above:
\begin{equation}\label{diffXbis} \begin{split}
dX_s(z_1)&=X_s(z_1) F(z_{1s})ds,\,\,\,
d\overline{X_s(z_2)}=\overline{X_s(z_2)} F(\bar z_{2s})ds,\\
F(z)&:=\frac{p}{2}-\frac{q}{2}-\frac{p}{(1-z)^2}+\frac{q}{1-z}.
\end{split}\end{equation} 
We thus obtain the simple expression
\begin{equation}\label{dMs2bis}
d\mathcal{M}_s= X_s(z_1)\overline{X_s(z_2)} \left[\left[F(z_{1s})+F(\bar{z}_{2s})\right]\,\widetilde G\,ds+ d\widetilde G\right],
\end{equation}  
and the vanishing of the $ds$-drift term in $d\mathcal{M}_s$ requires that of the drift term in the right-hand side bracket in \eqref{dMs2bis}, since $X_s(z)$ does not vanish in $\mathbb D$. 

The It\^o differential of $\widetilde G(z_{1s},\bar{z}_{2s},\tau)$ can be obtained  from Eqs. \eqref{dz12} and It\^o calculus as
\begin{align}\label{difftG}
d\widetilde G(z_{1s},\bar{z}_{2s},\tau)=&\partial_1\widetilde G\, dz_{1s} +\bar{\partial}_2\widetilde G\, d\bar{z}_{2s} -\partial_\tau\widetilde G\,ds\\ \nonumber
&-\frac{\kappa}{2} z_{1s}^2\,\partial_1^2\widetilde G\, ds -\frac{\kappa}{2}  \bar{z}_{2s}^2\,\bar{\partial}_2^2\widetilde G\, ds +\kappa  z_{1s}\, \bar{z}_{2s}\, \partial_1\bar{\partial}_2\widetilde G\, ds,
\end{align} 
where use was made of the shorthand notations, $\partial_1:=\partial_{z_1}$ and  $\bar{\partial}_2:=\partial_{\bar{z}_2}$. 
We observe that the only coupling between the $z_{1s},\bar{z}_{2s}$ variables arises in the last term of \eqref{difftG}, the other terms simply resulting from the independent contributions of the $z_{1s}$ and $\bar{z}_{2s}$ parts.

Using again the It\^o differentials \eqref{dz12}, we can rewrite \eqref{difftG} as
\begin{align}\label{difftGf}
d\widetilde G&=-i\sqrt{\kappa} \left(z_{1s} \partial_1- \bar{z}_{2s} \bar{\partial}_2\right)\widetilde G\,dB_s\\ \nonumber &+\frac{z_{1s}+1}{z_{1s}-1} z_{1s} \partial_1\widetilde G\,ds +\frac{\bar{z}_{2s}+1}{\bar{z}_{2s}-1} \bar{z}_{2s} \bar{\partial}_2\widetilde G\,ds  -\partial_\tau\widetilde G\,ds\\ \nonumber
&-\frac{\kappa}{2} (z_{1s}\partial_1-  \bar{z}_{2s}\,\bar{\partial}_2)^2\widetilde G\, ds,
\end{align} 
where we used the obvious formal identity 
\begin{equation}\label{zdzd}
(z_{1}\partial_1)^2+ (\bar{z}_{2}\,\bar{\partial}_2)^2-2z_{1}\partial_1\bar{z}_{2}\,\bar{\partial}_2=(z_{1}\partial_1-  \bar{z}_{2}\,\bar{\partial}_2)^2.
\end{equation} 

At this stage, comparing the computations \eqref{dMs2bis} and \eqref{difftGf}
above with those  in the one-point martingale study in Section \ref{mart}, it is clear that the PDE obeyed by $\widetilde G=\widetilde G(z_{1s},\bar{z}_{2s},\tau)$ is obtained as two duplicates of Eq. \eqref{PtildeGpqbis},  completed  as in \eqref{zdzd} by the derivative coupling between variables $z_{1s}$, $\bar{z}_{2s}$:
 \begin{equation}\label{PtildeG2pqbis}
\bigg[ F(z_{1s}) + z_{1s}\frac{z_{1s}+1}{z_{1s}-1}\partial_{1} +F(\bar{z}_{2s}) + \bar{z}_{2s}\frac{\bar{z}_{2s}+1}{\bar{z}_{2s}-1}\bar{\partial}_{2}-\partial_{\tau}- \frac{\kappa }{2}(z_{1s}\partial_{1}-\bar{z}_{2s}\bar{\partial}_{2})^2 \bigg]\widetilde{G}=0.
 \end{equation}
The existence of the limit \eqref{inftytildeG2pq} further implies that of
$$\lim_{\tau\to \infty}e^{(p-q)\tau}\partial_\tau\tilde{G}(z_1,\bar{z}_2,\tau)=-(p-q)G(z_1,\bar{z}_2).$$
 Multiplying the PDE \eqref{PtildeG2pqbis} satisfied by $\widetilde{G}$ by $\exp((p-q)\tau)$ and letting $\tau\to+\infty$, then gives the expected PDE for $G(z_1,\bar{z}_2)$.  It can be most compactly written in terms of the ODE \eqref{PGpq} as  
\begin{equation}\label{compactPDE}\left[\mathcal P(\partial_1)+\mathcal P(\bar{\partial}_2) +\kappa z_1\partial_1\bar{z}_2\bar{\partial}_2\right]G(z_1,\bar{z}_2)=0,\end{equation}
and its fully explicit expression is
 \begin{align}\label{eq1}  &\mathcal{P}(D)[G(z_1,\bar{z}_2)]=-\frac{\kappa}{2}(z_1\partial_1-\bar{z}_2\bar{\partial}_2)^2G-\frac{1+z_1}{1-z_1}z_1\partial_1 G-\frac{1+\bar{z}_2}{1-\bar{z}_2}\bar{z}_2\bar{\partial}_2 G\\ \nonumber
  &+\bigg[ -\frac{p}{(1-z_1)^2}-\frac{p}{(1-\bar{z}_2)^2}+\frac{q}{1-z_1}+\frac{q}{1-\bar{z}_2}+2p-2q \bigg]G=0.
 \end{align}
 
\subsection{Moduli one-point function}
 Note that one can take the $z_1=z_2=z$ case in Definition \eqref{gf2} above, thereby obtaining the moduli one-point  function,
 \begin{equation}\label{Gzz}
 G(z,\bar z)= \mathbb{E}\left(|z|^q\frac{\vert f'(z) \vert^p}{\vert f(z)\vert^q}\right).
 \end{equation}
 Because of Eq. \eqref{eq1}, it obeys the corresponding ODE,
  \begin{align}\label{eq1zz}  &\mathcal{P}(D)[G(z,\bar{z})]=-\frac{\kappa}{2}(z\partial-\bar{z}\bar{\partial})^2G-\frac{1+z}{1-z}z\partial G-\frac{1+\bar{z}}{1-\bar{z}}\bar{z}\bar{\partial} G\\ \nonumber
  &+\bigg[ -\frac{p}{(1-z)^2}-\frac{p}{(1-\bar{z})^2}+\frac{q}{1-z}+\frac{q}{1-\bar{z}}+2p-2q \bigg]G=0,
 \end{align}
 which is the generalization to $q\neq 0$ of the Beliaev--Smirnov equation studied in Refs. \cite{DNNZ} and \cite{IL}. 
 \subsection{Integrable case}
  \begin{lemma}
 The space of formal series $F(z_1,\bar{z}_2)=\sum_{k,\ell \in \mathbb N} a_{k,\ell} z_1^k\bar{z}_2^{\ell}$, with complex coefficients and that are solutions of the PDE \eqref{eq1}, is one-dimensional.
 \end{lemma}
 \begin{proof}
 We assume that $F$ is a solution to (\ref{eq1}) with $F(0,0)=0$; it suffices to prove that, necessarily, $F=0$. We argue by contradiction: If not, consider the minimal (necessarily non constant) term $a_{k,l}z^k \bar{z}^{\ell}$ in the series of $F$, with $a_{k,\ell}\ne 0$ and $k+\ell$ minimal (and non vanishing). Then $\mathcal P(D)[F]$ (\ref{eq1}) will have a minimal  term, equal to 
 $-a_{k,\ell}\left[ \frac{\kappa}{2}(k-\ell)^2+k+\ell \right]z_1^k\bar{z}_2^{\ell},$ which is non-zero, 
 contradicting the fact that $\mathcal P(D)[F]$ vanishes.
 \end{proof}
 
  As a second step, following Ref. \cite{DNNZ}, 
let us  consider the action of the operator $\mathcal P(D)$ of (\ref{eq1}) on a function of the factorized form $\varphi(z_1)\varphi(\bar{z}_2)P(z_1,\bar{z}_2)$, which we write, in a shorthand notation, as $\varphi\bar{\varphi} P$.  By Leibniz's rule, it is given by
   \begin{align*}
 \mathcal P(D)[\varphi\bar{\varphi}P] = &-\frac{\kappa}{2}\varphi\bar{\varphi}(z_1\partial_1-\bar{z_2}\bar{\partial_2})^2P-\kappa(z_1\partial_1-\bar{z_2}\bar{\partial_2})(\varphi\bar{\varphi})(z_1\partial_1 - \bar{z_2}\bar{\partial_2})P\\
 & +\kappa(z_1\partial_1\varphi)(\bar{z_2}\bar{\partial_2}\bar{\varphi})P - \varphi\bar{\varphi}\frac{1+z_1}{1-z_1}z_1\partial_1 P - \varphi\bar{\varphi}\frac{1+\bar{z_2}}{1-\bar{z_2}}\bar{z_2}\bar{\partial_2} P\nonumber\\
 &- \bigg[ \frac{\kappa}{2}\bar{\varphi}(z_1\partial_1)^2\varphi +\frac{\kappa}{2}\varphi(\bar{z_2}\bar{\partial_2})^2\bar{\varphi} + \bar{\varphi}\frac{1+z_1}{1-z_1}z_1\partial_1\varphi + \varphi\frac{1+\bar{z_2}}{1-\bar{z_2}}\bar{z_2}\bar{\partial_2}\bar{\varphi} \bigg]P\nonumber\\
 &+\bigg[ -\frac{p}{(1-z_1)^2}-\frac{p}{(1-\bar{z_2})^2}+\frac{q}{1-z_1}+\frac{q}{1-\bar{z_2}}+2p-2q \bigg]\varphi\bar{\varphi}P.\nonumber
 \end{align*}   
  Note that  the operator $z_1\partial_1-\bar{z_2}\bar{\partial_2}$ is antisymmetric with respect to $z_1,\bar{z}_2$; therefore, if we choose a symmetric function, $P(z_1,\bar{z}_2)=P(z_1\bar{z}_2)$, the first line of $ \mathcal P(D)[\varphi\bar{\varphi}P]$ above identically vanishes.
  
 One then looks for  solutions to \eqref{eq1} of the particular form,
 $$G(z_1,\bar{z}_2)=\varphi_\alpha(z_1)\varphi_\alpha(\bar{z}_2)P(z_1\bar{z}_2),$$
 where, as before, $\varphi_\alpha(z)=(1-z)^\alpha$. 
 The action of the differential operator then takes the simple form, 
 \begin{align*}
 \mathcal {P}(D)[\varphi_\alpha\bar{\varphi}_\alpha P]=&z_1\bar{z}_2 \varphi_{\alpha-1}\bar{\varphi}_{\alpha-1}\left(\kappa\alpha^2 P
  - 2(1-z_1\bar{z}_2)P' \right)\\ &+ \mathcal P(\partial_1)[\varphi_\alpha]\bar{\varphi}_\alpha P+\mathcal P(\partial_2)[\bar{\varphi}_\alpha]\varphi_\alpha P,
 \end{align*}
 where $P'$ is the derivative of $P$ with respect to $z_1\bar{z}_2$, and  $\mathcal P(\partial)$ is the so-called  boundary operator  \eqref{PGpq}   \cite{DNNZ}. 
 
 The ODE, $\kappa\alpha^2 P(x) - 2(1-x)P'(x)=0$ with $x=z_1\bar{z}_2$ and $P(0)=1$, has for  solution 
$P(z_1\bar{z}_2)=(1-z_1\bar{z}_2)^{-\kappa\alpha^2/2}$. 
   It is then sufficient to pick for $\alpha$ the value $\gamma=\gamma^{\pm}_0(p)$ \eqref{p,q} such that $\mathcal{P}(\partial)[\varphi_\gamma]=0$, as obtained  in the proof of Theorem \ref{maintheorem1/2}, to get a solution of the PDE, $ \mathcal {P}(D)[\varphi_\gamma\bar{\varphi}_\gamma P]=0$ \eqref{eq1}. By uniqueness of the solution with $G(0,0)=1$, it gives the explicit form of the SLE two-point function,  
 $$G(z_1,\bar{z}_2)=\varphi_\gamma(z_1)\varphi_\gamma(\bar{z}_2)(1-z_1\bar{z}_2)^{-\kappa\gamma^2/2}.$$
 We thus get: 
  \begin{theorem}\label{maintheorem}
 Let $f(z)=f_0(z)$ be the interior whole-plane $\SLE_\kappa$ map in the setting of  Proposition \eqref{prop1}; then, for $(p,q)$ belonging to the parabola $\mathcal R$ defined in Theorem \ref{maintheorem1/2} by Eqs. \eqref{para} or \eqref{paraRCart} or  \eqref{p,q}, and for any pair $(z_1,z_2)\in \mathbb D\times \mathbb D$,
   \begin{align*}\mathbb{E}\bigg(z_1^{\frac{q}{2}} \frac{(f'(z_1))^\frac{p}{2}}{(f(z_1))^\frac{q}{2}}  \overline{\left[z_2^{\frac{q}{2}}\frac{(f'(z_2))^\frac{p}{2}}{(f(z_2))^\frac{q}{2}}\right]}\bigg) &= \frac{(1-z_1)^{\gamma}(1-\bar{z}_2)^{\gamma}}{(1-z_1\bar{z}_2)^\beta},\,\,\,\,\,\,\beta=\frac{\kappa}{2}\gamma^2.
     \end{align*}
 \end{theorem}
 
 \begin{corollary}
 In the same setting as in Theorem \ref{maintheorem}, we have for $z\in \mathbb D$,
    \begin{align*}
     \mathbb{E}\bigg(\vert z \vert^q{\frac{\vert f'(z) \vert^p}{\vert f(z)\vert^q}}\bigg) &= \frac{(1-z)^{\gamma}(1-\bar{z})^{\gamma}}{(1-z\bar{z})^\beta},\,\,\,\,\,\,\beta=\frac{\kappa}{2}\gamma^2,
     \end{align*} 
     for
   \begin{equation*}
   \begin{split}
      \gamma&=\gamma_0^{\pm}(p):= \frac{1}{2\kappa}\left( 4+\kappa \pm \sqrt{(4+\kappa)^2-8\kappa p} \right),\,\,\, p \leq \frac{(4+\kappa)^2}{8\kappa},
      \\ q&=2p-\left(1+\frac{\kappa}{2}\right)\gamma_0^{\pm}(p).
     \end{split}
\end{equation*}    
      \end{corollary}
 Let us stress some particular cases of interest. First, the $p=0$ case gives some integral means of $f$.  
   \begin{corollary} The interior whole-plane $\SLE_\kappa$ map has the integrable moments, $$\mathbb{E} \left( \left[ \frac{f(z_1)}{z_1}\right] ^{\frac{(2+\kappa)(4+\kappa)}{4\kappa}} \left[\frac{\overline{f(z_2)}}{\bar{z}_2}\right]^{\frac{(2+\kappa)(4+\kappa)}{4\kappa}}\right) = \frac{(1-z_1)^{\frac{4+\kappa}{\kappa}}(1-\bar{z}_2)^{\frac{4+\kappa}{\kappa}}}{(1-z_1\bar{z}_2)^{\frac{(4+\kappa)^2}{2\kappa}}},$$
   $$\mathbb{E} \left( \left\vert \frac{f(z)}{z} \right\vert^{\frac{(2+\kappa)(4+\kappa)}{2\kappa}} \right) = \frac{(1-z)^{\frac{4+\kappa}{\kappa}}(1-\bar{z})^{\frac{4+\kappa}{\kappa}}}{(1-z\bar{z})^{\frac{(4+\kappa)^2}{2\kappa}}}.$$
  \end{corollary}
    Second, taking $p=q$ yields the logarithmic integral means we started with:
  \begin{corollary}\label{theorem p=q}
 The interior whole-plane $\SLE_\kappa$ map has the 
   integrable logarithmic derivative two-point function, 
    $$\mathbb{E} \bigg({\left[ z_1\frac{f'(z_1)}{f(z_1)} \right]}^{\frac{2+\kappa}{2\kappa}}{\left[\bar z_2\frac{\overline{f'(z_2)}}{\overline{f(z_2)}} \right]}^{\frac{2+\kappa}{2\kappa}}\bigg) =\frac{(1-z_1)^{\frac{2}{\kappa}}(1-\bar{z_2})^{\frac{2}{\kappa}}}{(1-z_1\bar{z_2})^{\frac{2}{\kappa}}},$$   
  $$\mathbb{E} \bigg( {\left| z\frac{f'(z)}{f(z)} \right|}^{\frac{2+\kappa}{\kappa}} \bigg) =\frac{(1-z)^{\frac{2}{\kappa}}(1-\bar{z})^{\frac{2}{\kappa}}}{(1-z\bar{z})^{\frac{2}{\kappa}}}.$$
\end{corollary}  
 Theorem \ref{theorem kappa2} describes the $\kappa = 2$ case of the latter result.

  \subsection{Generalization to processes with $m$-fold symmetry}\label{mfold}
  
The moments,
$ \mathbb {E}(\vert (f^{[m]})'(z)\vert^p)$ (for $m \in \mathbb N\setminus \{0\}$), 
as well as their associated integral means spectra were studied in Ref. \cite{DNNZ}. Using It\^o calculus, a PDE satisfied by these moments was derived for each value of $m$. 
The introduction of mixed $(p,q)$ moments allows us to circumvent these calculations in a unified approach for $m\in \mathbb Z\setminus \{0\}$. To see this, notice that
$$(f^{[m]})'(z)=z^{m-1}f'(z^m)f(z^m)^{\frac{1}{m}-1}.$$
As a consequence,
$$ \frac{\vert z\vert^q\vert (f^{[m]})'(z)\vert^p}{\vert f^{[m]}(z)\vert^q}=\vert z\vert^{q+p(m-1)}\frac{\vert f'(z^m)\vert^p}{\vert f(z^m)\vert^{p+\frac{q-p}{m}}},$$
so that we identically have
\begin{align}\label{midentity} &\mathbb {E}\left(\vert z\vert^q \frac{\vert (f^{[m]})'(z)\vert^p}{\vert f^{[m]}(z)\vert^q}\right)=G(z^m;p,q_m),\\ \label{qm}
&q_m=q_m(p,q):=p+\frac{q-p}{m}, \end{align}
with the notation, \begin{align}\label{Gzpq}
& G(z;p,q):=G(z,\bar z)=\mathbb {E}\left(\vert z\vert^q  \frac{\vert f'(z)\vert^p}{\vert f(z)\vert^q}\right),
\end{align}
where we have made explicit the dependence on the $(p,q)$ parameters of the SLE moduli one-point function \eqref{Gzz} introduced  in Section \ref{sec3}. From Theorem \ref{maintheorem}, we immediately get the following.
\begin{theorem}\label{mcase} Let $f^{[m]}$ be the $m$-fold whole-plane $\SLE_\kappa$ map, $m \in \mathbb Z\setminus \{0\}$, with $z\in \mathbb D$ for $m>0$ and $z\in \mathbb C\setminus \overline{\mathbb D}$ for $m<0$. Then, 
$$\mathbb {E}\left(\vert z\vert^q \frac{\vert (f^{[m]})'(z)\vert^p}{\vert f^{[m]}(z)\vert^q}\right)=\frac{(1-z^m)^\alpha (1-\bar{z}^m)^\alpha}{(1-(z\bar{z})^m)^{\frac{\kappa}{2}\alpha^2}},$$
for $(p,q)$ belonging to the $m$-dependent parabola $\mathcal R^{[m]}$, given in parametric form by
\begin{align}  \label{param}p=\left(2+\frac{\kappa}{2}\right)\alpha-\frac{\kappa}{2}\alpha^2,\,\,\,
q=\left(m+2+\frac{\kappa}{2}\right)\alpha-\frac{\kappa}{2}(m+1)\alpha^2, \,\,\,\alpha\in \mathbb R.
\end{align}
In Cartesian coordinates, an equivalent statement is 
$$\alpha = \frac{(m+1)p - q}{m\left(1+\frac{\kappa}{2}\right)},$$
with
\begin{align*}
q = (m+1)p  - m\frac{2+\kappa}{4\kappa} \left(4+\kappa \pm \sqrt{(4+\kappa)^2-8\kappa p}\right),\,\,\,
p\le \frac{(4+\kappa)^2}{8\kappa},
\end{align*}
or,
\begin{align*}p &= \frac{q}{m+1} + \frac{m}{(m+1)^2}\frac{2+\kappa}{4\kappa}\left(2m+4+\kappa \pm \sqrt{(2m+4 +\kappa)^2 - 8(m+1)\kappa q}\right),\\
q &\le \frac{(2m+4+\kappa)^2}{8(m+1)\kappa}.
\end{align*}
\end{theorem} 

As for logarithmic coefficients, first observe that trivially,
\begin{equation}\label{logm}
\log\frac{f^{[m]}(z)}{z}=\frac{1}{m} \log\frac{f(z^m)}{z^m}.
\end{equation}
From this, and Theorem  \ref{logarithmic theorem}, we thus get 
\begin{corollary}
Let $f^{[m]}(z)$ be the $m$-fold whole-plane $\SLE_2$ map and
   \begin{equation}\label{ptlog2}
   \log \frac{f^{[m]}(z)}{z}=2\sum_{n\geq 1}\gamma^{[m]}_{n} z^n;
   \end{equation}
   then
   $$\mathbb{E}(\vert \gamma^{[m]}_{n} \vert^2)= \left\{ 
   \begin{array}{ll}
   \frac{1}{2n^2}  &n=mk,\,\,\,\,k\geq 1,\\
  0 &\mathrm{otherwise.}
   \end{array} \right.$$
  \end{corollary}
 We can also see this result as a corollary of  Theorem \ref{mcase}, which, for the logarithmic case $p=q$, and for any value of  $m$, yields  $p=q=2$ for $\kappa=2$ as the only integrable case.

\section{Generalized integral means spectrum}\label{secintmean}
In this section we aim at generalizing to the setting of the present work the integral means spectrum analysis of Refs. \cite{BS}, \cite{DNNZ} and \cite{BDZ} (see also \cite{IL,2012arXiv1203.2756L,2013JSMTE..04..007L}) concerning the whole-plane $\SLE$. The original work by Beliaev--Smirnov \cite{BS} and Ref. \cite{BDZ} deal with the exterior case, whereas Ref. \cite{DNNZ} and this work are concerned with the interior case, both being related by duality \eqref{qq'}. \subsection{Modified One-Point Function}
Let us first consider the {\it modified} SLE moduli one-point function,
 \begin{equation}\label{Fzz}
 F(z,\bar z):=\frac{1}{\vert z\vert^{q}}G(z,\bar z)= \mathbb{E}\left(\frac{\vert f'(z) \vert^p}{\vert f(z)\vert^q}\right).
 \end{equation}
 Because of Eq. \eqref{eq1zz},  it obeys the modified PDE,
  \begin{align}\label{eq1Fzz}  \mathcal{P}(D)[F(z,\bar{z})]=&-\frac{\kappa}{2}(z\partial-\bar{z}\bar{\partial})^2F-\frac{1+z}{1-z}z\partial F-\frac{1+\bar{z}}{1-\bar{z}}\bar{z}\bar{\partial} F\\ \nonumber
  &+\bigg[ -\frac{p}{(1-z)^2}-\frac{p}{(1-\bar{z})^2}+2p-q \bigg]F(z,\bar z)=0,
 \end{align}
which, of course, differs from Eq. \eqref{eq1zz}. We can rewrite it as
  \begin{align}\label{eq1Fzzbis}  \mathcal{P}(D)[F(z,\bar{z})]=&-\frac{\kappa}{2}(z\partial-\bar{z}\bar{\partial})^2F-\frac{1+z}{1-z}z\partial F-\frac{1+\bar{z}}{1-\bar{z}}\bar{z}\bar{\partial} F\\ \nonumber
  &-p\bigg[ \frac{1}{(1-z)^2}+\frac{1}{(1-\bar{z})^2}+\sigma-1 \bigg]F=0,
 \end{align}in term of the important new parameter, 
\begin{equation}\label{sigmapq}
\sigma:=q/p-1.
\end{equation}
 This PDE then exactly coincides with Eq. (106) in Ref.  \cite{DNNZ}, where $\sigma$ was meant to represent $\pm 1$, whereas here $\sigma\in\mathbb R$. 
 
 The value $\sigma=+1$ corresponds to the original Beliaev--Smirnov case, where the integral means spectrum successively involves three functions \cite{BDZ,BS,BDone,2000PhRvL..84.1363D,2014arXiv1412.8764G,PhysRevLett.88.055506,0911.3983}:
\begin{align} \label{tip}
\beta_{\textrm{tip}}(p,\kappa):=& -p-1+\frac 14(4+\kappa-\sqrt{(4+\kappa)^2-8\kappa p}), \\   \label{p0**} & \textrm{for}\,\,\, p\leq p'_0(\kappa):=-1-\frac{3\kappa}{8}; \\ \label{bulk} \beta_0(p,\kappa):=&-p+\frac{4+\kappa}{4\kappa}(4+\kappa-\sqrt{(4+\kappa)^2-8\kappa p}),\\ \nonumber & \mathrm{for}\,\,\, p'_0(\kappa)\leq p\leq p_{0}(\kappa);\\ \label{lin}
\beta_{\textrm{lin}}(p,\kappa):=& \,\,p-\frac{(4+\kappa)^2}{16\kappa},\\ \label{p0*}  &\textrm{for}\,\,\,  p\geq  p_{0}(\kappa):=\frac{3(4+\kappa)^2}{32\kappa}.
\end{align}

As shown in Refs. \cite{DNNZ,IL,2012arXiv1203.2756L,2013JSMTE..04..007L} in the $\sigma=-1$ interior case, because of the unboundedness of the interior whole-plane $\SLE$ map, there exists a phase transition at $p=p^*(\kappa)$, with 
\begin{align}
\nonumber p^*(\kappa)&:=
  \frac{1}{16\kappa}\left((4+\kappa)^2-4-2\sqrt{2(4+\kappa)^2+4}\right)\\ \label{pstar} &=\frac{1}{32\kappa}\left(\sqrt{2(4+\kappa)^2+4}-6\right) \left(\sqrt{2(4+\kappa)^2+4}+2\right).\end{align}
 The 
   integral means spectrum is afterwards given by \begin{equation}\label{beta+} 
\beta(p,\kappa):=3p-\frac{1}{2}-\frac 12\sqrt{1+2\kappa p},\,\, \textrm{for}\,\,\,  p\geq p^{*}(\kappa).
\end{equation}
  Since $p^*(\kappa)< p_{0}(\kappa)$ \eqref{p0*}, this transition precedes and supersedes the transition from the bulk spectrum  \eqref{bulk} towards the linear behavior \eqref{lin}. 
 
\textcolor{black}{The search in Ref. \cite[Section 4.2.1]{DNNZ} for exact boundary solutions to Eq. \ref{eq1Fzz} (in the $\sigma=\pm 1$ cases) led to the introduction of the $\sigma$-dependent spectrum function \cite[Eq. (180)]{DNNZ}} 
\begin{equation} \label{betasigma+}\beta_{+}^\sigma(p,\kappa)=(1-2\sigma)p-\frac{1}{2}\big(1+ \sqrt{1-2\sigma\kappa p}\big).
\end{equation}
For $\sigma=-1$, it recovers for interior whole-plane SLE the integral means spectrum \eqref{beta+} mentioned above. For $\sigma=+1$ it introduces a new spectrum,
\begin{equation} \label{betasigma++}\beta_{+}^{(+1)}(p,\kappa)=-p-\frac{1}{2}\big(1+ \sqrt{1-2\kappa p}\big),
\end{equation} the relevance of which for the exterior whole-plane SLE case is analyzed in a joint work by D. Beliaev and two of the present authors  \cite{BDZ}. 

For general real values of $\sigma$ \eqref{sigmapq}, we can rewrite \eqref{betasigma+} as a function of $(p,q,\kappa)$,  
\begin{equation} \label{betapq}\beta_{+}^{\sigma}(p,\kappa)=\beta_1(p,q;\kappa):=3p-2q-\frac{1}{2}-\frac{1}{2}\sqrt{1+2\kappa (p-q)}.
\end{equation}
We claim that the spectrum generated by the integral means in Definition \eqref{def:gims} in the general $(p,q)$ case will involve the standard multifractal spectra \eqref{tip}, \eqref{bulk}, \eqref{lin}, that are independent of $q$, and also the new $(p,q)$-dependent multifractal spectrum \eqref{betapq}. Phase transitions between these spectra will occur along lines drawn  in the $(p,q)$ plane. \textcolor{black}{As we shall see in Section  \ref{proofs}, the analysis of the integral means spectrum performed in Ref. \cite{DNNZ},  in particular that concerning the range of validity of the form \eqref{beta+} of \eqref{betasigma+} for $\sigma=-1$, as well as that given in Ref. \cite{BDZ} for the range of validity of the form \eqref{betasigma++} for $\sigma=+1$, and their corresponding proofs, can be extended to general values of the  $\sigma$ parameter via   Propositions \ref{theo:racket} and \ref{theo:DD} to establish Theorem \ref{propseparbis}.} 

We first describe the corresponding partition of the  $(p,q)$ plane into the respective domains of validity of the four spectra above. For this, we  need to determine boundary curves where  pairs (possibly triplets) of these spectra coincide, that are marking the onset of the respective phase transitions.

\subsection{Phase transition lines}\label{PTLines}
\textcolor{black}{Following Ref. \cite{DNNZ}, let us first introduce the analytical form of the various multifractal spectra based on the use of functions $A$ \eqref{A}, $B$ \eqref{B} and $C$ \eqref{C}.} It will be convenient to use the notation \cite{DNNZ},
\begin{align}\label{Asigma}
A^\sigma(p,\gamma):=- \frac{\kappa}{2}\gamma^2+\gamma -\sigma p,
\end{align}
such that for $\sigma=q/p-1$ \eqref{sigmapq},
\begin{equation}\label{AA}
A^\sigma(p,\gamma)=A(p,q;\gamma):=p-q+\gamma-\frac{\kappa}{2}\gamma^2,\;
\end{equation}
as well as, 
\begin{align}\label{Bgamma}
&B(q,\gamma)=q-\left(3+\frac{\kappa}{2}\right)\gamma+\kappa\gamma^2,\\
\label{Cgamma}
&C(p,\gamma)=- \frac{\kappa}{2}\gamma^2+\left(2+\frac{\kappa}{2}\right)\gamma - p,\\ \label{betaCgamma}
&\beta(\gamma):=\beta(p,\gamma):=\frac{\kappa}{2}\gamma^2-C(p,\gamma)=\kappa\gamma^2-\left(2+\frac{\kappa}{2}\right)\gamma +p,
\end{align}
where the last function, $\beta(p,\gamma)$, is the so-called ``spectrum function'' of Ref. \cite{DNNZ}. Recall also that this function  possesses an important duality property \cite{DNNZ},
\begin{equation}\label{duality}
\beta(p,\gamma)=\beta(p,\gamma'),\,\,\, \gamma+\gamma':=\frac{2}{\kappa}+\frac{1}{2}.
\end{equation}


\begin{remark} \label{remark46}The B--S  
 parameter $\gamma_0$, and bulk spectrum \eqref{bulk} $\beta_0:=\beta(p,\gamma_0)$, (corresponding to Eqs. (11)  and (12) in Ref. \cite{BS}) are obtained from the equations (see Ref. \cite{DNNZ}),
\begin{eqnarray}\label{BScondition1}
C(p,\gamma_0)=0;\,\,\,\,\beta_0=\beta(p,\gamma_0)=\kappa\gamma^2_0/2. \end{eqnarray}   
The two solutions to \eqref{BScondition1} are $\gamma_0^{\pm}(p)$ as in Eq. \eqref{p,q}, 
	where the lower branch $\gamma_0:=\gamma_0^{-}$ is the one selected for the bulk spectrum, $\beta_0(p)=\frac{1}{2}\kappa{\gamma_0^{-}}(p)^2$. 
	
This spectrum \eqref{bulk} is defined only to the {\it left} of a vertical line in the $(p,q)$ plane,  as given by (see Fig. \ref{droiteter})
\begin{equation}\label{Delta0}
\Delta_0 :=\left\{p= \frac{(4+\kappa)^2}{8\kappa}, q\in \mathbb R\right\}.
\end{equation}
 \end{remark} 

\begin{remark} The $\sigma$-dependent spectrum \eqref{betasigma+}  is obtained from the equations
\begin{align}
\label{gamma}
A^\sigma(p,\gamma)=0;\,\,\,\,\beta(p,\gamma)=\kappa\gamma^2/2-C(p,\gamma).
\end{align}
The solutions to Eq. \eqref{gamma} are
\begin{eqnarray}\label{gammasigma}
&&\gamma^\sigma_{\pm}(p)=\frac{1}{\kappa}\big(1\pm \sqrt{1-2\sigma\kappa p}\big),\\
\label{betapgammasigma}
&&\beta^\sigma_{\pm}(p)=(1-2\sigma)p-\frac{\kappa}{2}\gamma^\sigma_{\pm}(p)=(1-2\sigma)p -\frac{1}{2}\big(1\pm \sqrt{1-2\sigma\kappa p}\big).
\end{eqnarray}
The multifractal spectrum \eqref{betasigma+} is then given by the upper branch $\beta^\sigma_{+}(p)$ \cite{DNNZ}. Note also that this spectrum is defined only for $2 \sigma \kappa p \leq 1$, hence for points  in the $(p,q)$ plane {\it below} the oblique line (Fig. \ref{droiteter}):
\begin{equation}\label{Delta1}\Delta_1 :=\left\{(p,q) \in \mathbb R^2, q=p+\frac{1}{2\kappa}\right\}.
\end{equation}
\end{remark}
 \begin{figure}[htbp]
\begin{center}
\includegraphics[angle=0,width=.603290\linewidth]{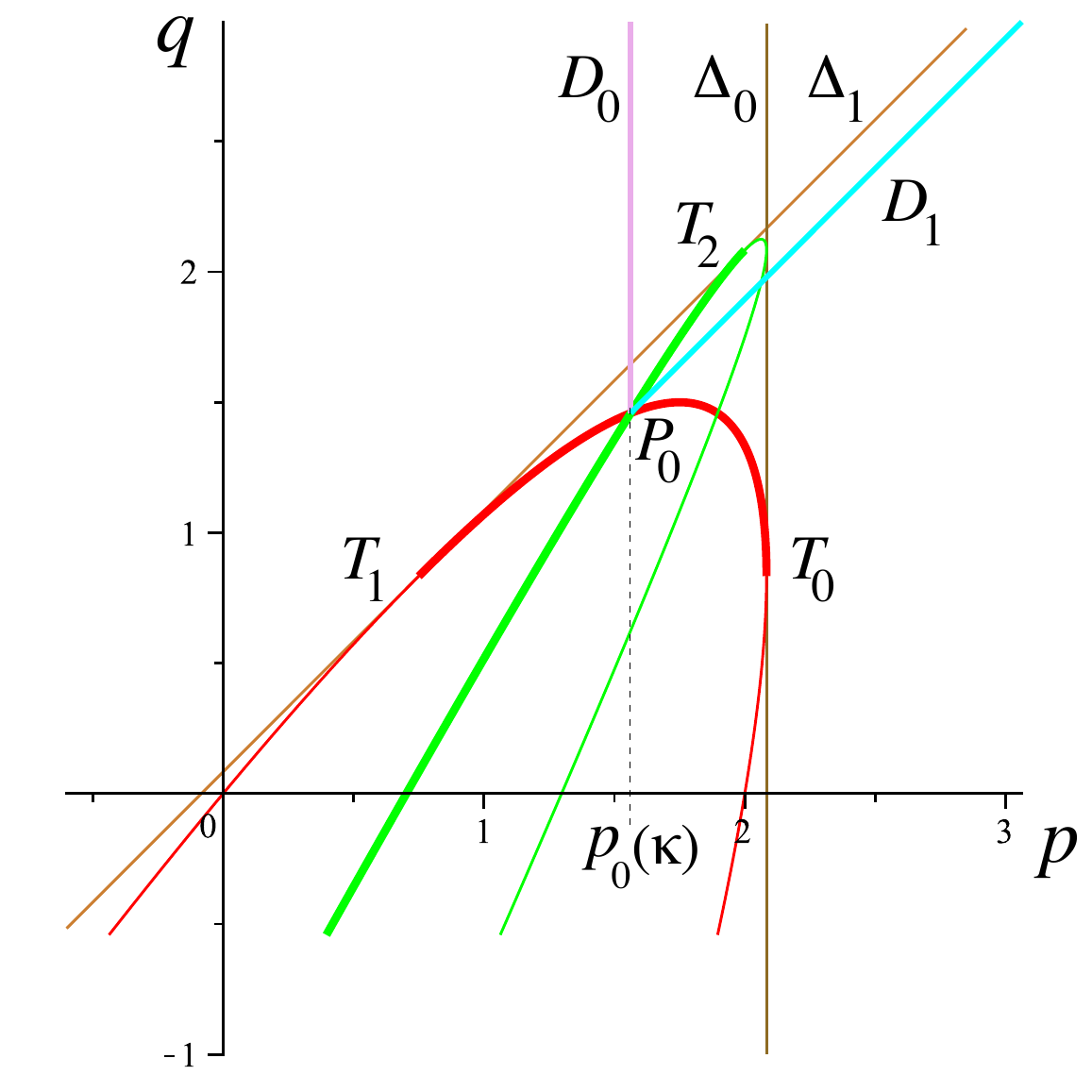}
\caption{{\it Red parabola $\mathcal R$  \eqref{C++} and green parabola $\mathcal G$ \eqref{C+dual} (for $\kappa=6$).  From the intersection point  $P_0$ \eqref{P0} originate the two (half)-lines $D_0$ \eqref{D0} and $D_1$ \eqref{D1}. The  bulk spectrum $\beta_0(p)$ and the generalized spectrum $\beta_1(p,q)$ coincide along the arc \eqref{2} of red parabola  between its tangency points $T_0$ and $T_1$ with $\Delta_0$  and $\Delta_1$ (thick red line). They also coincide along the infinite left branch \eqref{3}  of the green parabola, up to its tangency point $T_2$ to $\Delta_1$ (thick green line). The $\beta_0(p)$ spectrum and the  linear one $\beta_{\mathrm{lin}}(p)$ coincide along $D_0$, whereas $\beta_1(p,q)$ and $\beta_{\mathrm{lin}}(p)$ coincide  along $D_1$. 
}}
\label{droiteter}
\end{center}
\end{figure}   
\subsubsection{`Red' Parabola}
The parabola  $\mathcal R$ of Theorems  \ref{maintheorem1/2} and \ref{maintheorem}, which we shall hereafter call (and draw in) \textcolor{red}{\bf red} (see Fig. \ref{droiteter}), is given by the simultaneous conditions, 
\begin{equation}\label{seed1}
A^\sigma(p,\gamma)=A(p,q,\gamma)=0,\,\,C(p,\gamma)=0,
\end{equation} 
hence also $B(q,\gamma)=0$, which recovers the parametric form \eqref{para}
\begin{equation}\label{C++}\begin{split}
p&=p_{\mathcal R}(\gamma):=\left(2+\frac{\kappa}{2}\right)\gamma-\frac{\kappa}{2}\gamma^2,
\\ 
q&=q_{\mathcal R}(\gamma):=\left(3+\frac{\kappa}{2}\right)\gamma-\kappa\gamma^2,\,\,\gamma\in \mathbb R.
\end{split}\end{equation}
By construction, the associated spectrum $\beta(p,\gamma)$ is therefore both of the B--S type, $\beta_0^{\pm}(p)$, and of the novel type, $\beta^\sigma_{\pm}(p)$. We successively have:
\begin{align}\label{2.0}
&\gamma=\gamma^\sigma_{-}(p)=\gamma_0^{-}(p);\,\, \beta^\sigma_{-}(p)=\beta_0^{-}(p), \gamma\in\left(-\infty, 1/{\kappa}\right],\\ \label{2}
&\gamma=\gamma^\sigma_{+}(p)=\gamma_0^{-}(p);\,\, \beta^\sigma_{+}(p)=\beta_0^{-}(p), \gamma\in\left[1/{\kappa},{2}/{\kappa}+{1}/{2}\right],\\ \label{2bis}
&\gamma=\gamma^\sigma_{+}(p)=\gamma_0^{+}(p);\,\, \beta^\sigma_{+}(p)=\beta_0^{+}(p), \gamma\in\left[{2}/{\kappa}+{1}/{2},+\infty\right),
\end{align}
where the change of analytic branch from the first to the second line corresponds to a tangency at $T_1$ of the red parabola to the  boundary line $\Delta_1$, whereas the change from second to third corresponds to a tangency at $T_0$ to the vertical boundary line $\Delta_0$. 
The interval where the multifractal spectra coincide, i.e., when $\beta^\sigma_{+}(p)=\beta_0^{-}(p)$, is thus given by  line \eqref{2} in the equations above.
 
 In Cartesian coordinates, the red parabola $\mathcal R$ \eqref{C++} has for equation \eqref{paraRCart}. 
\subsubsection{`Green' Parabola} A second parabola  in the $(p,q)$ plane, hereafter called \textcolor{dgreen}{\bf green} (see Fig. \ref{droiteter}) and denoted by $\mathcal G$, is such that the multifractal spectra $\beta_0^-(p)$ and $\beta^\sigma_+(p)=\beta(p,q;\kappa)$ coincide on part of it. We use the  {\it duality} property \eqref{duality} of the spectrum function \cite{DNNZ}, and set the  simultaneous seed conditions, 
\begin{equation}\label{seed2}\begin{split}
&A^\sigma(p,\gamma')=A(p,q,\gamma')=0,\,\,C(p,\gamma'')=0,\\ 
&\gamma'+\gamma''={2}/{\kappa}+{1}/{2},
\end{split}
\end{equation}
where $\gamma'$ and $\gamma''$  are {\it dual} of each other and such that $\beta(p,\gamma')=\beta(p,\gamma'')$.
 
 Eqs. \eqref{A} and \eqref{C} immediately give the parametric form for the green parabola,
\begin{equation}\label{C+dual}\begin{split}
p&=p_{\mathcal G}(\gamma'):=\frac{(4+\kappa)^2}{8\kappa}-\frac{\kappa}{2}\gamma'^2,
\\ 
q&=q_{\mathcal G}(\gamma'):=\frac{(4+\kappa)^2}{8\kappa}+\gamma'-\kappa\gamma'^2,\,\,\gamma'\in \mathbb R.
\end{split}\end{equation}
Along this locus, we successively have:
\begin{align} \label{1}
&\gamma'=\gamma^\sigma_{-}(p), \gamma''=\gamma_0^{+}(p);\,\, \beta^\sigma_{-}(p)=\beta_0^{+}(p), \gamma'\in\left(-\infty, 0\right],\\ \nonumber
&\gamma'=\gamma^\sigma_{-}(p), \gamma''=\gamma_0^{-}(p);\,\, \beta^\sigma_{-}(p)=\beta_0^{-}(p), \gamma'\in\left[0,{\kappa}^{-1}\right],\\ \label{3}
&\gamma'=\gamma^\sigma_{+}(p),\gamma''=\gamma_0^{-}(p);\,\, \beta^\sigma_{+}(p)=\beta_0^{-}(p), \gamma'\in\left[{\kappa}^{-1},+\infty\right),
\end{align}
where the changes of branches correspond to a tangency of the green parabola to $\Delta_0$ followed by a tangency to $\Delta_1$. The multifractal spectra coincide when $\beta^\sigma_{+}(p)=\beta_0^{-}(p)$, which corresponds to the third line \eqref{3} in the equations above, i.e., to the domain where $\gamma'\geq 1/\kappa$. 
\subsubsection{Quadruple point}
The intersection of the red and green parabolae \eqref{C++} and \eqref{C+dual} can be found by combining the seed equations \eqref{seed1} and \eqref{seed2}. 
We find either $\gamma=\gamma'=1/\kappa+1/4$, or $\gamma=2/\kappa+1/4, \gamma'=-1/4$, 
which lead to the two intersection points,
\begin{align}\label{P0}
&P_0: p_0=p_0(\kappa)=\frac{3(4+\kappa)^2}{32\kappa},\,\,\quad q_0=\frac{(4+\kappa)(8+\kappa)}{16\kappa},\\ \label{P1}
&P_1: p_1=\frac{(8+\kappa)(8+3\kappa)}{32\kappa},\,\,\,\,\,\,\, \quad q_0=\frac{(4+\kappa)(8+\kappa)}{16\kappa}.
\end{align}
Note that these points have same ordinate, while the abscissa of the left-most one, $P_0$, is  $p_0(\kappa)$ \eqref{p0*}, where the integral means spectrum transits from the bulk form \eqref{bulk}  to its linear form \eqref{lin}.  
 
Through this intersection point $P_0$ further pass two important straight lines in the $(p,q)$ plane.
\begin{definition}\label{D0D1}
$D_0$ and $D_1$ are, respectively, the vertical line and the slope one line passing through point $P_0$, of equations
\begin{align}\label{D0}
&D_0:=\{(p,q): p=p_0 \},
\\ \label{D1}
&D_1:=\left\{(p,q): q-p=q_0-p_0=\frac{16-\kappa^2}{32\kappa} \right\}.
\end{align}
\end{definition}
\textcolor{black}{On $D_0$, one has $\beta_0(p,\kappa)=\beta_{\mathrm{lin}}(p,\kappa)$.} A key property of $D_1$ is the following. 
The difference,
\begin{equation}\label{beta-betalin}
\beta_1(p,q;\kappa)-\beta_{\mathrm{lin}}(p,\kappa)=\frac{1}{\kappa}\left(\frac{\kappa}{4}-\sqrt{1+2\kappa (p-q)}\right)^2,
\end{equation}
 is always positive, and vanishes only on line $D_1$, where  
\begin{equation}\label{fusionD1bis}
 \forall (p,q)\in  D_1,\,\,\, \beta_1(p,q;\kappa)=\beta_{\mathrm{lin}}(p,\kappa)=p-\frac{(4+\kappa)^2}{16\kappa}.
\end{equation}
 \begin{figure}[htbp]
\begin{center}
\includegraphics[angle=0,width=.603290\linewidth]{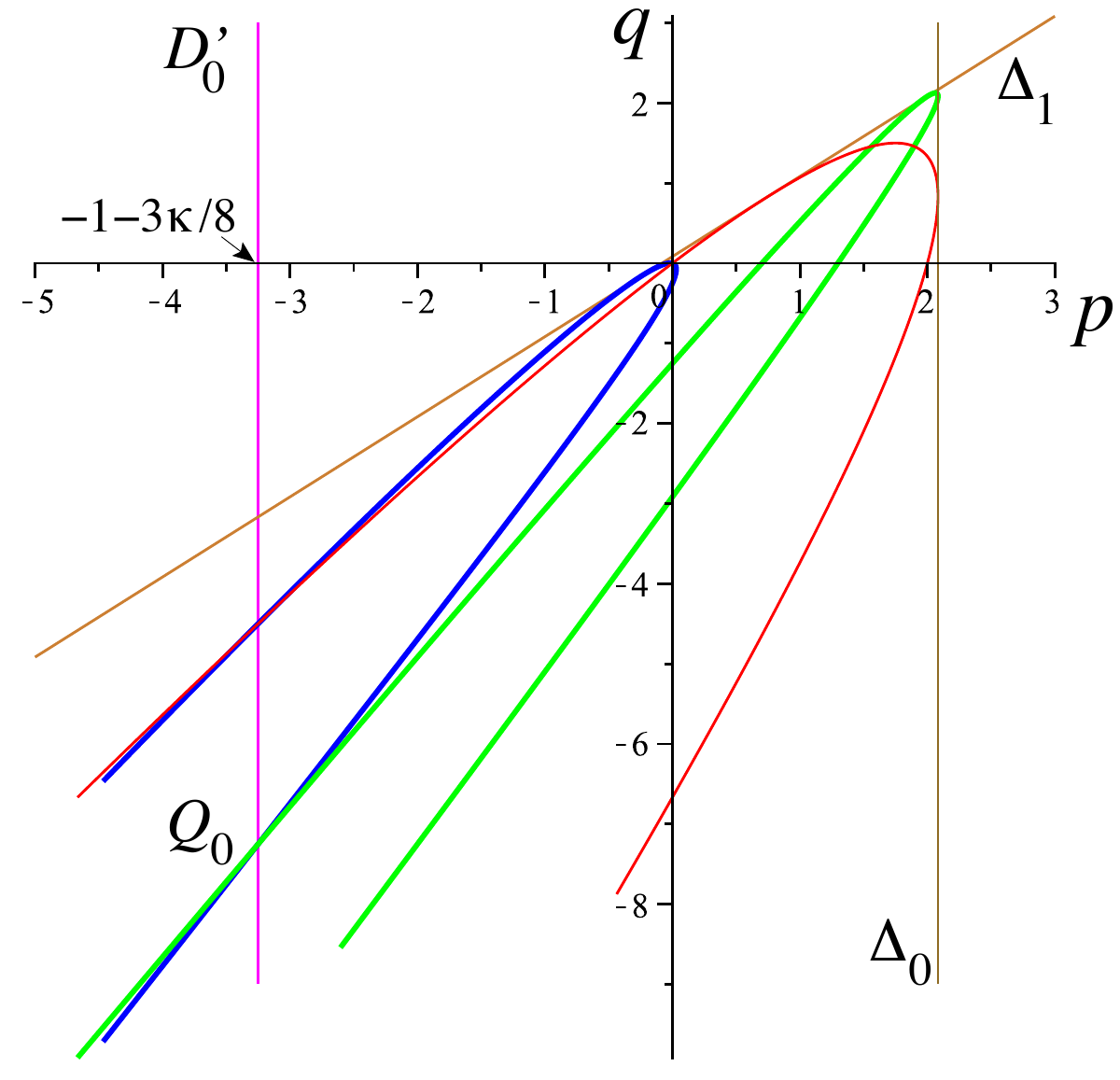}
\caption{{\it The blue quartic \eqref{pqquartic} for $\kappa=6$. It intersects the green parabola at point $Q_0$ \eqref{Q0} and the red parabola at point $Q_1$ \eqref{Q1} (not marked), both of  abscissa $p'_0(\kappa)=-1-3\kappa/8$.}}
\label{gauche}
\end{center}
\end{figure}   
\subsubsection{`Blue' Quartic} \label{bluequartic} A third locus, the \textcolor{blue}{\bf blue} quartic $\mathcal Q$, will also play an important role, that is where the tip-spectrum, $\beta_{\textrm{tip}}(p;\kappa)$ \eqref{tip}, coincides with the novel spectrum, $\beta^\sigma_+(p)=\beta_1(p,q;\kappa)$. The tip spectrum is given by $\beta_{\textrm{tip}}(p;\kappa)=\beta(p,\gamma_0)-2\gamma_0-1$, where $\gamma_0$ is solution to $C(p,\gamma_0)=0$ and such that the tip contribution is positive, $2\gamma_0+1 \leq 0$ \cite{BS,DNNZ}; this corresponds to the tip condition \eqref{p0**} \cite{BS}.  In the $(p,q)$ plane, this descibes the domain to the left of the straight line  $D_0'$ (Fig. \ref{gauche}),  defined by
\begin{equation}\label{D0'}
D_0':=\left\{(p,q): p=p'_0(\kappa)=-1-3\kappa/8 \right\}.
\end{equation}
The generalized spectrum is given by $\beta^\sigma_+(p)=\beta(p,\gamma)$ where $\gamma$ is solution to $A^\sigma(p,\gamma)=0$. We therefore look for simultaneous solutions to the seed equations, 
\begin{equation}\label{betatip}\begin{split}
&\beta(p,\gamma)=\beta(p,\gamma_0)-2\gamma_0-1,\,\,2\gamma_0+1\leq 0,\\ 
&A^\sigma(p,\gamma)=0,\,\,\,C(p,\gamma_0)=0.
\end{split}\end{equation}
Using Eq. \eqref{AA}, we first find, as for the red and green parabolae, 
\begin{equation}\label{q-pg}
q-p=\gamma-\frac{\kappa}{2}\gamma^2,
\end{equation}
and from  \eqref{betaCgamma} and \eqref{Cgamma}, by substitution in the above,
\begin{align}
\label{2p-q}
&2p-q+\frac{1}{2}=\frac{\kappa}{4}(\gamma+\gamma_0),\\ \label{quad}
&\frac{4+\kappa}{2}\gamma-\kappa\gamma^2-1=\frac{8+\kappa}{2}\gamma_0-\kappa\gamma_0^2.
\end{align}
Solving for $\gamma_0$ in terms of $\gamma$ gives
\begin{align}\label{gamma0gamma}
&\gamma_0=\gamma_0^{\pm}:=\frac{8+\kappa}{4\kappa}\pm\frac{1}{2\kappa}\Delta^{\frac{1}{2}}(\gamma),\\ \label{Deltagamma}
&\Delta(\gamma):=4\kappa^2\gamma^2-2\kappa(4+\kappa)\gamma+\frac{1}{4}(8+\kappa)^2+4\kappa,
\end{align}
with $\Delta(\gamma) >0, \forall \gamma\in \mathbb R$. 
The tip relevance inequality in \eqref{betatip}, $2\gamma_0+1\leq 0$, implies the choice of the negative branch in \eqref{gamma0gamma}:
$\gamma_0=\gamma_0^{-}.$ 
 We thus get the desired explicit parameterization of that branch of the quartic,
\begin{equation}\label{pqquartic}\begin{split}
&p=p_{\mathcal Q}(\gamma):=\frac{\kappa}{16}+\left(1+\frac{\kappa}{4}\right)\gamma-\frac{\kappa}{2}\gamma^2-\frac{1}{8}\Delta^{\frac{1}{2}}(\gamma),\\ 
&q=q_{\mathcal Q}(\gamma):=p_{\mathcal Q}(\gamma)+\gamma-\frac{\kappa}{2}\gamma^2,\,\,\,\gamma \in \mathbb R.
\end{split}\end{equation}
\begin{remark}\label{constitutive} Note that because of the very choice to parameterize the  parabolae and the quartic by $\gamma$,  such that $A$ \eqref{AA} vanishes, Eq. \eqref{q-pg} holds for each of the  pairs of parametric equations.\end{remark}
We successively have along the  branch \eqref{pqquartic} of the blue quartic:
\begin{align}\nonumber
&\gamma=\gamma^\sigma_{-}(p);\,\, \beta^\sigma_{-}(p)=\beta_{\textrm{tip}}(p), \gamma\in\left(-\infty, {1}/{\kappa}\right],\\  \label{3ter}
&\gamma=\gamma^\sigma_{+}(p);\,\, \beta^\sigma_{+}(p)=\beta_{\textrm{tip}}(p)<\beta_0^{-}(p), \gamma\in\left[{1}/{\kappa},1+{2}/{\kappa}\right],\\   \label{3bis}
&\gamma=\gamma^\sigma_{+}(p);\,\, \beta^\sigma_{+}(p)=\beta_{\textrm{tip}}(p)\geq \beta_0^{-}(p), \gamma\in\left[1+{2}/{\kappa},+\infty\right),
\end{align}
The intersection of the \textcolor{blue}{{\bf blue}} quartic \eqref{pqquartic} with the \textcolor{red}{{\bf red}} parabola $\mathcal R$ \eqref{C++} is located at
\begin{equation}\label{Q1}Q_1:  p'_0=-1-\frac{3\kappa}{8},q= -\frac{1}{2}(3+\kappa); \,\,\,\gamma=\gamma_0=-\frac{1}{2},
\end{equation} 
followed by a second intersection at the origin, 
$p=q=0,$ for $\gamma=\frac{2}{\kappa}$ and $\gamma_0=0$. 

The intersection of the \textcolor{blue}{{\bf blue}} quartic \eqref{pqquartic} 
with the  \textcolor{dgreen}{{\bf green}} parabola $\mathcal G$ \eqref{C+dual}
is located at  
\begin{equation}\label{Q0}
Q_0:  p'_0=-1-\frac{3\kappa}{8}, q'_0:=-2-\frac{7\kappa}{8};\,\,\, \gamma=\gamma'=1+\frac{2}{\kappa},\gamma_0=-\frac{1}{2}.
\end{equation} 
Notice that these two intersection points have same abscissae, $p'_0(\kappa)$ \eqref{p0**}, where the transition for $\gamma_0=-\frac{1}{2}$ from the bulk spectrum $\beta_0$ to the tip spectrum $\beta_{\mathrm{tip}}$  takes place. They are found by combining Eqs. \eqref{seed1}  or Eqs. \eqref{seed2}  with \eqref{betatip}. 

The tip spectrum and the generalized one coincide in both $\gamma$-intervals  \eqref{3ter} and \eqref{3bis}, which together parameterize the branch of the quartic located below its contact with $\Delta_1$ (see Fig. \ref{gauche}). 
Because of the tip relevance condition \eqref{p0**},  only the interval \eqref{3bis} describing the lower infinite branch of the  quartic located to the {\it left} of $Q_0$  will matter for  the integral means spectrum. 
\subsection{Whole-plane $\SLE_\kappa$ generalized spectrum}
\subsubsection{Phase diagram} Let us briefly summarize the results of Section \ref{PTLines}. We know from Eq. \eqref{2} that the bulk spectrum $\beta_0(p)$ and the mixed spectrum $\beta_1(p,q)$ coincide along the finite sector of  parabola $\mathcal R$ located between tangency points $T_0$ and $T_1$ (Fig. \ref{droiteter}). From Eq. \eqref{3}, we also know that they coincide along the infinite left branch of parabola $\mathcal G$ below the tangency point $T_2$ (Fig. \ref{droiteter}). 

 The linear bulk spectrum $\beta_{\mathrm{lin}}(p)$ coincides with $\beta_0(p)$ along line $D_0$ and supersedes the latter to the right of $D_0$ (Fig. \ref{droiteter}). We  know from  \eqref{fusionD1bis} that $\beta_{\mathrm{lin}}(p)$ and $\beta_1(p,q)$ coincide along the line $D_1$ (Fig. \ref{droiteter}). 
 
 The tip spectrum $\beta_{\mathrm{tip}}(p)$ coincides with $\beta_0(p)$ along line $D_0'$, and  supersedes it to the left of $D_0'$. We finally know from Eq. \eqref{3bis}  that this tip spectrum $\beta_{\mathrm{tip}}(p)$ coincides with $\beta_1(p,q)$ along the lower branch of the blue quartic located below point $Q_0$ \eqref{Q0} (Fig. \ref{gauche}).

 The only possible scenario which thus emerges to construct the average generalized integral means spectrum by a continuous matching of the 4 different spectra along the phase transition lines described above, is the partition of the $(p,q)$ plane in 4 different regions as indicated in Fig. \ref{separ}: \begin{itemize}
  \item a  part (I) to the left of $D_0'$ and located above the blue quartic up to point $Q_0$, where the average integral means spectrum is $\beta_{\mathrm{tip}}(p)$;
   \item an upper part (II) bounded by lines $D_0'$, $D_0$, and located above the section of the green parabola between points $Q_0$ and $P_0$, where the spectrum is given by  $\beta_0(p)$; 
 \item an infinite wedge (III) of apex $P_0$ located between the upper half-lines $D_0$ and $D_1$, where the spectrum is given by  $\beta_{\mathrm{lin}}(p)$; 
  \item a lower part  (IV) whose boundary is the blue quartic up to point $Q_0$, followed by the arc of green parabola between points $Q_0$ and  $P_0$, followed by the half-line $D_1$ above $P_0$ where the spectrum is $\beta_1(p,q)$.
 \end{itemize}
The two wings $T_1P_0$ and $P_0T_0$ of the red parabola (Fig. \ref{droiteter}), where we know from Theorem \ref{maintheorem} that the 
average spectrum is given by $\beta_0(p)=\beta_1(p,q)$, can thus be seen as the respective extensions of region IV into  II and of region II into IV. The validity of this geometrical analysis of the phase diagram of Figure \ref{separ}, associated with the generalized integral means spectrum of whole-plane $\mathrm{SLE}_\kappa$, is established in Theorem \ref{propseparbis}. 
\subsubsection{The B--S line}\label{BSline}
As mentioned above, 
the whole-plane SLE case studied by Beliaev and Smirnov corresponds  to the $q=2p$ line.   Because of Eq. \eqref{paraRCart}, it intersects the red parabola $\mathcal R$ only at $p=0$. The green parabola $\mathcal G$ \eqref{C+dual} has for Cartesian equation, 
\begin{equation}
\label{cartesiangreen}
\frac{\kappa}{2}(2p-q)^2-\frac{1}{8}(4+\kappa)^2\left(2p-q\right)+p+\frac{1}{128}(4+\kappa)^2(8+\kappa)=0,
\end{equation}
which shows that it intersects the B--S line at \cite{BDZ}
 \begin{equation}\label{pseconde0}
 p=p_0''(\kappa):=-\frac{1}{128}(4+\kappa)^2(8+\kappa),
 \end{equation} which is to the {\it left} of the tip transition line at $p_0'(\kappa)=-1-\frac{3}{8}\kappa$ \eqref{p0**}. 
The quartic $\mathcal Q$ \eqref{pqquartic} obeys 
\begin{align}\nonumber
\left[\left(2p-q-\frac{\kappa}{16}\right)^2-\frac{c}{4}\right]\left(2p-q-1-\frac{\kappa}{8}\right)\left(2p-q\right)
&=\frac{\kappa}{2}(p-q)\left(2p-q-\frac{1}{4}-\frac{\kappa}{8}\right)^2\\ \label{cartesianquartic}
c&=c(\kappa):=\frac{1}{64}(8+\kappa)^2+\frac{\kappa}{4},
\end{align}
which immediately shows that  the B--S line $q=2p$ intersects $\mathcal Q$ only at the origin and stays above its lower branch.

The B--S line therefore does not intersect the segment of green parabola $\mathcal G$ between $P_0$ and $Q_0$, nor the  quartic $\mathcal Q$ below $Q_0$ (Fig. \ref{separ}). Thus, as shown rigorously in Ref. \cite{BDZ}, the novel spectrum $\beta_1$ {\it does not} a priori appear in the version of whole-plane SLE considered in Ref. \cite{BS}. The B--S line nevertheless intersects $\mathcal G$ at $p''_0$ \eqref{pseconde0}  to the left of $Q_0$, in a domain lying above the quartic and where the integral mean receives a non-vanishing contribution from the SLE tip. But if that integral mean is restricted to avoid a neighborhood of $z=1$, whose image is the tip, only the bulk spectrum remains, and a phase transition will take place from $\beta_0$ to $\beta_1$ when the line $q=2p$ crosses $\mathcal G$, as shown in Ref. \cite{BDZ}. As we shall see in Section \ref{mfoldspectrum}, the $\beta_1$ spectrum can also directly appear in the averaged integral means spectra of higher $m$-fold transforms of the B--S version of whole-plane SLE.
\subsubsection{The Koebe $\kappa \to 0$ limit}\label{koebe} In this limit, \textcolor{black}{the whole-plane SLE map tends to the Koebe function.} Eq. \eqref{paraRCart} for the red parabola $\mathcal R$ becomes $3p-2q=0$, Eq. \eqref{cartesiangreen} for the green one $\mathcal G$ becomes $3p-2q-1=0$, and Eq. \eqref{cartesianquartic} for the quartic factorizes into that of four parallel lines, among which $q=2p$ gives the relevant lower branch. Point $P_0$ moves up to infinity, whereas $Q_0\to (-1,-2)$. The phase diagram is thus made of only three different regions, I, where $\beta_{\mathrm{tip}}(p)=-p-1$, II, where $\beta_0=0$, and IV, where $\beta_1(p,q)=3p-2q-1$.
\subsection{Proof of Theorem \ref{propseparbis}}\label{proofs}
 \begin{figure}[htbp]
\begin{center}
\includegraphics[angle=0,width=.73290\linewidth]{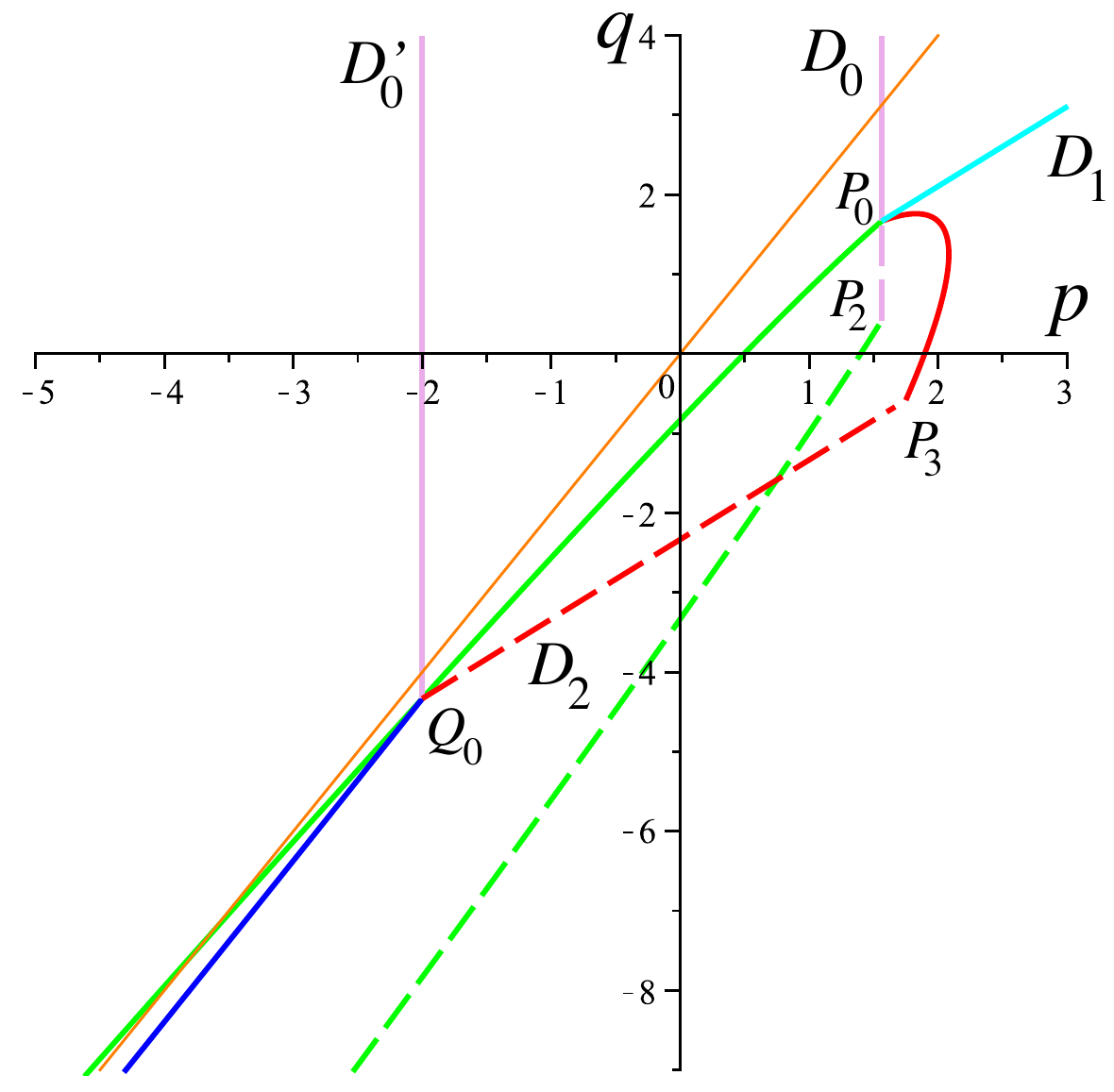}
\caption{{\it \textcolor{black}{(Case $\kappa=8/3$ shown here.)} \textcolor{black}{Various geometrical elements appearing in Section \ref{proofs}, when establishing the respective domains of validity of integral means spectra $\beta_{\mathrm{tip}}(p)$, $\beta_0(p)$, $\beta_{\mathrm{lin}}(p)$, and $\beta_1(p,q)$. \textcolor{black}{The green parabola $\mathcal G$ intersects $D_0$ at points $P_0$ \eqref{P0} and $P_2=\big(\frac{3(4+\kappa)^2}{32\kappa},\frac{4+\kappa}{16}\big)$}. The dashed red line $D_3$ corresponds to Eq. \eqref{D_3}, and intersects $D'_0$ at $Q_0$ and the red parabola $\mathcal R$ at point \textcolor{black}{$P_3=(1+\frac{2}{\kappa},\frac{4-\kappa^2}{2\kappa})$.} The $q=2p$ continuous straight line in coral, corresponding to the whole-plane SLE version of Ref. \cite{BS}, does not intersect the blue quartic, but intersects the green parabola at a point of abscissa \eqref{pseconde0}. 
}}}
\label{dfk}
\end{center}
\end{figure}   
\textcolor{black}{Let us first consider the following geometrical elements (Fig. \ref{droiteter}). The red parabola $\mathcal R$ \eqref{C++}  partitions the half-plane below $\Delta_1$ into an open interior $\mathcal I$ of $\mathcal R$,  an open exterior $\mathcal E_-$ located to the left of tangency point $T_1$, and  an open exterior $\mathcal E_+$ to the right of $T_1$. 
\begin{lemma}\label{lemmefondamental}
The average generalized integral means spectrum $\beta(p,q)$ of whole-plane SLE is bounded below as $\beta(p,q)\geq \beta_1(p,q)$  in $ \mathcal E_-\cup\mathcal I$, whereas  $\beta(p,q)\leq \beta_1(p,q)$ in $\mathcal E_+$.
\end{lemma}
\begin{proof}
The red parabola $\mathcal R$  is parameterized by $\gamma$ in \eqref{C++} such that $C(p,\gamma)=0$, where owing to \eqref{2}, \eqref{2.0} and \eqref{2bis}, $\gamma=\gamma^\sigma_-(p)$ before the tangency point $T_1$, and $\gamma=\gamma^\sigma_+(p)$ after it (Fig. \ref{droiteter}). In the open interior $\mathcal I$ of $\mathcal R$,  $C(p,\gamma^\sigma_+)>0$ and $C(p,\gamma^\sigma_-)<0$; in the  open exterior $\mathcal E_-$ to the left of the tangency point $T_1$,  $C(p,\gamma^\sigma_{\pm})>0$; in the open exterior $\mathcal E_+$ to the right of $T_1$,  $C(p,\gamma^\sigma_{\pm})<0$. 
According to \cite[Section {4.2.5}]{DNNZ}, and the generalization thereof to PDE \eqref{eq1Fzzbis}, there exists then in $ \mathcal E_-\cup \mathcal I$ a supersolution to Eq. \eqref{eq1Fzzbis} of critical exponent \eqref{betaCgamma} $\beta(\gamma^\sigma_+)=\beta^\sigma_+(p)=\beta_1(p,q)$, such that the average integral means spectrum $\beta(p,q)$ is bounded below as $\beta(p,q)\geq \beta_1(p,q)$, whereas there exists in $\mathcal E_+$ a subsolution to \eqref{eq1Fzzbis} with the same critical exponent $\beta(\gamma_+^\sigma)=\beta_1(p,q)$, such that now $\beta(p,q)\leq \beta_1(p,q)$.
\end{proof}}
 \subsubsection*{i) Linear spectrum}\label{linspec}
\textcolor{black}{\begin{proposition}\label{theo:lin}Consider in the $(p,q)$ plane the upward wedge-like domain delimited by lines $D_0$ and $D_1$ interecting at $P_0$ (Figs. \ref{separ} and \ref{dfk}). In this domain, the average generalized integral means spectrum has the {\it linear} form $\beta_{\mathrm{lin}}$ \eqref{lin}.
\end{proposition}} 
\begin{proof}For this purpose, we consider the convex function,
\begin{equation}\label{convex} 
\mathfrak b(p,q):=\beta(p,q)-\beta_{\mathrm{lin}}(p,q),
\end{equation}
where $\beta$ is the true average generalized integral means spectrum. As all integral means spectra, $\beta(p,q)$, hence $\mathfrak b(p,q)$,  are convex  functions in $\mathbb R^2$, as can be seen by using H\"older's inequality (see also Section \ref{univ}). We know that the function $\mathfrak b$ \eqref{convex} is equal to $0$ on the boundary part of  the wedge sector along $D_0$, since $\beta=\beta_0=\beta_{\mathrm{lin}}$ there \textcolor{black}{(recall Definition \ref{D0D1})}.  On the other hand, \textcolor{black}{from Lemma \ref{lemmefondamental}}, we know that on the boundary part of  the same sector along $D_1$,  $\beta\leq \beta_1=\beta_{\mathrm{lin}}$, so that $\mathfrak b\leq 0$ there. We now use the Beliaev--Smirnov result asserting that $\mathfrak b=0$ on the intersection of the sector with the  B--S line $q=2p$. We will then conclude with the maximum principle satisfied by the convex function $\mathfrak b$ in the wedge sector.  To be precise, let us consider the bounded triangular domain, intersection of the wedge with the half-plane $q\leq M$, with $M$ a large positive number. We already know that  $\mathfrak b\leq 0$ on the vertical and oblique parts of the triangle's boundary. On the horizontal part, we know that at the end-points the convex function satisfies $\mathfrak b$ $\leq 0$, and is equal to $0$ at the interior point of intersection of the horizontal side with the line $q=2p$. It follows that $\mathfrak b$ is identically  $0$ on that side. We can then apply the maximum principle in the interior of the bounded triangle to conclude that $\mathfrak b=0$ there; letting $M\to+\infty$ gives the result in the whole domain.
\end{proof}
\subsubsection*{ii) Duplantier--Hastings--Beliaev--Smirnov spectrum}\label{BSgreenplum}
\textcolor{black}{\begin{proposition}\label{theo:DHBS}Consider the infinite domain of the $(p,q)$ plane above the infinite upper branch \eqref{3} of parabola $\mathcal G$ \eqref{seed2} located below point $P_0$, and to the left of the half-line $D_0$ above $P_0$ (Fig. \ref{dfk}). 
The average generalized integral means spectrum $\beta(p,q)$ is given in this domain by the standard bulk spectrum $\beta_0(p)$ for $p\geq -1-3\kappa/8$  or by  the tip version $\beta_{\mathrm{tip}}(p)$ in the opposite case.
\end{proposition}}
\begin{proof}We use the analysis performed in Ref. \cite{DNNZ}, Section 4.2.4, in the proof of Proposition 4.1 there.
Two important quantities are defined as
\begin{align}\label{a0}
a(p):=&\gamma_0(p)-\gamma_+^\sigma(p)=\gamma_0(p)-\frac{1}{\kappa}-\frac{1}{\kappa}\sqrt{1-2\sigma\kappa p},\\
\label{b0}
b(p):=&\gamma_0(p)-\gamma_-^\sigma(p)=\gamma_0(p)-\frac{1}{\kappa}+\frac{1}{\kappa}\sqrt{1-2\sigma\kappa p},
\end{align}
where $\gamma_0(p)=\gamma_0^-(p)$  is the B--S parameter of Eq. \eqref{p,q},
\begin{equation}\label{gamma0}
\gamma_0^{\pm}(p)= \frac{1}{2\kappa}\left( 4+\kappa \pm \sqrt{(4+\kappa)^2-8\kappa p} \right),\,\,\, p \leq \frac{(4+\kappa)^2}{8\kappa}.
\end{equation} The analysis in Refs. \cite{BS} and \cite{DNNZ} was originally meant to apply to $\sigma=\pm 1$, but it identically goes through for $\sigma \in \mathbb R$. It is shown there  that two (necessary and sufficient) conditions for the validity of the B--S method, hence of the corresponding spectra $\beta_0$ and $\beta_{\mathrm{tip}}$, are respectively:
\begin{align}\label{a0b0}
\frac{1}{2}-a(p)-b(p)\geq 0,\,\,\,\,\,\,\,
\frac{1}{2}-b(p)\geq 0.
\end{align}
The first condition is independent of $\sigma$ and simply amounts to $p\leq p_0(\kappa)=\frac{3(4+\kappa)^2}{32\kappa}$, i.e., to points $(p,q)$ belonging to the half-plane located to the left of $D_0$. The second condition can be rewritten as the inequality,
\begin{equation}\label{ineqb0}
\frac{1}{2}+\frac{2}{\kappa}-\gamma_0(p)-\gamma_+^\sigma(p)\geq 0.
\end{equation}
The equality case corresponds to the duality condition \eqref{duality} obeyed by $\gamma_0(p)$ and $\gamma_+^\sigma(p)$, as in Eq. \eqref{seed2}, the solution of which is precisely given by the branch \eqref{3} of the green parabola below $P_0$. Since  at the origin, $a_0(0)=b_0(0)=-1/\kappa$, the inequalities \eqref{a0b0} hold together in the unbounded domain to the left of the union of the left branch of $\mathcal G$ below $P_0$ and of the half-line $D_0$ above $P_0$. We thus conclude that the average generalized integral means spectrum $\beta(p,q)$ is given in this domain by $\beta_0(p)$ for $p\geq -1-3\kappa/8$  or by  $\beta_{\mathrm{tip}}(p)$ in the opposite case. 
\end{proof}
\subsubsection*{iii) Inside the green parabola}\label{ugp}

 \textcolor{black}{In Ref. \cite{BDZ}, a new method of proof was given, which establishes the validity of the D--H--BS spectrum along the $\sigma=+1$ line $q=2p$, below the transition point \eqref{pseconde0}. There $\beta_0(p) \leq \beta_1(p, 2p)\leq \beta_{\mathrm{tip}}(p)$, and the integral means spectrum is still given by  $\beta_{\mathrm{tip}}$; however avoiding the tip neighborhood in the integral mean reveals the existence of a phase transition to the $\beta_1$ spectrum. We now generalize this method to $\sigma \in \mathbb R$. 
 \begin{definition}\label{def:D}
 The domain $\mathcal D$, defined by the inequalities,
\begin{align}\label{a0b0bis}
0<\frac{1}{2}-a(p)-b(p),\,\,\,\,\,\,\,
 0 < b(p)-\frac{1}{2}< \frac{2}{\kappa}, 
 \end{align}
is the interior part of the green parabola $\mathcal G$ located to the left of $D_0$ \eqref{D0}.
 \end{definition}}
 \textcolor{black}{For general $\sigma$ in the $(p,q)$ plane, the first inequality in \eqref{a0b0bis} corresponds to the left of $D_0$, while the second set corresponds to the interior of the $\mathcal G$ parabola. Indeed, Eqs. \eqref{b0} and \eqref{gamma0} show that $b=(\gamma_0^-+\gamma^\sigma_+)-2/\kappa=1+4/\kappa-(\gamma_0^++\gamma^\sigma_-)$. Thus, the r.h.s. of \eqref{a0b0bis} is equivalent to $\gamma_0^-+\gamma^\sigma_+ \geq 1/2+2/\kappa$ and $\gamma_0^++\gamma^\sigma_->1/2+2/\kappa$. The first inequality is saturated on the top branch \eqref{3} of the green parabola \eqref{seed2} (the phase transition line), while the second one is saturated on the lower branch \eqref{1}. The domain $\mathcal D$ is closed to the right by the segment of line $D_0$ in between the intersection points $P_0$ \eqref{P0} and $P_2:=\big(\frac{3(4+\kappa)^2}{32\kappa},\frac{4+\kappa}{16}\big)$ of $D_0$ with the green parabola (Fig. \ref{dfk}).
 \begin{proposition}\label{theo:DD}
 In domain $\mathcal D$, the average generalized integral means spectrum is given by $\beta(p,q)=\max \{\beta_{\mathrm{tip}}(p),\beta_1(p,q)\}$. 
 \end{proposition}}
\begin{corollary}\label{wedge} 
\textcolor{black}{Since in $\mathcal D$ the blue quartic $\mathcal Q$ below $Q_0$ is the separatrix for $\beta_{\mathrm{tip}}(p)=\beta_1(p,q)$, the average generalized integral means spectrum is given  by $\beta_{\mathrm{tip}}(p)$ in the infinite thin wedge $\mathcal W$  in $\mathcal D$ of apex $Q_0$, located to the left of line $D'_0$, inbetween the green parabola $\mathcal G$ and the blue quartic $\mathcal Q$, and by $\beta_1(p,q)$ in the remaining part $\mathcal D\setminus \mathcal W$ of $\mathcal D$ (Fig. \ref{dfk}).}
\end{corollary}
\begin{proof} 
\textcolor{black}{The extension to $\sigma \in \mathbb R$ of the proof given in Ref. \cite{BDZ} for $\sigma=+1$ and for $|z|>1$, involves, now for $z\in \mathbb D$ and $u:=|1-z|^2$, the function (see \cite[Eq. (4.1)]{BDZ}):
\begin{equation}\label{eq:psifin}
\psi:=\varsigma \psi_0+\psi_1:= \varsigma g_0(u)u^{\gamma_0} (1-z\bar z)^{-\beta_0} + u^{\gamma_1}(1-z\bar z)^{-\beta_1},
\end{equation}
where we set $\gamma_0=\gamma_0^{-}(p)$ \eqref{gamma0} with $\beta_0:=\beta_0(p)$ \eqref{bulk}, together with $\gamma_1:=\gamma^\sigma_+(p)$ \eqref{gammasigma} and $\beta_1=\beta^\sigma_+(p)=\beta_1(p,q)$ \eqref{betapgammasigma}. The function $g_0$ is a peculiar combination  of two hypergeometric functions, parameterized by $a$ \eqref{a0} and $b$ \eqref{b0}, with $\varsigma:=\sign g_0(0)$ (see \cite[Eqs. (3.5)--(3.7)]{BDZ}),
\begin{equation}\label{eq:g0}
g_0(u)=C_0\,\F(a,b,c,{u}/4)-C_0'\, ({u}/4)^{1/2-a-b}\F(a',b',c',{u}/4),
\end{equation} 
where
\begin{align}
\label{eq:abc}
&a=\gamma_0-\gamma_1=\gamma_0-\gamma^\sigma_+, \quad b=\gamma_0-\gamma_-^\sigma, \quad c=\frac{1}{2}+a+b,\\ \nonumber
&a'=\frac{1}{2}-a, \quad b'=\frac{1}{2}-b, \quad c'=\frac{1}{2}+a'+b'.
\end{align}
The hypergeometric functions are singular at ${u}=4$, but  the coefficients $C_0$ and $C_0'$ are  chosen in such a way that the function $g_0$ is  smooth at $z=-1$ \cite{BS}:
\begin{equation}
\label{eq:C0}
g_0(0)=C_0=\frac{\Gamma(3/2-a-b)}{\Gamma(1/2-a)\Gamma(1/2-b)},\,\,\,C_0'=\frac{\Gamma(c)}{\Gamma(a)\Gamma(b)},
\end{equation}
such that near $u=4$,
$g_0(u)= \frac{1}{\sqrt{\pi}}\left(\frac{1}{2}-a-b\right)+ O(4-u).$
This, together with the condition $1/2-a-b>0$ in $\mathcal D$, insures that the function $g_0$ stays bounded on the whole interval $u\in[0,4]$.}

\textcolor{black}{One further defines 
\begin{equation}\label{elldelta}
\ell_\delta:=(-\log(1-z\bar z))^\delta,\,\,\, \delta \in \mathbb R .
\end{equation} As in the original Ref. \cite{BS}, the proof then consists in showing that $\psi \ell_\delta$ provides a positive sub- or supersolution of the differential operator \eqref{eq1zz}, depending on the sign of $\delta$.}

\textcolor{black}{The key steps are Lemmas 4.1 and 4.2 in Ref. \cite{BDZ}. {\it Mutatis mutandis}, these Lemmas and their proofs are as follows for general $\sigma$ in the $(p,q)$ plane. 
\begin{lemma}\label{lemma4.1}
For $(p,q)\in \mathcal D$, there is $r_0<1$ such that $\psi=\varsigma\psi_0+\psi_1>0$ for all $z$ such that $r_0<|z|<1$.
\end{lemma}
\begin{proof}
According to Ref. \cite[Lemma 4.1]{BDZ}, the condition for this to hold is $\beta_1>\beta_0$; this is satisfied in the intersection of the interior domains of the green parabola $\mathcal G$ and of the red parabola $\mathcal R$, and $\mathcal D$ is contained in this intersection.
\end{proof}
Along the top branch \eqref{3} of the green parabola \eqref{seed2} (the onset of the phase transition), one has $b=1/2$, and  this precisely corresponds to $g_0(0)=0$ in Eq. \eqref{eq:C0}. Actually, this \textcolor{black}{vanishing}, which is critical to the proof below,  also happens for  half-integer values  $b=n+1/2, n\in \mathbb N$.  Proposition \ref{pr:tkn} below shows that there indeed exist a finite set of integers $ \mathcal J_\kappa$, and a discrete set $\mathcal T_\kappa$ of parabolic trajectories in the $(p,q)$-plane, where  this is realized. We need the following 
\begin{definition}\label{tkappa}
For $n\in \mathbb N$, define $\mathcal P_{n}$ as the parabola given in parametric equations by
\begin{equation}\label{Pn}\begin{split}&p=p_n(\gamma)=-\frac{\kappa}{2}\gamma^2+\kappa n\gamma+\frac{\kappa}{2}\left(\frac{(4+\kappa)^2}{4\kappa^2}-n^2\right),\,\,\,\gamma\in \mathbb R,\\
&q=q_n(\gamma)=-\kappa\gamma^2+(\kappa n+1)\gamma+\frac{\kappa}{2}\left(\frac{(4+\kappa)^2}{4\kappa^2}-n^2\right); 
\end{split}
\end{equation}
note that for $n=0$, $\mathcal P_0=\mathcal G$. Define then $\mathcal P^+_{n}$ as the branch corresponding to $\gamma\geq \max\{n,\kappa^{-1}\}$. 
\end{definition} For $n=0$, $\mathcal P^+_0$ coincides with the upper branch of the green parabola $\mathcal G$ stemming from its   tangency point to $\Delta_1$. For $\kappa n<1$, the branch originates from the point of tangency of $\mathcal P_n$ to $\Delta_1$ \eqref{Delta1}, whereas for $\kappa n>1$, it originates from the point of tangency to $\Delta_0$ \eqref{Delta0}. For $\kappa n=1$, it degenerates into the bisector of $\Delta_0$ and $\Delta_1$, $q=2p-p(\kappa)$ with $p(\kappa)={(6+\kappa)(2+\kappa)}/{8\kappa}$. (See Fig. \ref{tram}.)
\begin{figure}[htbp]\label{tram}
\begin{center}
\includegraphics[angle=0,width=.63290\linewidth]{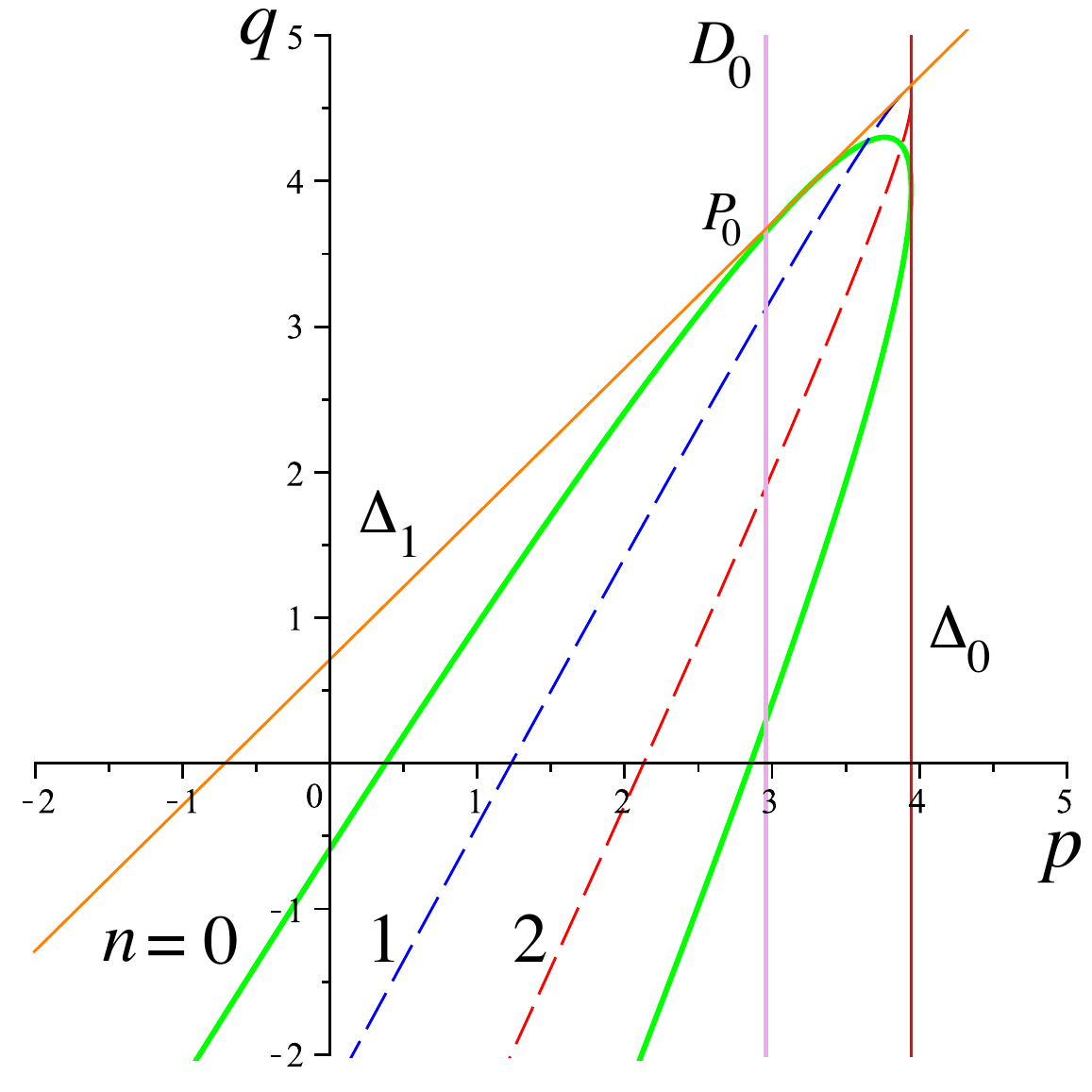}
\caption{{\it \textcolor{black}{Green parabola $\mathcal G=\mathcal P_0$, and parabolic branches $\mathcal P^+_n$ for $n\in \{0,1,2\}$, and $\kappa=0.7$. They cross into domain $\mathcal D$ at the intersection points \eqref{eq:pn} with $D_0$. 
}}}
\end{center}
\end{figure}
\begin{lemma}\label{pr:tkn}
The set of equations, $b=n+\frac{1}{2}, n\in \mathbb N$, where $g_0(0)=0$, is realized on the set of parabola branches $\mathcal P^+_{n}$ as defined in Def. \ref{tkappa}. Inside domain $\mathcal D$,  this yields the set $\mathcal T_\kappa:=\cup_{n \in \mathcal J_\kappa} \mathcal P^+_{n}\cap\mathcal D$, where $\mathcal J_\kappa:=\{n\in \mathbb N, 0\leq n\leq \lfloor2\kappa^{-1}\rfloor\}$. 
\end{lemma}
\begin{proof} 
From definition \eqref{b0} or \eqref{eq:abc} of $b$, we have 
$$b=\gamma_0^-+\gamma_+^\sigma-\frac{2}{\kappa}=\gamma_0+\gamma_1-\frac{2}{\kappa},$$
and we get  the condition $\gamma_1+\gamma_0=n+\frac{1}{2}+\frac{2}{\kappa}$, together with, as in \eqref{seed2}, $A^\sigma(p,\gamma_1)=0,C(p,\gamma_0)=0$. This yields the parabola equations \eqref{Pn} for $\mathcal P_n$, where $\gamma=\gamma_1$.  The root determinations $\gamma_1=\gamma_+^\sigma$ and $\gamma_0=\gamma_0^-$ then determine the range $\gamma=\gamma_1>\max\{n,1/\kappa\}$ of branch $\mathcal P_n^+$. The restriction on the values of $n$ in $\mathcal J_\kappa$ directly comes from definition \eqref{a0b0bis} of $\mathcal D$.
\end{proof}
\begin{remark}
Note that  for $\kappa >2$, $\mathcal T_\kappa$ is reduced to the $n=0$ phase transition line on $\mathcal G$, whereas positive values of $n$ exist in $\mathcal J_\kappa$ only for $0<\kappa\leq 2$. 
\end{remark}
\begin{remark}
The parabolic branch $\mathcal P^+_n$ intersects $D_0$ for $\gamma= \gamma^{(n)}:=n+\frac{4+\kappa}{4\kappa}$,  at 
\begin{equation} 
\label{eq:pn}p_n=p_0=\frac{3}{32\kappa}(4+\kappa)^2, q_n=\frac{(4+\kappa)(8+\kappa)}{16 \kappa}-\frac{\kappa}{2}n\left(n+\frac{1}{2}\right),
\end{equation}
its intersection with $\mathcal D$ corresponding to the range $\gamma> \gamma^{(n)}$.
\end{remark}
\begin{remark} The  parabolic branch $\mathcal P^+_n$ intersects the B--S line $q=2p$ only when $0\leq n\leq \lfloor\kappa^{-1}\rfloor$, and their intersection is given by 
\begin{equation} 
\label{eq:tkn}
 p''_{n}(\kappa):=-\frac{(1+2n)(8+\kappa-2n\kappa)(4+\kappa+2n\kappa)(4+\kappa-2n\kappa)}{128(1-n\kappa)^2},
\end{equation}
in agreement with Ref. \cite[Eq. (3.18)]{BDZ}. The $n=0$ case corresponds, for any value of $\kappa$, to the phase transition point $p''_0(\kappa)$ \eqref{pseconde0}, as studied in Ref. \cite{BDZ}.
\end{remark}
\begin{lemma}\label{subsupD}
For $(p,q)\in \mathcal D$ and $(p,q) \notin \mathcal T_\kappa$, there is $r_0<1$ such that  $\mathcal P(D) [\psi \ell_{\delta}]$ for $\psi$ \eqref{eq:psifin} has a constant sign  in the annulus $ r_0< |z|<1$, which depends only on that of $\delta$.
\end{lemma}}
\begin{proof}
\textcolor{black}{We follow the proof of the similar \cite[Lemma 4.2]{BDZ}, which distinguishes three cases according to the relative rates at which  $r:=|z| \to 1$ and $u=|1-z|^2\to 0$, with $1-r\leq u^{1/2}$. In case I, $u$ is bounded away from zero, and the only requirement is $\beta_0(p)< \beta_1(p,q)$, which here holds in domain $\mathcal D$.  In case II, one assumes that $(1-r)^{2-\varepsilon}<u<u_0$, with $\varepsilon >0$, and $u_0>0$ chosen such that $\varsigma g_0(u) >0$ for $0<u<u_0$. This $u_0$ exists here thanks to $(p,q) \notin \mathcal T_\kappa$ and Lemma \ref{pr:tkn}. In case III, one assumes that $1-r>u^{1/2+\varepsilon}$, for some $\varepsilon>0$ to be determined such that the following bound holds,  $\psi_1\leq \varsigma\psi_0 u^{1/2}$ \cite[Eq. (4.3)]{BDZ}. For that, the proof of \cite[Lemma 4.2]{BDZ}  makes a crucial use of the inequality $\beta_1<\beta_{\mathrm{tip}}$. Here, this is precisely valid in the wedge $\mathcal W \subset \mathcal D$ (recall Corollary \ref{wedge}), which thus establishes Lemma \ref{subsupD} in $\mathcal W$. In $\mathcal D\setminus \mathcal W$, we need a different, more general argument.}  

\textcolor{black}{For functions $A$ and $B$ of $r$, we shall use the short-hand notations, $A \lesssim B$ for $A \leq c B$ with $c$ some positive constant, and $A \approx B$ when both $A \lesssim B$ and $A\gtrsim B$ hold. We start with 
$$
\psi_1\approx (1-r)^{-\beta_1}{u}^{\gamma_1}=(1-r)^{-\beta_0}{u}^{\gamma_0}{u}^{\gamma_1-\gamma_0}(1-r)^{\beta_0-\beta_1}\approx \varsigma \psi_0 {u}^{\gamma_1-\gamma_0}(1-r)^{\beta_0-\beta_1}.
$$
Note that the last estimate requires that $g_0(0)\neq 0$, hence the condition $(p,q)\notin \mathcal T_\kappa$ of Lemma \ref{subsupD}. We expect ${u}^{\gamma_1-\gamma_0}(1-r)^{\beta_0-\beta_1}$ to be bounded by some positive power of $u$. 
Recall that in $\mathcal D$, $\beta_0-\beta_1<0$. Hence, for case III where $1-r>u^{1/2+\varepsilon}$,
$$
u^{\gamma_1-\gamma_0}(1-r)^{\beta_0-\beta_1}<u^{\gamma_1-\gamma_0+(1/2+\varepsilon)(\beta_0-\beta_1)}.
$$
Recalling notation \eqref{betaCgamma}, we have 
$$\beta_0-\beta_1=\beta(\gamma_0)-\beta(\gamma_1)=(\gamma_0-\gamma_1)\big(\kappa(\gamma_0+\gamma_1)-\left(2+{\kappa}/{2}\right)\big),$$
which yields 
$
u^{\gamma_1-\gamma_0}(1-r)^{\beta_0-\beta_1}<u^{\alpha},
$ 
with 
\begin{align*}
\alpha&:=(\gamma_1-\gamma_0)\left[1-(1/2+\varepsilon)\big(\kappa(\gamma_0+\gamma_1)-\left(2+{\kappa}/{2}\right)\big)\right]\\
&=(\gamma_1-\gamma_0)\big[1-\kappa\left({1}/{2}+\varepsilon\right)\left(b-{1}/{2}\right)\big].
\end{align*}
Recall that $\gamma_1-\gamma_0>0$ and  $b-1/2<2/\kappa$ by Definition \eqref{def:D} of $\mathcal D$, so that for sufficiently small $\varepsilon$, we have $\alpha >0$. Hence for this $\varepsilon$, we have that
\begin{equation}\label{eq:psi0psi1}
\psi_1\lesssim \varsigma \psi_0 {u}^{\alpha},\,\,\,\alpha >0.
\end{equation}
As in the proof of \cite[Lemma 4.2]{BDZ}, this estimate insures that Lemma \ref{subsupD} holds  in case III.}
\end{proof}
\textcolor{black}{To conclude the proof of Proposition \ref{theo:DD}, Lemmas \ref{lemma4.1} and \ref{subsupD} show that for $(p,q) \notin \mathcal T_\kappa$, $\psi\ell_\delta$ provides both positive sub- and supersolutions depending on the sign of $\delta$, and the generalized spectrum $\beta(p,q)$ can then be read off from the explicit form \eqref{eq:psifin} of $\psi$. The spectrum function $\beta(p,q)$ is extended by  convexity, hence continuity on $\mathcal T_\kappa$, which  by Lemma \ref{pr:tkn} has no interior.}
\end{proof} 

 \textcolor{black}{\subsubsection*{iv) Between the green and red parabolae}
We now concentrate on the infinite domain $\tilde{\mathcal D}$ located between the top branch \eqref{3} of the green parabola $\mathcal G$ stemming from point $P_0$ and the right branch of the red parabola $\mathcal R$ \eqref{seed1} starting from $P_0$ (Fig. \ref{dfk}). Ref. \cite[Theorem 1.4]{DNNZ},  states the validity of the spectrum $\beta_1(p,0)$ on the $q=0$ line,  i.e., for $\sigma=-1$, between the transition point $p^*(\kappa)$ \eqref{pstar}, which is the (left) intersection point of the green parabola $\mathcal G$ with that line,  up to the point $\min \{\hat p(\kappa), p(\kappa)\}$,  where $\hat p(\kappa):=1+\kappa/2$, and where $p(\kappa):=(2+\kappa)(6+\kappa)/8\kappa$ is the rightmost intersection point of the red parabola $\mathcal R$ with the $q=0$ line.}  
\textcolor{black}{\begin{proposition}\label{theo:racket}
Define $\tilde{\mathcal D}$ as the infinite domain located between the left branch \eqref{3} of the green parabola $\mathcal G$ stemming from point $P_0$ and the right branch of the red parabola $\mathcal R$ \eqref{seed1} starting from $P_0$. Define $\hat{\mathcal D}\subset \tilde{\mathcal D}$ as the racket-shaped finite domain located above the straight line $D_3$ of equation $p-q= \hat p(\kappa)=1+\frac{\kappa}{2}$, which intersects  $\mathcal G$ at point $Q_0$ and  $\mathcal R$ at point $P_3:=(1+\frac{2}{\kappa},\frac{4-\kappa^2}{2\kappa})$  (Fig. \ref{dfk}). In $\hat{\mathcal D}$, the average generalized integral means spectrum is $\beta_1(p,q)$.
\end{proposition}}
\begin{proof}
 \textcolor{black}{The proof given for $\sigma=-1$ in \cite[Section 4.2.7]{DNNZ},  was based on the key {\it duality} property  \eqref{duality}  of function  \eqref{betaCgamma}, i.e.,  $\beta(\gamma)=\beta(\gamma')$ for $\gamma'=1/2+2/\kappa-\gamma$. 
  By using the general $\sigma$-formalism set up in Ref. \cite[Sections 4.2.1 -- 4.2.3]{DNNZ}, the same proof applies word for word to the  $\sigma\in \mathbb R$ case. It suffices to replace everywhere the quantities there, $\gamma_{\pm}=\frac{1}{\kappa}(1\pm \sqrt{1+2\kappa p})$, by \eqref{gammasigma}, $\gamma_{\pm}^\sigma=\frac{1}{\kappa}(1\pm \sqrt{1-2\kappa \sigma p}).$ Instead of the test function  \eqref{eq:psifin}, \eqref{eq:g0} as above, use is made of 
\begin{equation}\label{eq:psifinbis}
\psi:=  g(u)(1-z\bar z)^{-\beta(\gamma)},
\end{equation}
where  $g$ is the weighted combination  of two hypergeometric functions,
\begin{equation}\label{eq:g0bis}
g(u)=({u}/4)^{\gamma}\F(a,b,c,{u}/4)-({C_0'}/{C_0})\, ({u}/4)^{\gamma'}\F(a',b',c',{u}/4),
\end{equation} 
where now $\gamma$ is a free parameter, taken such that $\gamma>\gamma_1=\gamma^\sigma_+(p)$ \eqref{gammasigma}, and   
\begin{align}
\label{eq:abcbis}
&a=\gamma -\gamma_1=\gamma -\gamma^\sigma_+, \quad b=\gamma -\gamma_-^\sigma, \quad c=\frac{1}{2}+a+b,\\ \nonumber
&a'=\frac{1}{2}-a, \quad b'=\frac{1}{2}-b, \quad c'=\frac{1}{2}+a'+b'.
\end{align}
At the end of the argument, one lets $a\to 0^+$, $\gamma \to \gamma_1=\gamma^\sigma_+$,  so that $\beta(\gamma)\to \beta(\gamma_1)=\beta_1(p,q)$ \eqref{betapgammasigma}. 
}
 
  \textcolor{black}{The proof given in Section 4.2.7 of Ref. \cite{DNNZ} holds provided that two simultaneous conditions are fulfilled, $C(p,\gamma_+)> 0$ and $C(p,\gamma'_+)< 0$, where $\gamma_+':=\frac{1}{2}+\frac{2}{\kappa}-\gamma_+$ is the dual value of $\gamma_+$. These two conditions precisely gave $p^*(\kappa)<p<p(\kappa)$. By definition, on the branch \eqref{2}, \eqref{2bis} of the red parabola $\mathcal R$ \eqref{seed1}, hence on its right branch starting from $P_0$, we have $C(p,\gamma^\sigma_+)=0$. On the branch \eqref{3} of the green parabola $\mathcal G$ \eqref{seed2},  hence on its left branch  starting from $P_0$, we have $C(p,{\gamma^\sigma_+}')=0$, where ${\gamma^\sigma_+}'$ is the dual value of ${\gamma^\sigma_+}$. By continuity, since $C(p,\gamma^\sigma_+)>0$ to the right of  $\mathcal G$ on the $q=0$, $\sigma=-1$ line, and $C(p,{\gamma^\sigma_+}')<0$ to the left of  $\mathcal R$ on the same  line, the generalized proof conditions $C(p,\gamma^\sigma_+) > 0$ and $C(p,{\gamma^\sigma_+}') < 0$ simultaneously hold  in the domain $\tilde{\mathcal D}$ of the $(p,q)$-plane, as defined in Proposition \ref{theo:racket}.}
  
    \textcolor{black}{Lastly, the rationale for the occurrence of $\min \{\hat p(\kappa), p(\kappa)\}$ in  \cite[Theorem 1.4]{DNNZ} was the technical requirement that the dual value $\gamma_+'$ of $\gamma_+$  
  is such that there is no dual tip contribution, i.e., that it satisfies $2\gamma'_++1\geq 0$, yielding the auxiliary condition $p\leq \hat p(\kappa)$. 
 Since $\sigma=q/p-1$, we simply have $\gamma_{\pm}^\sigma(p)=\gamma_{\pm}(p-q)$, and the condition in the $(p,q)$-plane is now for a general value of $\sigma$,             
 \begin{equation}\label{D_3}
p-q\leq \hat p(\kappa)=1+\frac{\kappa}{2}.
\end{equation}
This restricts the domain of validity of the argument to the part $\hat {\mathcal D}$ of $\tilde{\mathcal D}$ located above the straight line $D_3$ of equation $p-q= \hat p(\kappa)$ (Fig. \ref{dfk}), and concludes the proof of the validity of the spectrum $\beta_1(p,q)$ \eqref{betapq} in $\hat{\mathcal D}$.}
\end{proof}
\noindent{\bf Proof of Theorem \ref{propseparbis}.} \textcolor{black}{\begin{proof} Propositions \ref{theo:lin} in item {\it i)}, \ref{theo:DHBS} in {\it ii)}, \ref{theo:DD} and its corollary \ref{wedge} in {\it iii)}, and \ref{theo:racket} in {\it iv)},  altogether establish the validity of the various spectra given in Theorem \ref{propseparbis} and Fig. \ref{separ}, in the following maximal domain of the $(p,q)$-plane. It is in Figure \ref{dfk} the infinite domain above the frontier line made by the union  of the lower branch of the green parabola up to its intersection point with $D_3$, obtained for $\gamma'=-1$ in Eq. \eqref{C+dual} as $\big((4-\kappa)(4+3\kappa)/8\kappa,(16-7\kappa^2)/8\kappa\big)$, of the segment of $D_3$ from that point up to $P_3$, of the branch of the red parabola $\mathcal R$ between $P_3$ and $P_0$, and finally of $D_1$ from $P_0$ up to infinity. This concludes the proof of Theorem \ref{propseparbis}.\end{proof}}

\subsection{$m$-fold spectrum}\label{mfoldspectrum}
For $m\geq 1$, the generalized integral means spectrum $\beta^{[m]}(p,q;\kappa)$, associated with the $m$-fold transform $f^{[m]}$ of the SLE whole-plane map, can be directly derived from the analysis given in Section \ref{mfold}. Definition \ref{def:gims} and identities  \eqref{midentity} and \eqref{qm} immediately imply that
\begin{equation}\label{bmb1} \begin{split}
&\beta^{[m]}(p,q;\kappa)=\beta^{[1]}(p,q_m;\kappa),\\
&q_m=q_m(p,q)=\left(1-{1}/{m}\right)p+{q}/{m},
\end{split}
\end{equation}
 where $\beta^{[1]}(p,q;\kappa):=\beta(p,q;\kappa)$ is the $m=1$ average generalized integral means spectrum of whole-plane $\SLE_\kappa$ studied above. As we shall see in Section \ref{mneg}, Definition \ref{def:gims} can be formally extended to $m\leq -1$, and the resulting $\beta^{[m]}$  
 still obeys \eqref{bmb1}. 
\textcolor{black}{\subsubsection{Standard $m$-fold integral means spectrum} Let us first focus on the $m$-fold \emph{standard}  spectrum, i.e., on the $q=0$ case, for which 
\begin{equation} \label{q0qm}
q_m(p,0)=\left(1-1/m\right)p, \,\,\, m\in \mathbb Z\setminus\{0\}, p\in \mathbb R.
\end{equation}
This defines a line $q=q_m(p,0)$ in the original $(p,q)$ plane. 
 For the $\beta_1$ part of the spectrum, this yields
 \begin{align} \label{betam0} 
 \beta_1\big(p,(1-1/m)p;\kappa\big)
 =\left(1+\frac{2}{m}\right)p-\frac{1}{2}-\frac{1}{2}\sqrt{1+\frac{2\kappa p}{m}},
\end{align}
in agreement with the result obtained in \cite[Eq. (22)]{DNNZ}. Note, however, that for $m=-2$ this function is negative and cannot appear in the spectrum.}
\begin{figure}[htbp]
\begin{center}
\includegraphics[width=.45\linewidth]{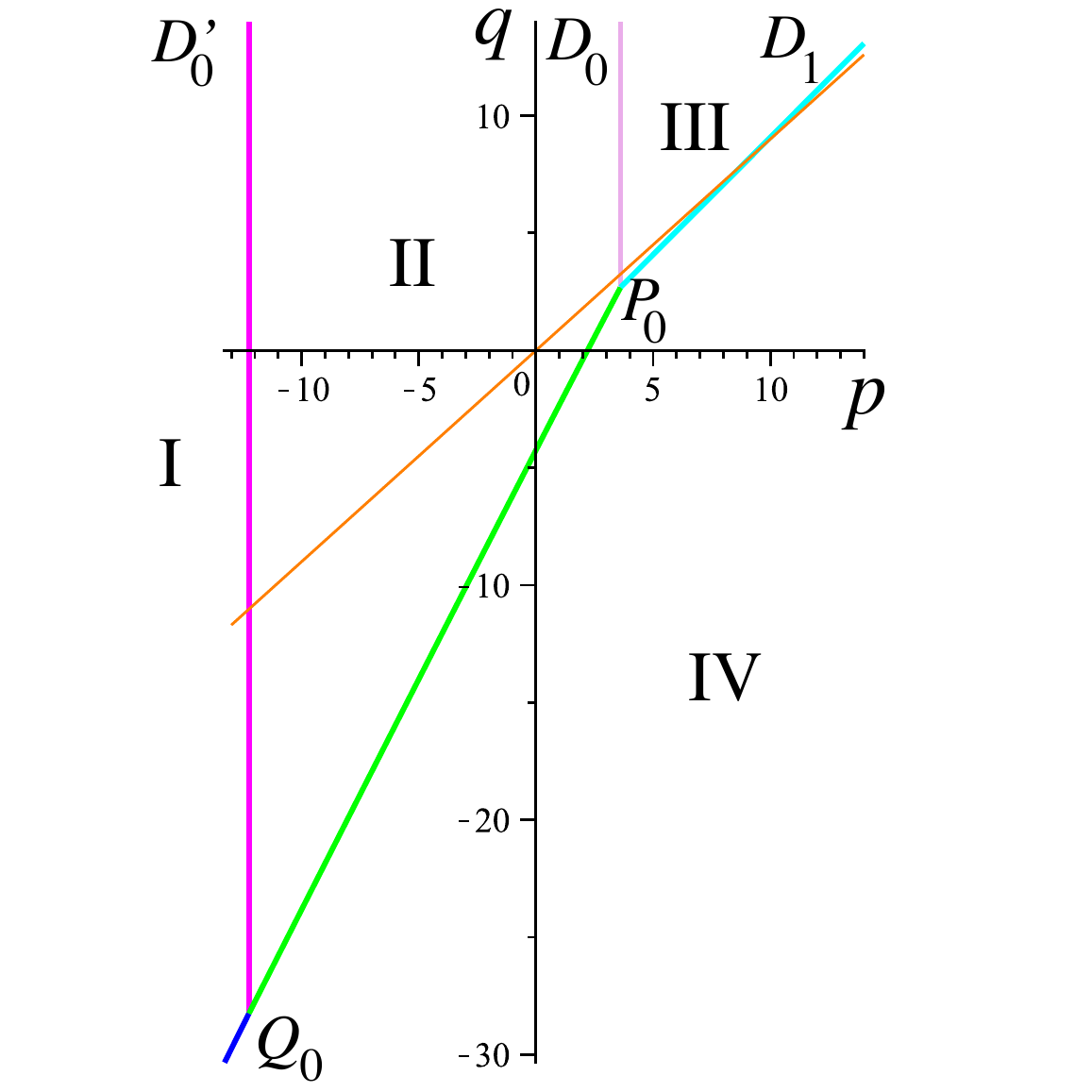}
\hskip1cm
\includegraphics[width=.45\linewidth]{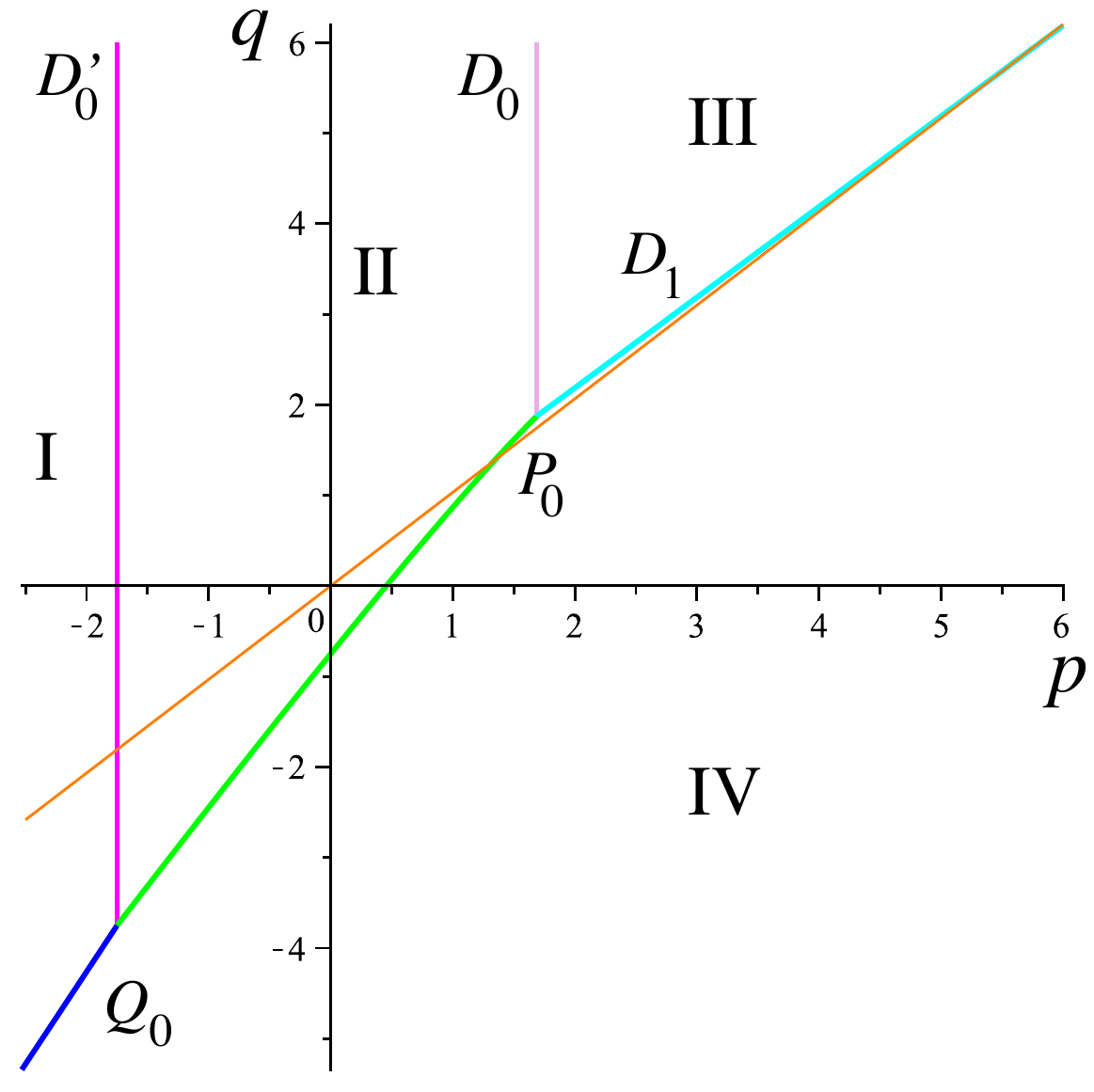}
\caption{{\it Phase diagram for the $m$-fold whole-plane $\SLE_{\kappa}$ in the $(p,q_m)$-plane, along trajectories \eqref{q0qm} (coral color). 
 Left:  For $m=+10$ and $\kappa=30$, the line successively crosses  domains I, II, III and IV. Right:   For $m=-30$ and $\kappa=2$, it successively crosses domains I, II, IV and III.}}
\label{figpqmline}
\end{center}
\end{figure}

Let us consider the point $P_0$ whose coordinates $(p_0,q_0)$ are given by \eqref{P0}. The line $OP_0$ has slope
$\frac{q_0}{p_0}=
1-\frac{1}{3}\frac{\kappa-4}{\kappa+4},$
a quantity that decreases from $4/3$ to $2/3$ as $\kappa$ runs from $0$ to $\infty$ and takes the value $1$ for $\kappa=4$. 

In the case $\kappa>4$, it is thus possible to find $m>0$ such that $\frac{q_0}{p_0}\leq 1-\frac{1}{m}<1.$  \textcolor{black}{This holds for $m\geq m(\kappa)$, with the definition, 
\begin{equation}\label{m+}\begin{split}
&m(\kappa):={3}\frac{\kappa+4}{\kappa-4};\,\,\, m({16}/{\kappa})=-m(\kappa),\\ &\kappa>4, m(\kappa) >3;\,\,\,0<\kappa <4, m(\kappa)<-3.
\end{split}\end{equation}} The line \eqref{q0qm} then first intersects the vertical line $p=p_0$ {\it above} point $P_0$, and  the line $D_1$ of unit slope afterwards (Fig. \ref{figpqmline}, left). So in this case the standard $m$-fold integral means spectrum $\beta^{[m]}(p):=\beta^{[m]}(p,0)$, for $p \in \mathbb R$, has four phases, $\beta_{\mathrm{tip}}$, $\beta_{0}$, $\beta_{\mathrm{lin}}$ and $\beta_1$ \eqref{betam0}. 
In the second case, $\kappa <4$, it is possible to find $m<0$ such that 
$1<1-\frac{1}{m}\leq\frac{q_0}{p_0},$ which is equivalent to $m\leq m(\kappa)<-3$. 
The line \eqref{q0qm} then crosses the line $p=p_0$ at a point {\it below} point $P_0$. It follows that it first crosses the green parabola and  afterwards crosses the line $D_1$, after which the spectrum becomes linear (Fig. \ref{figpqmline}, right). So the standard $m$-fold spectrum again has four phases, but now in the order $\beta_{\mathrm{tip}}$, $\beta_{0}$, $\beta_1$ \eqref{betam0} and $\beta_{\mathrm{lin}}$ for $p\in\mathbb R$. 
 \textcolor{black}{For the critical case $\kappa=4$, $\frac{q_0}{p_0}=1$, and for $m \geq 1$, we observe the ordered set $\beta_{\mathrm{tip}}$, $\beta_{0}$, $\beta_1$, or for $m \leq -1$, the set $\beta_{\mathrm{tip}}$, $\beta_{0}$ and $\beta_{\mathrm{lin}}$.}
 
\textcolor{black}{We finally obtain the following table for the successive expressions taken by $\beta^{[m]}(p)$ for $p\in \mathbb R$.
\bigskip
\begin{center}
\begin{tabular}{|c|c|c|c|c|c|c|}
\hline
\multirow{2}{*}{$\kappa>4$} & \multicolumn{2}{c|}{$m\leq -1$} &
\multicolumn{2}{c|}{$1\leq m\leq m({\kappa})$} & \multicolumn{2}{c|}{$3<m({\kappa})\leq m$}\\
\cline{2-7}
& \multicolumn{2}{c|}{$\beta_{\mathrm{tip}}, \beta_0, \beta_{\mathrm{lin}}$} & \multicolumn{2}{c|}{$\beta_{\mathrm{tip}}, \beta_0, \beta_1$} & \multicolumn{2}{c|}{$\beta_{\mathrm{tip}}, \beta_0, \beta_{\mathrm{lin}}, \beta_1$}\\
\hline
\multirow{2}{*}{$\kappa=4$} & \multicolumn{2}{r}{$m\leq -1$} &  \multicolumn{2}{c}{$\vert$} & \multicolumn{2}{l|}{$1\leq m$}\\
\cline{2-7}
& \multicolumn{2}{r}{$\beta_{\mathrm{tip}}, \beta_0, \beta_{\mathrm{lin}}$} &  \multicolumn{2}{c}{$\vert$} & \multicolumn{2}{l|}{$\beta_{\mathrm{tip}}, \beta_0, \beta_1$}\\
\hline
\multirow{2}{*}{$\kappa<4$} & \multicolumn{2}{c|}{$m\leq m({\kappa})<-3$}  & \multicolumn{2}{c|}{$m({\kappa})\leq m\leq -1$} & \multicolumn{2}{c|}{$1\leq m$}\\
\cline{2-7} 
& \multicolumn{2}{c|}{$\beta_{\mathrm{tip}}, \beta_0, \beta_1, \beta_{\mathrm{lin}}$} & \multicolumn{2}{c|}{$\beta_{\mathrm{tip}}, \beta_0, \beta_{\mathrm{lin}}$} & \multicolumn{2}{c|}{$\beta_{\mathrm{tip}}, \beta_0, \beta_1$}\\
\hline
\end{tabular}
\end{center}}
\textcolor{black}{\begin{remark}
In the three cases $m\in\{-1, -2,-3\}$, and for any $\kappa>0$, only the usual spectra $\beta_{\mathrm{tip}}$, $\beta_{0}$, and $\beta_{\mathrm{lin}}$ appear in the $m$-fold standard integral means spectrum $\beta^{[m]}(p)$.
\end{remark}}  

 
The discussion shows that the spectrum $\beta_1$ (at $(p, q_m)$) \eqref{betam0} may appear in $\beta^{[m]}(p)$ even if the boundary of the SLE  image domain is {\it bounded}. This indeed happens \textcolor{black}{for $m \leq -4$ and $\kappa$ small enough:} for negative $m$, the $m$-fold transform of the outer whole-plane SLE  is then conjugate by $z\mapsto 1/z$ of the $(-m)$-fold transform of the inner whole-plane $\SLE$, which gives rise to an univalent function map onto a domain with bounded boundary. In this case, the appearance of the $\beta_1$ spectrum is due to a higher $|m|$-fold branching at the origin for $\kappa<4$ and $|m|\geq |m(\kappa)|$.

\subsubsection{Generalized $m$-fold integral means spectrum ($m\geq 1$)} 
Let  now $T_m=\bigl(\begin{smallmatrix}
1&0\\ 1-1/m&1/m
\end{smallmatrix} \bigr)$ be the endomorphism of $\mathbb R^2$ defined by $T_m(p,q)=(p,q_m)$, with inverse $T_m^{-1}=\bigl(\begin{smallmatrix}
1&0\\ 1-m&m
\end{smallmatrix} \bigr)$. 
The separatrix lines for the $m$-fold case are the images by $T_m^{-1}$ of those for $m=1$.  
\begin{figure}[htbp]
\begin{center}
\includegraphics[width=.55\linewidth]{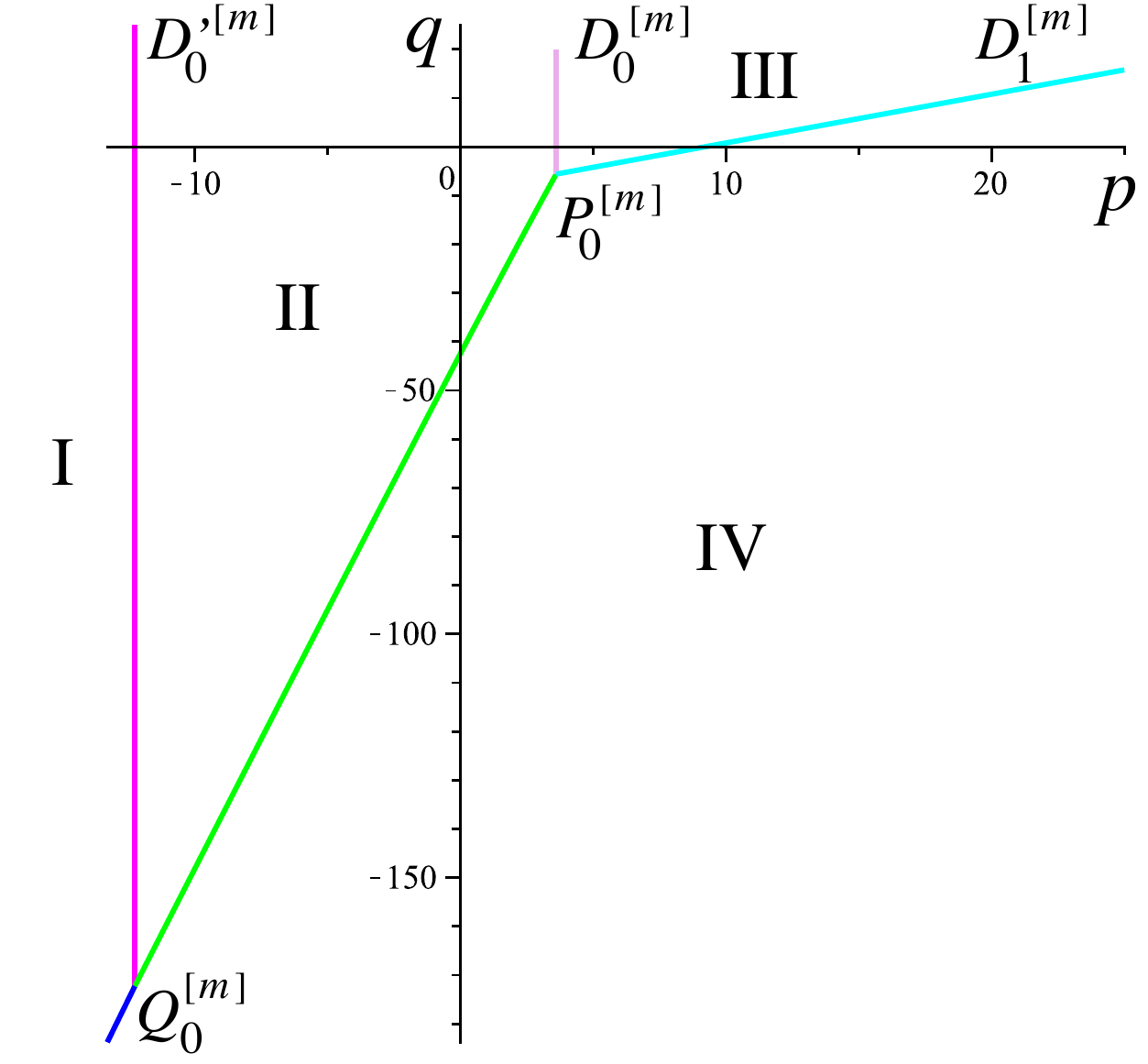}
\caption{{\it Phase diagram for the $m$-fold whole-plane $\SLE_{\kappa}$ and domains of validity of spectra $\beta_{\mathrm{tip}}$ (I), $\beta_0$ (II), $\beta_{\mathrm {lin}}$ (III), and $\beta_m$ (IV). 
  For $m=+10$ and $\kappa=30$, the $q=0$ line successively crosses  domains I, II, III and IV.}}
\label{figpqline}
\end{center}
\end{figure}
Theorem \ref{propseparbis} then yields: 
\begin{theorem}\label{propseparmlightbis} (Figure \ref{figpqline}) 
  Separatrix curves for the generalized integral means spectrum $\beta^{[m]}(p,q;\kappa)$ of the $m$-fold  whole-plane $\SLE_\kappa$ are given, for $m\geq 1$, by the same as in Theorem \ref{propseparbis} for $m=1$, provided that one replaces there, 
 \begin{itemize}
  \item $D_0$ by $D_0^{[m]}$, $P_0$ by $P^{[m]}_0=(p_0,q_0^{[m]})$,  
$q_0$ by $q_0^{[m]}:=p_0+m(16-\kappa^2)/32\kappa$;
  \item $D_1$ by $D_1^{[m]}$ of equation $q-p=m(16-\kappa^2)/32\kappa$;
   \item $q_{\mathcal G}(\gamma)$ by 
   $p_{\mathcal G}(\gamma)+m\left(\gamma-\kappa\gamma^2/2\right)$, $Q_0$ by $Q_0^{[m]}=(p'_0,q'^{[m]}_0)$,  
     with\\ $q'^{[m]}_0:=p'_0-m\left(1+\kappa/2\right)$;    
    \item 
 $D_0'$ by the vertical half-line $D'^{[m]}_0$ above $Q_0^{[m]}$;
      \item  
      $q_{\mathcal Q}(\gamma)$ by $p_{\mathcal Q}(\gamma)+m\left(\gamma-\kappa\gamma^2/2\right)$. 
 \end{itemize}
\end{theorem}
The figures of the $m$-separatrices are easily deduced from Theorem \ref{propseparmlightbis}; in particular, the transformed quartic $\mathcal Q^{[m]}$ is asymptotic for $p\to -\infty$ to a straightline,
\begin{equation}\label{slope}
q=(m+1)p-m(2+\kappa)/8,
\end{equation} 
whose direction is also that of the axis of the transformed green parabola $\mathcal G^{[m]}$. In region IV, the $m$-fold integral means spectrum is thus given by 
\begin{align}\label{betam}
\beta^{[m]}(p,q;\kappa)=\beta_1(p,q_m;\kappa)=\left(1+\frac{2}{m}\right)p-\frac{2}{m}q-\frac{1}{2}-\frac{1}{2}\sqrt{1+\frac{2\kappa}{m}(p-q)}.
\end{align} 
\subsubsection{The $m\leq -1$ case}\label{mneg}
\textcolor{black}{In this range, $f^{[m]}$ is an exterior map from $\mathbb D_-$, whose continuation to $\partial \mathbb D$ vanish on $-m$ points, and the image domain is the plane slit by $-m$ curves in a star configuration at the origin. When defined via a formal extension of Definition \ref{def:gims}, this generalized means spectrum $\beta^{[m]}(p,q;\kappa)$ now describes for $q > 0$ the singular behavior of the map near $(f^{[m]})^{-1}(0)$. It still obeys identity \eqref{bmb1},} and the same conclusions as in Theorem \ref{propseparmlightbis} hold formally, except that, because $\det T_m^{-1}=m <0$, the vertical positions of the respective domains of validity of the spectra are all {\it reversed}, the vertical separatrix lines $D_0^{[m]}$ and $D_0'^{[m]}$ being now half-lines going downwards from  $P_0^{[m]}$ and $Q_0^{[m]}$ to $-\infty$, and the domain IV lying  {\it above} the half-line $D_1^{[m]}$, the transformed green parabola and the transformed quartic. The concavity of the separatrix curves is correspondingly inverted (Fig. \ref{figm=-1}).

\begin{figure}[htbp]
\begin{center}
\includegraphics[width=.43290\linewidth]{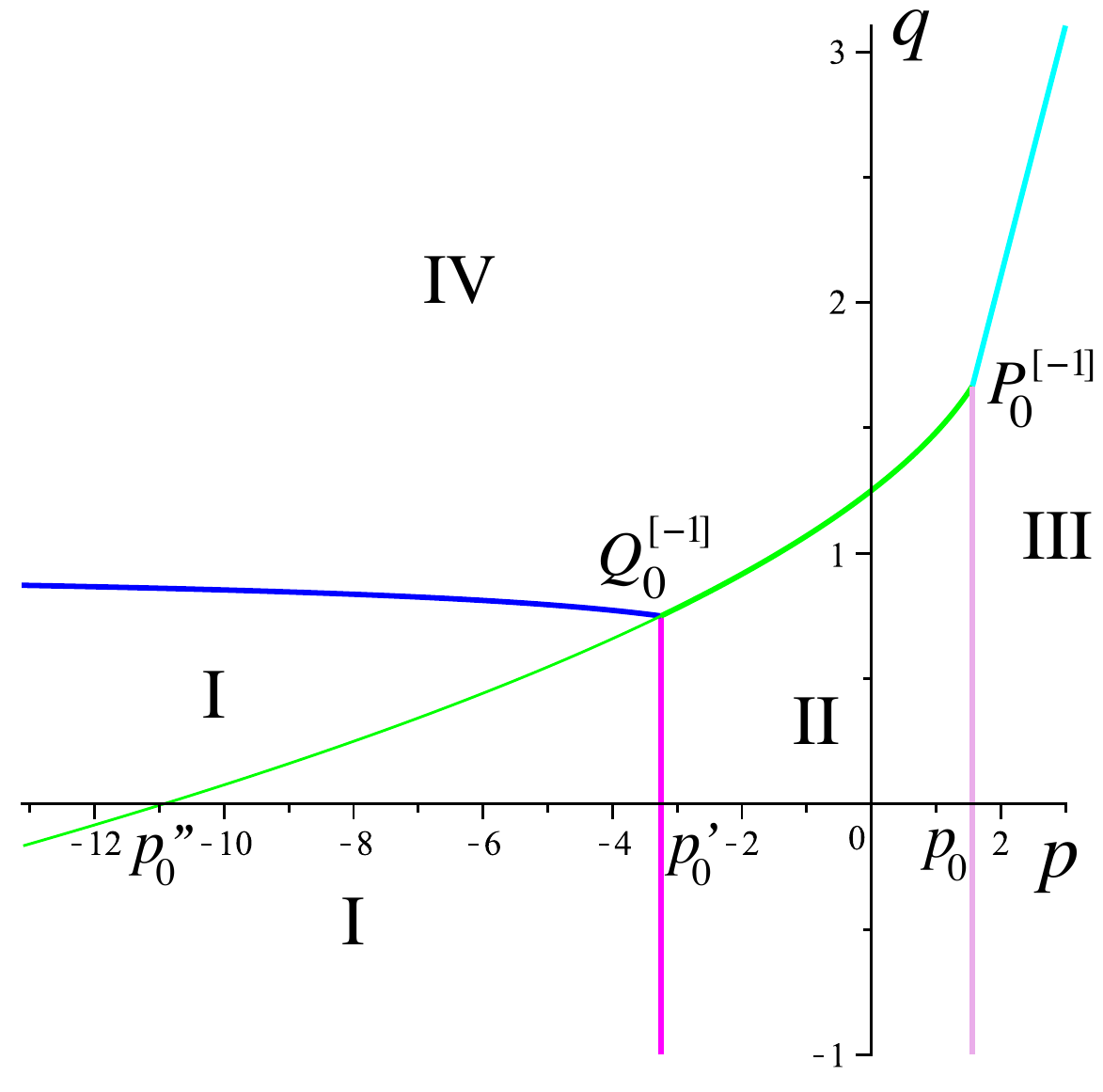}
\includegraphics[width=.45\linewidth]{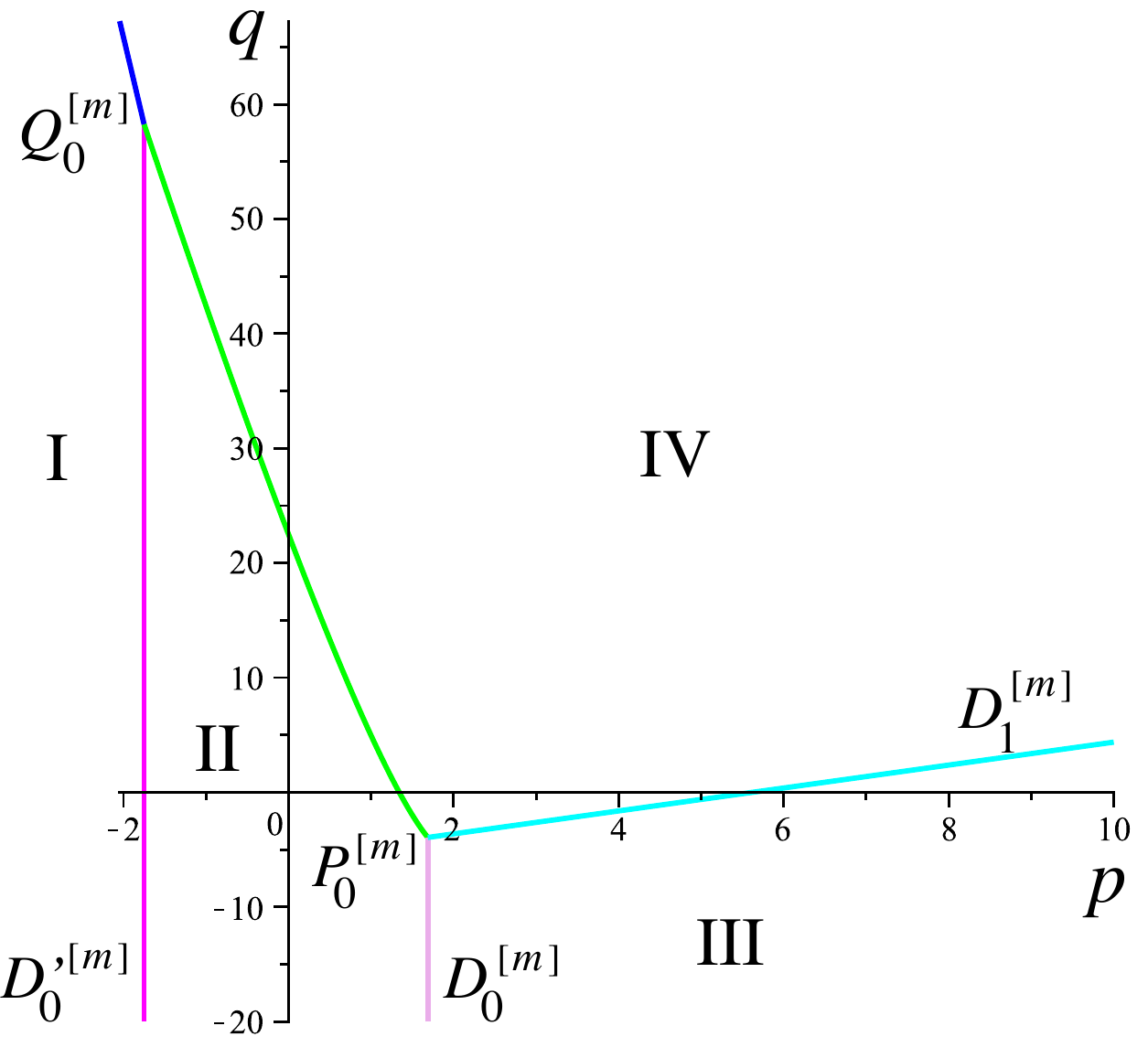}
\caption{{\it Left:  Phase diagram for $m=-1$ 
and $\kappa=6$.  Right: Phase diagram for $m=-30$ and $\kappa=2$; the $q=0$ line successively crosses domains I, II, IV and III.}}
\label{figm=-1}
\end{center}
\end{figure}

Note the slope inversion in \eqref{slope} when going from $m\geq 1$ to $m\leq -1$.  
The $m=-1$ case, i.e., the B--S exterior version of whole-plane SLE, is thus peculiar: the parabola's axis and the quartic's linear asymptote are both horizontal, and $\mathcal G^{[-1]}$ intersects the $p$-axis at $p_0''(\kappa) \leq p_0'(\kappa)$, in agreement with Section \ref{BSline}. (See Fig. \ref{figm=-1}, left.) \textcolor{black}{As hinted above, the $m=-2$ case is special, and one can check that the green parabola $\mathcal G^{[-2]}$ stays strictly above the $p$-axis for any $\kappa>0$. Hence, this axis never enters region IV, rendering the (degenerate) negative spectrum \eqref{betam0} irrelevant.}

In Theorem \ref{propseparmlightbis}, the coordinate $q_0^{[m]}$ of point $P_0^{[m]}$ can become {\it negative}. This corresponds to the fact that the $q=0$ axis intersects all four different regions in the $m$-fold spectrum (Figs. \ref{figpqline} and \ref{figm=-1}). \textcolor{black}{This happens for $\frac{m}{3}\frac{\kappa-4}{\kappa+4} \geq 1$, 
which gives either $m\geq m(\kappa)>3$ for $\kappa > 4$ or $m\leq m(\kappa)<-3$ for $\kappa<4$, in agreement with the discussion of Fig. \ref{figpqmline} above.} 
\section{Universal spectrum}\label{univ}
It is worthwhile to compare the above results to universal ones. We aim at generalizing the universal spectrum for integral means of derivatives of univalent functions to the case of mixed integrals of the type 
$\int_0^{2\pi}\frac{\vert f'(re^{i\theta})\vert^p}{\vert f(re^{i\theta})\vert^q}d\theta.$
More precisely, for $f$ injective and holomorphic in the unit disk, we define $\beta_f(p,q)$ as being the smallest number such that
$$\int_0^{2\pi}\frac{\vert f'(re^{i\theta})\vert^p}{\vert f(re^{i\theta})\vert^q}d\theta\leq O(1-r)^{-\beta_f(p,q)-\varepsilon}, \forall \varepsilon >0,r \to 1.$$
   The {\it universal spectrum} $B(p,q)$ is then defined as the supremum of $\beta_f(p,q)$ over all holomorphic and injective $f$'s on the disk.
   
It should be first noticed that if one restricts oneself to bounded univalent functions, there will be no change with respect to the usual integral means spectrum, i.e., the denominator in the integrand (and thus $q$) plays no role in this case. In the general case, we will mimic the Feng-McGregor approach \cite{FengMcg}. 
\begin{theorem} \label{Universalspectrum} Let  $f$ be holomorphic and injective in the unit disk. For $p\in \mathbb R^+,q \in \mathbb R$  such that $q<\min\{2,\frac{5}{4} p-\frac{1}{2}\}$,  there exists a constant $C>0$ such that
\begin{equation} \int_0^{2\pi} \frac{\vert f'(re^{i\theta})\vert^p}{\vert f(r e^{i\theta})\vert^q}d\theta\leq \frac{C}{(1-r)^{3p-2q-1}} \label{FMcggen}.\end{equation}
\end{theorem}
The universal spectrum  is therefore finite and such that $B(p,q)\leq 3p-2q-1$, at least in the domain $\mathcal D_0:=\{0\leq p, q<\min\{2,\frac{5}{4} p-\frac{1}{2}\}\}$ of Theorem \ref{Universalspectrum}. In that domain, the Koebe function, $\mathcal K(z)=z (1+z)^{-2}$,  saturates the bound, therefore $B(p,q)=3p-2q-1 >0$ for $(p,q)\in \mathcal D_0$.

In order to make the proof lighter, we will neither write the variables in the functions involved, which are of the form $re^{i\theta}$ with $r$ fixed, nor the angular integration interval, which is meant to be $[0,2\pi]$. 
\begin{proof}
Let $a,b$ be two reals, to be fixed later, such that $a-b=1$. Let us first consider the case $p<2$, for which H\"older's inequality gives
\begin{equation}\label{holder}\int\frac{\vert f'\vert^p}{\vert f\vert^q}=\int\frac{\vert f'\vert^p}{\vert f\vert^{aq}}{\vert f\vert^{bq}}
\leq \left(\int\frac{\vert f'\vert^2}{\vert f\vert^{2aq/p}}\right)^{p/2}\left(\int\vert f\vert^{2bq/(2-p)}\right)^{(2-p)/2}.
\end{equation}
In order to estimate the first integral on the right-hand side, we invoke Hardy's inequality \cite{Pommerenke}: For any $p'>0$, there exists a constant $C'>0$ such that for any function $f$ which is  holomorphic and injective in the unit disk,
$$ \int\vert f'\vert^2\vert f\vert^{p'-2}\leq \frac{C'}{(1-r)^{2p'+1}}.$$
For the rightmost integral in \eqref{holder}, we use the Prawitz inequality \cite{Pommerenke}: For any $p''>1/2$, there exists a constant $C''>0$ such that for any function $f$  holomorphic and injective in the unit disk,
$$ \int\vert f\vert^{p''}\leq \frac{C''}{(1-r)^{2p''-1}}.$$
We then take $p':=2-\frac{2aq}{p}, p'':=\frac{2bq}{2-p}$, and assume that $p'>0$ and $p''>1/2$; we may then use the two  inequalities above and get from \eqref{holder}
\begin{equation} \label{universal}\int\frac{\vert f'\vert^p}{\vert f\vert^q}\leq\frac{C}{(1-r)^{3p-2q-1}},
\end{equation}
for some $C>0$ and any $f$ as above. For this, we need to find $a,b\in\mathbb R$ such that 
$a-b=1,\;p'>0,\;p''>1/2.$  
The first inequality is equivalent to $p>aq$, and the second one gives  $aq>q+\frac{2-p}{4}$.
We thus find that the universal bound \eqref{universal} holds for 
$q+\frac{1}{2}<\frac{5}{4} p.$ 
Recall then the original condition of validity, $p<2$, which implies that $q<2$. 
The theorem being already proved for $p<2$,  we may now assume that $p\geq 2$. Let then $p'$ be such that
$\frac 45 q+\frac 25<p'<2\leq p.$ 
 We now invoke Koebe distortion theorem:
$$ \forall z\in \mathbb D,\;\vert f'(z)\vert \leq 2\frac{\vert f'(0)\vert}{(1-\vert z\vert)^3},$$
from which follows, by writing $\vert f'\vert^p=\vert f'\vert^{p'}\vert f'\vert^{p-p'}$ and by using \eqref{FMcggen} for the couple $(p',q)$, that for some $C>0$,
\begin{equation} \label{universalprime}\int\frac{\vert f'\vert^p}{\vert f\vert^q}\leq\frac{C}{(1-r)^{3p'-2q-1+3(p-p')}}=\frac{C}{(1-r)^{3p-2q-1}}.
\end{equation}
\end{proof}
Guided by the results obtained above for the generalized integral means spectrum of  whole-plane $\SLE$, we will now state a conjecture concerning the universal generalized spectrum. As we shall see, its structure turns out to be very similar, each of the $\SLE$ four spectra having its own analogue in the universal case.

Let us first recall that the universal spectrum for {\bf bounded} holomorphic and injective functions, $B(p)$, is known to be  equal to $p-1$ for $p\geq 2$, and equal to $-p-1$ below a certain threshold $p^{\dagger}\leq -2$. For $p\in [p^{\dagger}, 2]$, it is equal to a unknown function, $B_0(p)$. Two famous conjectures are that by Brennan, stating that $B_0(-2)=1$ and implying that $p^{\dagger}=-2$, and the broader conjecture by Kraetzer stating that $B_0(p)=p^2/4$ (see Ref. \cite{Pommerenke} and references therein). 

For {\bf unbounded} functions, a classical result by Makarov \cite{Makanaliz} states that the universal spectrum is simply given by 
\begin{equation}\label{makarov}
\max\{B(p),3p-1\},
\end{equation}
the second term corresponding to the extremal case of the Koebe function.

Now, in the case of generalized spectra, the universal analogue of the SLE generalized spectrum $\beta_1(p,q;\kappa)$ is naturally the spectrum that we have just obtained in Theorem \ref{Universalspectrum}, and that corresponds to the Koebe limit of Section \ref{koebe},
\begin{equation}\label{B1}
B_1(p,q):=3p-2q-1.
\end{equation}
The analogue of the SLE bulk spectrum, $\beta_0(p)$, is then naturally given by the function $B_0(p)$ of the bounded universal spectrum above, while the two remaining SLE spectrum functions, $\beta_{\mathrm{tip}}(p)$ and $\beta_{\mathrm{lin}}(p)$, have respectively for universal analogues, 
$B_{\mathrm{tip}}(p):=-p-1$ for $p\leq p^{\dagger}$, and $B_{\mathrm{lin}}(p):=p-1$ for $p\geq 2$.

We then proceed as for $\SLE$, looking for the sets of points in the $(p,q)$ plane such that $B_1(p,q)=B_{\mathrm {tip}}(p)$, $B_1(p,q)=B_0(p)$, $B_1(p,q)=B_{\mathrm{lin}}(p)$. They turn out to be, in the same order,
\begin{itemize}
\item the line $q=2p$ for $p\leq p^\dagger$,
\item the curve $2q=3p-1-B_0(p)$ for $p\in [p^\dagger,2]$,
\item the line $p=q$ for $p\geq 2$.
\end{itemize}
Note that if Brennan's conjecture holds, $p^\dagger=-2$, and it is equivalent to the fact that the separatrix curve, $2q=3p-1-B_0(p)$, the vertical line, $p=-2$, and the separatrix, $q=2p$, all meet at point $(-2,-4)$. If Kratzer's conjecture also holds, the first curve becomes the segment of parabola $2q=3p-1-p^2/4$, with $p\in [-2,2]$.

In complete analogy with the $\SLE$ case (see Fig. \ref{separ}), we thus obtain a prediction for the universal spectrum  $B(p,q)$, with a partition of the plane into four zones corresponding to the four spectra introduced above, as illustrated in  Fig. \ref{univsepar}.
\begin{figure}[!h]
\begin{center}
\includegraphics[width=.473290\linewidth]{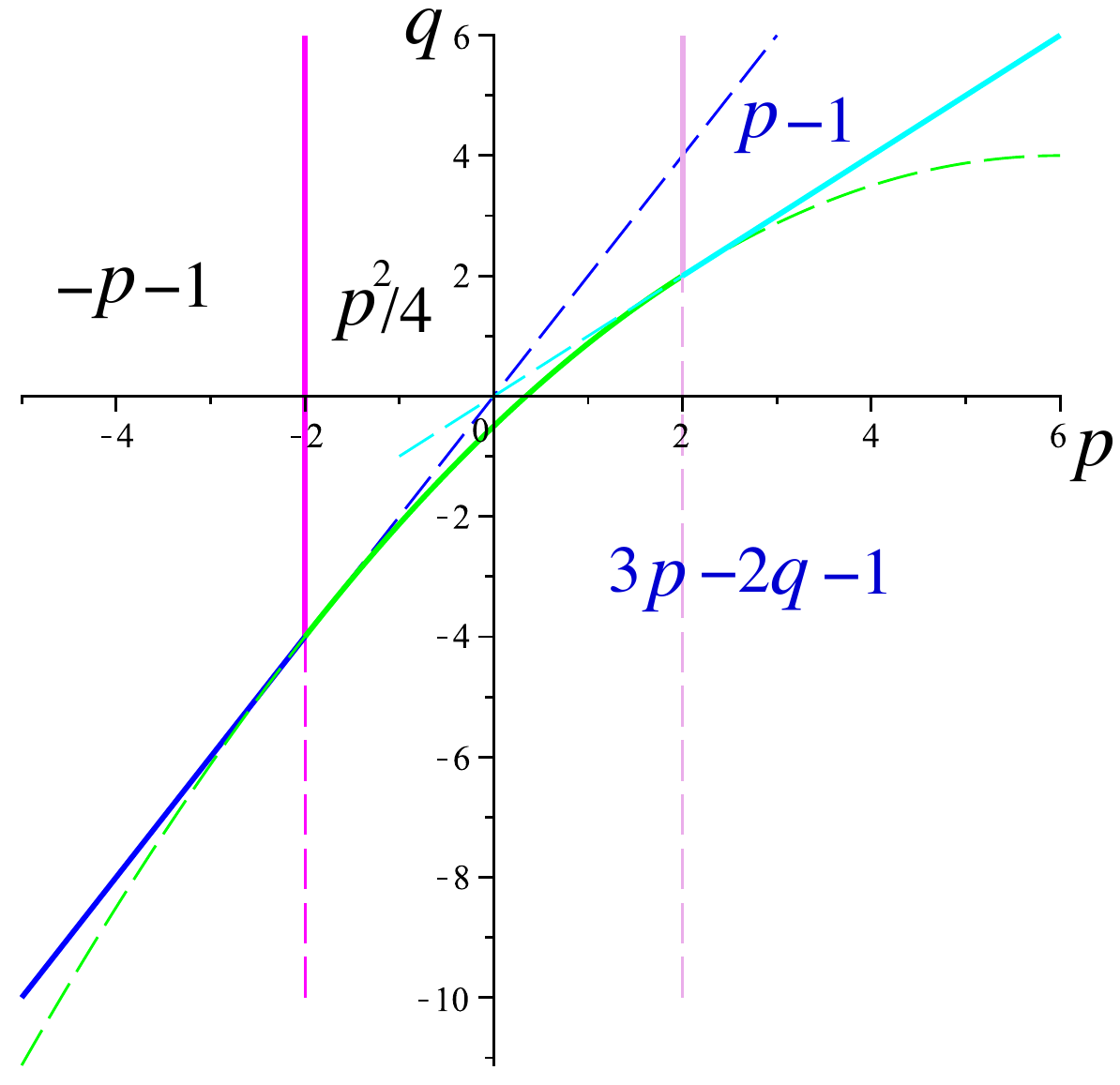}
\caption{The four functions giving the universal generalized spectrum (assuming here the validity of Kraetzer's conjecture).}
\label{univsepar}
\end{center}
\end{figure}

Observe that the figure contains, for $q=0$, the universal spectrum for all univalent functions, as well as, along the line $q=2p$, the spectrum of bounded univalent ones. For $p\leq p^\dagger$ (possibly  $(-2)$), the latter line also appears as  a {\it separatrix}  of the (conjectured) universal spectrum. A small departure from it triggers a phase transition in the spectrum, which is thus {\it unstable} along the bounded functions line.

As work done with Kari Astala shows \cite{ADZ}, it is actually possible to extend Makarov's approach \cite{Makanaliz} to the universal generalized spectrum $B(p,q)$, and to  generalize result \eqref{makarov} into
\textcolor{black}{\begin{theorem}
The universal generalized spectrum is given by
$$B(p,q)=\max\{B(p),3p-2q-1\},$$ 
where $B(p)$ is the universal spectrum for bounded univalent functions.
\end{theorem}}
This confirms the conclusions drawn above for the universal  generalized spectrum, the unknown remaining the position of $p^\dagger$ and the form of $B_0(p)$ in the standard universal spectrum.
\bibliographystyle{plain}
\bibliography{biblio}
\end{document}
 \textcolor{red}{\begin{proposition}\label{propseparmneg} \textcolor{red}{{\bf To be deleted since we refer to Figure 7}} (Figure \ref{figpqline})  
  The separatrix lines for the generalized multifractal spectrum, $\beta^{[m]}(p,q;\kappa)$, of the $m$-fold transform of the whole-plane $\SLE_\kappa$ are given, for $m<0$, in a anti-clockwise direction in the $(p,q)$ plane, respectively by the following equations:
 \begin{itemize}
 \item (i) $p=p_0=p_0(\kappa)=3(4+\kappa)^2/32\kappa,q\leq q_0^{[m]}:=p_0+m(16-\kappa^2)/32\kappa$, for the vertical  half-line separatrix $D_0^{[m]}$ below point $P^{[m]}_0:=(p_0,q^{[m]}_0)$;
  \item (ii) $q-p=m(16-\kappa^2)/32\kappa$ with $p\geq p_0(\kappa)$, for the unit slope half-line separatrix $D_1^{[m]}$ originating at $P_0^{[m]}$;
   \item (iii) the combined parametric equations \eqref{C+dual},
    \begin{align*}
   &p=p_{\mathcal G}(\gamma),\,\,\, \gamma\in[1/4+1/\kappa,1+2/\kappa],\\ 
   &q
   = p_{\mathcal G}(\gamma)+m\left( q_{\mathcal G}(\gamma)-p_{\mathcal G}(\gamma)\right)=p_{\mathcal G}(\gamma)+m\left(\gamma-\kappa\gamma^2/2\right),\nonumber
   \end{align*}
   for the finite section of the deformed green parabola between  $P_0^{[m]}$ and $Q_0^{[m]}:=\left(p^{**}_0(\kappa),q_0^{[m]'}\right)$, $p'_0(\kappa):=-1-3\kappa/8$ and $q_0^{[m]'}:=p'_0(\kappa)-m\left(1+\kappa/2\right)$;
    \item (iv) $p=p'_0(\kappa), q \leq q_0'^{[m]}$, for the vertical half-line separatrix $D^{[m]'}_0$ below point $Q_0^{[m]}$;
      \item (v) the combined parametric equations \eqref{pqquartic}, 
      \begin{align*}
   &p=p_{\mathcal G}(\gamma),\,\,\, \gamma\in[1+{2}/{\kappa},+\infty),\\ 
   &q
   = p_{\mathcal Q}(\gamma)+m\left( q_{\mathcal Q}(\gamma)-p_{\mathcal Q}(\gamma)\right)=p_{\mathcal Q}(\gamma)+m\left(\gamma-\kappa\gamma^2/2\right),\nonumber
   \end{align*}
for the infinite upper  branch of the deformed blue quartic running from  $Q_0^{[m]}$ to $\infty$.
 \end{itemize}
\end{proposition}}


\end{document}